\documentclass[11pt,letterpaper]{article}
\usepackage{jheppub}
\usepackage{comment,xcolor}
\usepackage{amsmath}
\usepackage{mathtools}
\usepackage{amssymb}
\usepackage{amsfonts}
\usepackage{amsthm}
\usepackage{mathrsfs}
\usepackage[all]{xy}
\usepackage{tensor}
\usepackage{footnote}
\usepackage{graphicx} 
\usepackage[justification=raggedright,singlelinecheck=false]{caption}
\usepackage{subcaption}
\usepackage{cleveref}
\usepackage{booktabs}
\usepackage{hyperref}

\newtheorem{innercustomgeneric}{\customgenericname}
\providecommand{\customgenericname}{}
\newcommand{\newcustomtheorem}[2]{%
  \newenvironment{#1}[1]
  {%
   \renewcommand\customgenericname{#2}%
   \renewcommand\theinnercustomgeneric{##1}%
   \innercustomgeneric
  }
  {\endinnercustomgeneric}
}

\newcustomtheorem{customthm}{Theorem}
\newcustomtheorem{customlemma}{Lemma}

\newtheorem{thm}{Theorem}
\newtheorem{lem}{Lemma}

\newtheorem{defn}{Definition}

\newcommand{\dd}{\mathrm{d}}
\newcommand{\beq}{\begin{equation}}
\newcommand{\eeq}{\end{equation}}

\numberwithin{equation}{section}
\newcommand{\ket}[1]{| #1 \rangle}
\newcommand{\bra}[1]{\langle #1 |}
\newcommand{\braket}[2]{\left\langle #1 \mid #2 \right\rangle}

\def\b{\overline{b}}

\def\be{\beta}

\def\ma{\mathrm{max}}
\def\mi{\mathrm{min}}
\def\eps{\varepsilon}

\def\O{\mathcal{O}}

\def\h{\mathcal{H}}
\def\S{\mathcal{S}_=}
\def\s{\mathcal{S}_\leq}
\newcommand*{\tr}{\mathrm{Tr}}

\def\rrho{\tilde{\rho}}

\newcommand{\JZ}[1]{\textcolor{black}{#1}}

\title{The refined quantum extremal surface prescription from the asymptotic equipartition property}

\author{Jinzhao Wang }
\affiliation{\small \it Institute for Theoretical Physics, ETH 8093 Z\"urich, Switzerland}
\emailAdd{jinzwang@phys.ethz.ch}

\abstract{Information theoretic ideas have provided numerous insights in the progress of fundamental physics, especially in our pursuit of quantum gravity. In particular, the holographic entanglement entropy is a very useful tool in studying AdS/CFT, and its efficacy is manifested in the recent black hole page curve calculation. On the other hand, the one-shot information theoretic entropies, such as the smooth min/max-entropies, are less discussed in AdS/CFT.  They are however more fundamental entropy measures from the quantum information perspective and should also play pivotal roles in holography. We combine the technical methods from both quantum information and quantum gravity to put this idea on firm grounds. In particular, we study the quantum extremal surface (QES) prescription that was recently revised to highlight the significance of one-shot entropies in characterizing the QES phase transition. Motivated by the asymptotic equipartition property (AEP), we derive the refined quantum extremal surface prescription for fixed-area states via a novel AEP replica trick, demonstrating the synergy between quantum information and quantum gravity. We further prove that, when restricted to pure bulk marginal states, such corrections do not occur for the higher R\'enyi entropies of a boundary subregion in fixed-area states, meaning they always have sharp QES transitions. Our path integral derivation suggests that the refinement applies beyond AdS/CFT, and we confirm it in a black hole toy model by showing that the Page curve, for a black hole in superposition of two radiation stages, receives a large correction that is consistent with the refined QES prescription.}

\begin{document}
\maketitle

\section{Introduction}\label{sec:intro}

The notion of entanglement (von Neumann) entropy is getting popular in our pursuit of understanding fundamental physics. As an abstract information-theoretic measure of quantum correlation, its significance to physics was not recognized until the landmark discovery of black hole entropy~\cite{bekenstein1972black,Bekenstein1973,bekenstein1974generalized,hawking1975particle} and the associated information loss paradox~\cite{hawking1976breakdown}. This was a paradigm shift. High energy physicists were then motivated to examine the previously ignored kinematic aspect of the quantum field theory more carefully than ever~\cite{unruh1976notes,bisognano1976duality,hooft1985quantum,bombelli1986quantum,srednicki1993entropy,susskind1998some,callan1994geometric,holzhey1994geometric}. They are also joined by the condensed matter physicists~\cite{calabrese2004entanglement,calabrese2005evolution,calabrese2006entanglement,hastings2006solving,kitaev2006topological,wolf2008area,li2008entanglement}, to whom the entanglement entropy can serve as a useful probe of the complicated wave function of a many-body system. The entanglement entropy soon acquired independent interests in both communities. Progress has been made not only in developing the techniques that compute the entropy, such as the replica trick method~\cite{mezard1987spin,calabrese2004entanglement,lewkowycz2013generalized}, but also in discovering surprising seminal results, such as the entropic c/F/a-theorems~\cite{casini2004finite,casini2015mutual,casini2017markov} in quantum field theory (QFT) and the Ryu-Takayanagi (RT) formula in AdS/CFT~\cite{ryu2006holographic,ryu2006aspects}. 

In particular, the RT formula has ever since occupied the center stage in the research of holographic duality~\cite{hubeny2007covariant,headrick2007holographic,casini2011towards,swingle2012entanglement,lewkowycz2013generalized,faulkner2013quantum,engelhardt2015quantum,dong2016gravity,harlow2017ryu,dong2018entropy,dong2019flat,akers2019holographic} (cf. reviews by~\cite{rangamani2017holographic,nishioka2018entanglement,headrick2019lectures}). It has far-reaching implications such as the subregion-subregion duality~\cite{dhw,harlow2017ryu,faulkner2017bulk,harlow2018tasi} and the corresponding quantum error correction picture~\cite{almheiri2015bulk}. The modern understanding of the RT formula is encapsulated in the quantum extremal surface (QES) prescription~\cite{engelhardt2015quantum}, which is the statement that the von Neumann entropy of a boundary marginal state $\rho_B$ supported on a boundary subregion $B$ can be computed in terms of the \emph{generalized entropy} evaluated on the bulk QES:
\begin{equation}\label{nettaaron}
    S(B)_\rho = \mathrm{min\,ext}_\gamma \left[\frac{A[\gamma]}{4G_N} + H(b)_\rho\right]
\end{equation}
where the optimization identifies the minimal extermal\footnote{Here extremal for a codimension-two surface means the mean curvature vector (or the ingoing and outgoing expansions) vanishes. Throughout we shall work in the time-symmetric setting in regard of the QES prescription, so the extremization practically reduces to a minimization.} codimension-two surface $\gamma$ homologous to $B$ to be the QES and $b$ is the corresponding \emph{entanglement wedge} (EW) of $\rho_B$, defined as the region enclosed by $\gamma\cup B$. The first term denotes the Bekenstein-Hawking entropy of the surface $\gamma$. Throughout, the area $A[\gamma]$ should be implicitly interpreted as the expectation value of the area operator $\langle \hat A\rangle$ to account for the quantum fluctuations in the geometry~\cite{engelhardt2015quantum,dong2018entropy}.\footnote{We should also broadly interpret the term ``area'' broadly as a general quasilocal geometric quantity evaluated on $\gamma$ to account for the higher curvature corrections~\cite{wald1993black,iyer1994some,dong2014holographic,camps2014generalized}.} The second term denotes the von Neumann entropy of the bulk matter in the EW.\footnote{\label{ft:symbols} We shall use the symbol $S$ for any holographic entropies defined at boundary, and the symbol  $H$ for any bulk entropy quantities.} This formula~\eqref{nettaaron} extends the generalized entropy defined for black holes to QES in the broader context of AdS/CFT.

In words, the QES prescription simply looks for the bulk surface that gives the smallest generalized entropy. Interesting physics is indicated by the QES formula when the bulk entropy or its variation becomes comparable to the leading area term. The power of \eqref{nettaaron} culminates in the recent Page curve calculations partially resolving the black hole information paradox~\cite{almheiri2019entropy,penington2020entanglement,almheiri2020page,almheiri2020replica,penington2019replica,almheiri2020entropy}. 

The widely used method to compute the entanglement entropy is called the \emph{replica trick}, which is originated from the spin glass theory~\cite{mezard1987spin}. It is then successfully applied to 2d CFTs leading to some landmark results~\cite{calabrese2004entanglement,calabrese2006entanglement}.  In holography, this is the essential tool in the Lewkowycz–Maldacena (LM) derivation of the RT formula~\cite{lewkowycz2013generalized}. The idea is to first compute the one parameter generalisation of the von Neumann entropy, called the R\'enyi entropies indexed by integers $n>1$, by evaluating the trace via a path integral $Z_nB,\rho$ over $n$ glued-up replicas of the original system. Then we need to analytically continue the resulting expression as an entire function of $n$ given only its values at integer $n$'s. 
\begin{equation}
   S_n(\rho)_B := \frac{1}{1-n}\log\tr\rho^n = \frac{1}{1-n}\log\frac{Z_n[B,\rho]}{Z_1[\rho]^n},\quad S(\rho)_B:= \lim_{n\rightarrow 1}  S_n(\rho)_B
\end{equation}
where $Z_1[\rho]^n$ is the normalization.

The replica trick is not always working because the very step of analytic continuation can be subtle (cf. discussions in e.g.~\cite{rangamani2017holographic,headrick2019lectures,d2021alternative}). First of all, the legitimacy of such procedure needs the uniqueness guarantee from the Carlson's theorem~\cite{boas2011entire}. The validity of its associated assumptions is often ad-hocly assumed in QFT and the analytic continuation is heuristically used as long as the end result is physically sensible. Furthermore, even if we consider a regularized field theory that is effectively finite-dimensional and the Carlson's theorem can be applied, the analytic continuation can still be difficult in practise when the integer R\'enyi entropies are not in analytic expressions. For example, it is not yet known how to extrapolate the von Neumann entropy for the two-intervals case in CFTs of free compact bosons from the values at integer R\'enyi entropies~\cite{calabrese2009entanglement,calabrese2011entanglement,de2015entanglement,ruggiero2018entanglement,d2021alternative}. 

An outside perspective could be useful here. When the entanglement entropy was getting attention in the physics community, it was already a well-established concept in quantum information theory. Through the lens of information theory, the entanglement entropy together with its classical counterpart, Shannon entropy, are only ever relevant in the \emph{asymptotic} regime~\cite{shannon1948mathematical,schumacher1995quantum,bennett1996concentrating,bennett1996purification,schumacher1997sending,holevo1998capacity}. That means when we try to characterize the ability to perform certain information-processing tasks, like data compression or channel coding, we are allowed to consider an infinite amount of resources, such as channels, and only ask about the rate at which the information can be transformed in the limit when the law of large number kicks in. In these scenarios, it turns out that the von Neumann entropy or Shannon entropy is the most relevant quantity to look at. However, people start to realize the limitation of the asymptotic analysis in real world applications with finite resources, and this calls upon the development of \emph{one-shot} information theory. It lives on the other end of the spectrum where one only considers a single use of the input resource, so in a way characterizing one-shot scenario concerns the most fundamental aspect of information theory. More entropy measures are developed when information theorists move onto this new territory~\cite{renner2004smooth,renner2008security,wang2012one,muller2013quantum,wilde2014strong,dupuis2014generalized,tomamichel2015quantum}. In particular, the smooth min/max-entropies are introduced for their operational significance in various one-shot tasks~\cite{renner2004smooth,renner2008security}, such as compression, state merging, randomness extraction and leftover hashing, etc~\cite{berta2009single,konig2009operational,tomamichel2011leftover,tomamichel2015quantum}. We shall also generally refer them as the \emph{one-shot entropies} among others~\cite{wang2012one,dupuis2014generalized}. The smooth min/max-entropies are particularly interesting for us in light of the asymptotic equipartition property (AEP)~\cite{tomamichel2009fully}. The AEP formalizes the intuition that the entanglement entropy is the asymptotic limit of the one-shot entropies. Formally, it claims that for some state $\rho_{B^m}=\rho^{\otimes m}_{B}$ and some $0<\varepsilon<1$, the (unconditional) smooth min/max-entropy asymptotically behave as
 \begin{equation}\label{AEP}
     \begin{aligned}
           \lim_{m\rightarrow\infty} \frac1m S^\varepsilon_{\mathrm{min/max}}(B^m)_\rho = S(B)_\rho\,.
     \end{aligned}
 \end{equation}

The main contribution of this work is to propose a novel \emph{AEP replica trick} that entails \emph{no analytical continuation} (see below). We show its efficacy in the context of AdS/CFT by deriving the recently revised QES prescription due to Akers and Penington (AP)~\cite{akers2021leading}, in analogy with how Lewkowycz–Maldacena derived the Ryu-Takayanagi forumla using the standard replica trick~\cite{lewkowycz2013generalized,faulkner2013quantum}. It is a nice demonstration of the synergy between quantum information and quantum gravity in terms of both ideas and techniques.\\
\begin{figure}
  \centering
\includegraphics[width=0.45\linewidth]{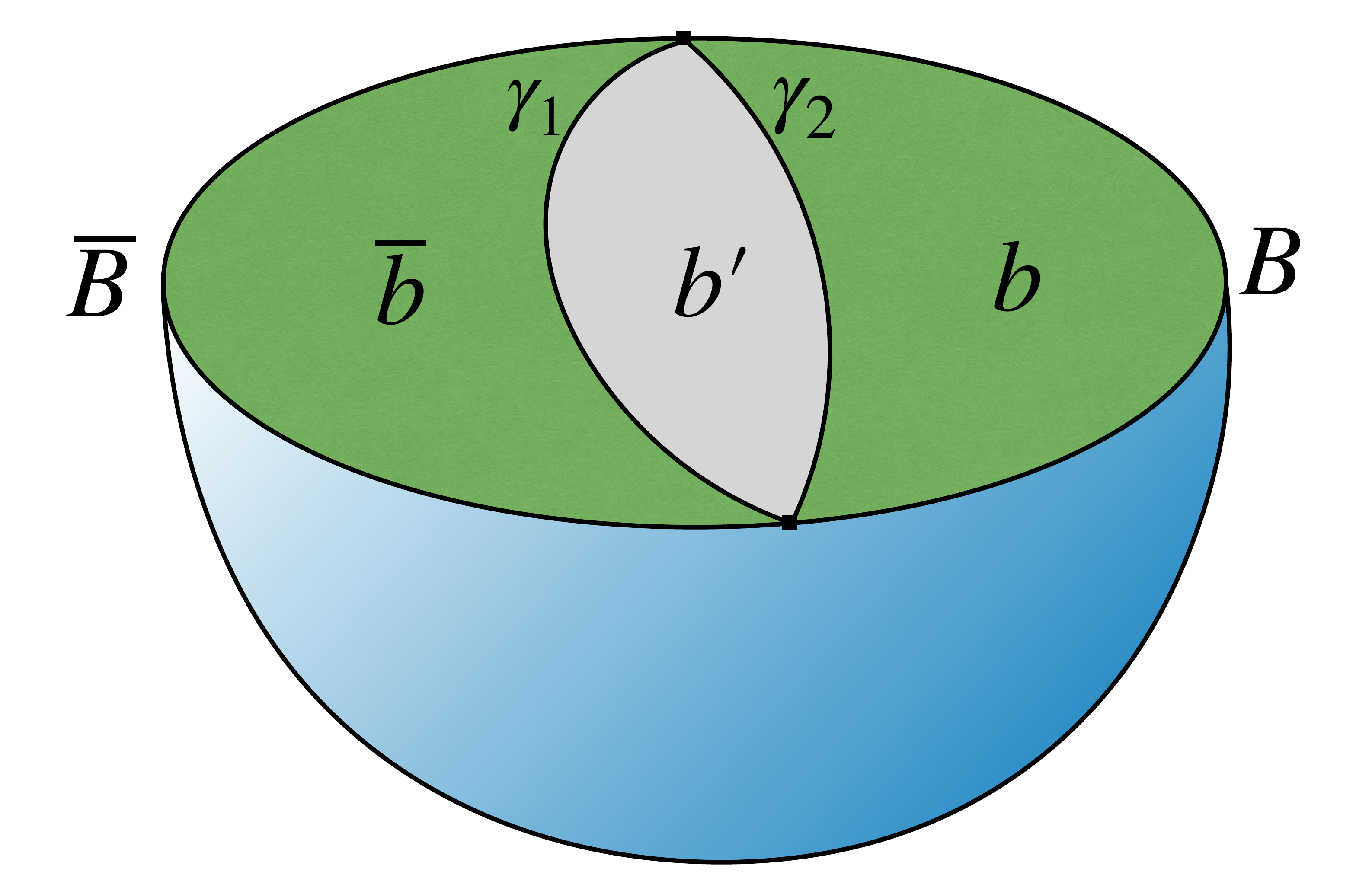}
  \caption{A simple setup for the QES phase transition. There are two QES candidates $\gamma_{1,2}$ homologous to the boundary $B$ with areas $A_{1,2}$. The bulk is divided into three subregions $b,b'$ and $\b$. The global bulk state $\rho_{bb'\b}$ is pure. }
  \label{fig:setting1}
\end{figure}

 AP pointed out that in regimes close to a entanglement wedge phase transition, the naive QES predication deviates from the actual von Neumann entropy calculated on the boundary at the leading $\O(G^{-1})$ order. A subleading deviation of order $\O(G^{-\frac12})$ is already observed in the Page curve calculation~\cite{penington2019replica}, and it is also analyzed in two other related works~\cite{marolf2020probing,dong2020enhanced}. However, only AP's work emphasizes the key roles played by the one-shot entropies. When we have multiple QES candidates to pick from, instead of comparing the bulk von Neumann entropies, AP argued that it is the bulk smooth min/max-entropies that decide which QES can be used for calculating the von Neumann entropy of the boundary subregion. In the particular case of two competing QES candidates ($\gamma_{1,2}$ with area $A_{1,2}$) for some CFT state $\rho$ with a geometric dual (cf. Figure~\ref{fig:setting1} for an illustration of the setup\footnote{AP used a different setup (cf. Figure 3 in~\cite{akers2021leading} and also~\cite{akers2019large}) where the boundary region $B$ consists of two disconnected pieces. At the center there is a dustball with large entropy, which can be purified by an auxiliary system $R$. As far as the QES phase transition is concerned, our setups are effectively the same once we include $R$ as part of $\overline{B}$. }.), we have the following refined QES prescription:
\begin{equation}\label{eq:refined_QES}
    S(B)_\rho \approx^\eps 
    \begin{cases}
        \frac{A_1}{4G_N} + H(bb')_\rho ,& H^\eps_\mathrm{max}(b'|b)_\rho \le \frac{A_2 - A_1}{4G_N} \quad\quad\quad\quad\quad\quad\quad\,(\mathrm{regime}\; 1)\\
        \text{(indefinite)} ,& H^\eps_\mathrm{min}(b'|b)_\rho \le \frac{A_2 - A_1}{4G_N} \le H^\eps_\mathrm{max}(b'|b)_\rho \,\,\,\,\,\,(\mathrm{regime}\; 2)\\
        \frac{A_2}{4G_N} + H(b)_\rho ,& H^\eps_\mathrm{min}(b'|b)_\rho \ge \frac{A_2 - A_1}{4G_N}\,\quad\quad\quad\quad\quad\quad\quad \,(\mathrm{regime}\; 3)
    \end{cases}
\end{equation}
where $b$ and $b'$ are bulk subregions partitioned by $\gamma_{1,2}$, \JZ{and $\approx^\eps$ means the equality holds up to $\O(\eps)$ corrections.} 

On the other hand, the naive QES prescription says
\begin{equation}\label{eq:naive_QES}
    S(B)_\rho = 
    \begin{cases}
        \frac{A_1}{4G_N} + H(bb')_\rho ,& H(b'|b)_\rho \le \frac{A_2 - A_1}{4G_N} \\
        \frac{A_2}{4G_N} + H(b)_\rho ,& H(b'|b)_\rho \ge \frac{A_2 - A_1}{4G_N}\,,
    \end{cases}
\end{equation}
which is just a rewriting of \eqref{nettaaron} in this particular scenario with the conditional von Neumann entropy $H(b'|b)_\rho := H(b')_\rho-H(bb')_\rho$. These entropy measures satisfy the following relation (for small $\eps$):
\begin{equation}\label{threeentropies}
    \quad H^\eps_\mathrm{min}(b'|b)_\rho\leq H(b'|b)_\rho\leq H^\eps_{\mathrm{max}}(b'|b)_\rho\,.
\end{equation}
Therefore, we see that \eqref{eq:refined_QES} is refines the naive QES prescription \eqref{eq:naive_QES} by restricting the regimes where the RT formula holds.

Using \eqref{threeentropies}, we see that the smooth min/max-entropies help to refine the QES prescription by restricting the validity regimes $1$ and 3 where the generalized entropy correctly measures $S(B)$, while also opening up an indefinite regime $2$ where $S(B)$ is generally not captured by the generalized entropy of some bulk surface. \JZ{This refinement is relevant whenever we are dealing with a bulk state with an entanglement spectrum that deviates largely from the flat spectrum, such that the min/max-entropies differ from the von Neumann entropy.} With a dustball example, AP showed that a leading order correction indeed occurs in regime 2 when the bulk dustball is in a mixture of a thermal state with some pure state. The details in regime 2 is not completely understood yet, and we shall not attempt to go beyond \eqref{eq:refined_QES} in this work. This refined prescription shows that the switching between the QES' is not simply a sharp first-order transition and the ususal RT formula is only applicable within the bounds defined by the min/max-entropies. \\

 From the quantum information point of view, this refinement is not surprising as the naive QES prescription is too simple to be correct. It is well known that the validity of the quantum Ryu-Takayanagi formula and also the QES prescription is equivalent to the achievability of the \emph{entanglement wedge reconstruction} (EWR) task~\cite{dhw,harlow2017ryu}. The puzzle is that there is only ever one copy of the system considered in EWR, whereas the von Neumann entropy alone usually cannot characterize one-shot scenarios\footnote{There is however one exception that we know of~\cite{boes2019neumann,wilming2020entropy}.} and it only becomes relevant in the asymptotic limit with many copies of the system. We propose that one way to make sense of it is via the AEP \eqref{AEP}. 

How do we see the AEP in all these QES discussions? It's natural to think that since the replica method involves multiple copies, a path-integral derivation of refined QES prescription should provide the answer. However, it cannot simply be the standard replica trick that's been used to derive the RT formula~\cite{lewkowycz2013generalized,faulkner2013quantum}, because there the $n$ copies are all contracted and $n$ is then analytically continued to one. Motivated by the AEP, here we propose an alternative replica method to compute the von Neumann entropy by introducing $n\times m$ replicas. The idea is to use the AEP as the \emph{defining formula} for the von Neumann entropy. Instead of taking the limit $n\rightarrow 1$ following the standard definition, we take the opposite limit $n,m\rightarrow \infty$ following the AEP for the smooth min-entropy:
 \begin{equation}\label{minAEP}
           S(B)_\rho := \lim_{m\rightarrow\infty}\lim_{n\rightarrow\infty}\max_{\tilde{\rho}_{B^m}\approx_\varepsilon\rho_B^{\otimes m}}  \frac{1}{m(1-n)}  \log\tr\left[\tilde{\rho}_{B^m}\right]^n \,.
 \end{equation}

The trace will be evaluated using the gravitational path integral with the replica trick for fixed-area states~\cite{dong2019flat,akers2019holographic}. An important advantage of the AEP replica trick over the standard one is that no analytical continuation is needed, because we can take the infinity limits as integer sequences. On the other hand, one drawback is that since the smoothing is presumably hard to implement in a gravitational calculation, one may only be able to evaluate the RHS with some feasible state $\tilde{\rho}$. In this case, one can only obtain lower bounds for the von Neumann entropy via \eqref{minAEP}. Nevertheless, together with the upper bounds obtained from the standard replica trick calculation, it turns out the bounds are enough to give the refined QES prescription, in particular, the refined conditions in \eqref{eq:refined_QES}. We believe \eqref{minAEP} is more intuitive as a formula for the von Neumann entropy because it manifests the asymptotic nature of the von Neumann entropy, and they could have applications in other problems concerning the holographic entanglement entropy. We shall name this approach the \emph{AEP replica trick}. 

As compared to the standard replica trick, the AEP replica trick is a more robust approach to compute entropies. We have mentioned that the analytical continuation step required in standard replica trick can be difficult when we have non-analytic expressions for the integer R\'enyi entropies. In holography, similar difficulties occur exactly when we consider the refinement of the QES prescription in phase transitions. To properly resolve the transition over the indefinite regime, one needs to add up the contributions from non-replica-symmetric saddles in the gravitational path integral calculation~\cite{marolf2020probing,dong2020enhanced,akers2021leading}. This can sometimes be done for simple states via the resolvent method~\cite{penington2019replica,akers2021leading}, which effectively extracts the entanglement spectrum from the integer R\'enyi entropies, but is hard to implement for general states. 

Of course, there is no cheat here when using the AEP replica trick in the sense that if one would like to look for the optimal state for the smooth entropy and thus compute the entanglement entropy with equality, then that's equivalent to resolve the entanglement spectrum. However, one often has to resort to bounds in more realistic physical scenarios without much symmetry to exploit. Then the integer R\'enyi entropies can hardly be extracted as analytic expressions, whereas the smooth min/max-entropies can nevertheless be estimated with mathematically rigorous bounds. Our derivation of the refined QES prescription exactly showcases this advantage, where the resolvent calculation cannot be implemented for general states. This very approach of getting to the von Neumann entropy via AEP has also been used in various other applications in quantum theory, from proving the strong subadditivity~\cite{beaudry2012intuitive} and its strengthened version~\cite{fawzi2015quantum} to entropic uncertainty relations~\cite{berta2010uncertainty}.\\

\JZ{Besides featuring a new technique, our path-integral derivation also suggests that the refinement should apply beyond the context of AdS/CFT, just as how the RT formula applies to broader gravitational scenarios. One particularly important application of the QES prescription is the calculation of the Page curve for a evaporating black hole, where an \emph{island formula} is discovered to compute properly the fine-grained entropy of the Hawking radiation~\cite{almheiri2020page,penington2019replica,almheiri2020replica,almheiri2020entropy}. The island formula is in fact a specific application of the QES formula, so we naturally expect the same revision we had for the QES prescription also extends to the island formula and the Page curve. In this work, we support this claim by demonstrating an example of a corrected Page curve in the toy model of~\cite{penington2019replica}.}\\

The paper is organized as follows: we start in section~\ref{sec:prelim} with some necessary backgrounds on the techniques used in our derivation. In section~\ref{sec:aepreplica} the standard replica trick for fixed-area states is reviewed and the AEP replica trick is introduced.  We apply the AEP replica trick to derive the refined QES prescription in section~\ref{sec:derivation} and extend it to the general multi-QES scenario in section~\ref{sec:multiqes}. On the other hand, we prove in section~\ref{sec:holorenyi} that for pure bulk marginal states the higher holographic R\'enyi-entropies, unlike the von Neumann entropy, do not have such leading order corrections. \JZ{In section~\ref{sec:pagecurve}, we show that the Page curve in a black hole toy model also obeys the refined QES prescription.} We finish with some comparisons with AP's argument and discussions on plausible future directions in section~\ref{sec:discussion}. The technical details that come into the \JZ{derivation and the calculations} are provided in Appendix~\ref{sec:proof} \JZ{and~\ref{sec:JTEOW}}, and we also sketch a complementary max-entropy replica method for obtaining upper bounds in Appendix~\ref{sec:maxchainrule}.

\section{Preliminaries}\label{sec:prelim}

Our derivation features techniques from both quantum information and quantum gravity, so here we introduce the essential tools for readers who are not familiar with both subjects. In particular, we briefly review the conditional min/max-entropies and their chain rules; and the gravitational path integral, the replica trick and fixed-area states. 

\subsection{The one-shot entropies}\label{sec:entropies}

We consider finite dimensional quantum states on some Hilbert space $\h$ described by density operators, and denote the space of normalized states as $\S(\h)$. Since the smoothing operation is important to our derivation, it's technically convienient to consider subnormalized states while performing smoothing, and we denote the space as $\s(\h)$. We shall define the conditional min/max-entropies using the sandwiched quantum R\'enyi divergences $\widetilde{D}_n$, which is an one parameter family that generalises the Umegaki relative entropy. For $\rho,\sigma\in\s(\h)$,
\begin{equation}
    \widetilde{D}_n(\rho||\sigma) = \frac{1}{n-1}\log\left[\tr \left(\sigma^{\frac{1-n}{2n}}\rho\sigma^{\frac{1-n}{2n}}\right)^n/\tr\rho^n\right] .
\end{equation}
We shall restrict to the domain $n\in[\frac12,\infty]$ where the data processing inequality holds for the sandwiched quantum R\'enyi divergences, and we assume the support of $\rho$ is contained in the support of $\sigma$ such that the divergence is bounded. We can recover the relative entropy by taking the limit, $\lim_{n\rightarrow  1} \widetilde{D}_n(\rho||\sigma) = S(\rho||\sigma)$. $\widetilde{D}_n$ is monotonic in $n$, the two ends of the R\'enyi spectrum mark the two important instances of the R\'enyi divergence, namely the log-fidelity $D_\frac12$ and the max-divergence $D_\infty$. The min/max-entropies of some quantum state $\rho_{AB}$ are then defined as 
\begin{equation}\label{minmaxentropies}
\begin{aligned}
    H_\ma (A|B)_\rho = \sup_{\sigma_B} - \widetilde{D}_\frac12 (\rho_{AB}||I_A\otimes\sigma_B) =& \max_{\sigma_B}\log F(\rho_{AB},I_A\otimes\sigma_B),\\
    H_\mi (A|B)_\rho = \sup_{\sigma_B} - \widetilde{D}_\infty (\rho_{AB}||I_A\otimes\sigma_B) =& \max_{\sigma_B}-\log ||\sigma_B^{-\frac12}\rho_{AB}\sigma_B^{-\frac12}||_\infty,\\
     =&\max_{\sigma_B}\max\{\lambda\in \mathbb{R}, \rho_{AB}\leq 2^{-\lambda}I_A\otimes\sigma_B\}.
\end{aligned}
\end{equation}
Similarly, one can define the general R\'enyi conditional entropies $H_n$ in the same way and
\begin{equation}
    H_\ma(A|B)_\rho=H_\frac12(A|B)_\rho,\,\,H_\mi(A|B)_\rho= H_\infty(A|B)_\rho,\,\, \lim_{n\rightarrow 1}H_n(A|B)_\rho = S(A|B)_\rho\,.
\end{equation}

Because of the minus sign in \eqref{minmaxentropies}, $n\mapsto H_n$ is anti-monotone as opposed to $\widetilde{D}_n$ so we have \eqref{threeentropies} as mentioned in the introduction,
\begin{equation}
 \quad H_\mathrm{min}(A|B)_\rho\leq S(A|B)_\rho\leq H_{\mathrm{max}}(A|B)_\rho\,.
\end{equation}
The smooth min/max-entropies are then defined as~\cite{renner2004smooth,tomamichel2015quantum}
\begin{equation}
    H^\eps_\ma(A|B)_\rho := \min_{\eta_{AB}\in\mathcal{B}^\eps(\rho_{AB})} H_\ma(A|B)_\eta,\quad H^\eps_\mi(A|B)_\rho := \max_{\eta_{AB}\in\mathcal{B}^\eps(\rho_{AB})} H_\mi(A|B)_\eta\,,
\end{equation}
and they satisfy a \emph{duality relation} for any pure state $\rho_{ABC}$:
\begin{equation}\label{eq:duality}
    H^\eps_\ma (A|B)_\rho = - H^\eps_\mi (A|C)_\rho\,.
\end{equation}

In the above definition, $\mathcal{B}^\eps(\rho_{AB})$ denotes the $\eps$-ball centered at $\rho_{AB}$ whose radius is measured by the purified distance~\cite{tomamichel2010duality}, $\mathcal{B}^\eps(\rho_{AB}):=\{\eta\in\s(\h_{AB}), P(\eta,\rho)\leq\eps\}$ and
\begin{equation}
    P(\eta,\rho):=\sqrt{1-F(\eta,\rho)}\,.
\end{equation}
where the (generalized) fidelity for subnormalized states are given by 
\begin{equation}
    F(\eta,\rho):= ||\sqrt{\eta}\sqrt{\rho}||_1+\sqrt{(1-\tr \eta)(1-\tr \rho)}\,.
\end{equation}
Note that here it's important to allow the smoothing to include subnormalized states in order to leave the smooth entropies unchanged when embedded into larger Hilbert spaces. This relaxation is important in proving the chain rules (cf. the Appendix~\ref{sec:proof}).

The purified distance, as opposed to the more common trace distance, has the advantage that we can always find extensions or purifications of the given density operator such that their purified distance remains the same (see Lemma~\ref{lem:lemma2})~\cite{tomamichel2010duality}. This is useful in establishing several chain rules for the smooth min/max-entropies~\cite{vitanov2013chain}. Unlike the chain rule for the von Neumann entropy, the chain rules for the smooth min/max-entropies only hold approximately as inequalities. For a density operator $\rho_{ABC}$, and $\varepsilon,\varepsilon',\varepsilon''\in [0,1)$ with $\varepsilon>2\varepsilon''+\varepsilon'$, we have 
\begin{equation}
    \begin{aligned}
    &H^\eps_\mi(AB|C)_\rho\geq H^{\eps''}_\mi(A|BC)_\rho+H_\mi^{\eps'}(B|C)_\rho-\delta ,\\
    &H^{\eps'}_\mi(AB|C)_\rho\leq H^{\eps''}_\ma(A|BC)_\rho+H_\mi^{\eps}(B|C)_\rho+3\delta
\end{aligned}
\end{equation}
where $\delta:=-\log(1-\sqrt{1-(\eps-2\eps''-\eps')^2})$. There are a couple of more chain rules but we will only need these two. As often in applications using the chain rules, we choose the smoothing parameters such that the $\delta$ is small compared to the entropy values. For our purposes, the $\delta$ term is of order $\O(\log\eps)\sim O(\log G_N)$ and is therefore subleading as compared to the entropy terms of order $\O(G^{-1}_N)$. We henceforth ignore these remainder terms.

The chain rules play the key role in our derivation as it compares the min-entropies evaluated on nested regions in terms of the conditional min/max-entropy. However, one issue is that the above chain rule holds when each entropy is evaluated on some distinct state within the $\eps$-ball around $\rho$, whereas we would like a chain rule relating min-entropies evaluated on the same state. Therefore, we will not use the chain rule exactly as given above, but rather a stronger version that actually serves as an intermediate step while proving it. The claim is that there exists a state $\rho'_{ABC}$, which is more than $\eps'+2\eps''$ distance away from $\rho_{ABC}$ and has the marginal state $\rho'_{BC}$ that maximizes $H^{\eps'}_\mi(B|C)_\rho$, such that
\begin{equation}\label{chainrule0}
    H_\mi(AB|C)_{\rho'}\geq H^{\eps''}_\mi(A|BC)_\rho+H_\mi^{\eps'}(B|C)_\rho=H^{\eps''}_\mi(A|BC)_\rho+H_\mi(B|C)_{\rho'}\,.
\end{equation}
 It then implies the first chain rule above. When applying \eqref{chainrule0} later, we shall also set $\eps''=0$ and let system $C$ be trivial. So we have a state $\rho'_{AB}$ and some $1>\eps>\eps'>0$, which is at $\eps$ distance away from $\rho_{AB}$ and has the marginal state $\rho'_{B}$ that maximizes $H^{\eps'}_\mi(B)_\rho$ i.e. $H_\mi(B)_{\rho'} = H^{\eps'}_\mi(B)_\rho$, such that:
 \begin{equation}\label{chainrule1}
    H_\mi(AB)_{\rho'}\geq H_\mi(A|B)_\rho+H_\mi^{\eps'}(B)_\rho=H_\mi(A|B)_\rho+H_\mi(B)_{\rho'}\,.
\end{equation}
 Now the chain rule is set in a useful form for our purposes: the min-entropies $H_\mi(AB), H_\mi(B)$ are evaluated on the same state $\rho'_{AB}$ and they can be compared using the conditional min-entropy evaluated on the original state. It will help us characterize the transition between regimes 2 and 3 in \eqref{eq:refined_QES}. 
 
We have a similar statement corresponding to the other chain rule, and it will help us characterize the transition between regimes 1 and 2 in \eqref{eq:refined_QES}. There exists a state $\rho''_{AB}$ and some $1>\eps>\eps''>0$, which is at $\eps$ distance away from $\rho_{AB}$ and maximizes $H^{\eps''}_\mi(AB)_\rho$ i.e. $H_\mi(AB)_{\rho''} =H^{\eps''}_\mi(AB)_\rho$ , such that
 \begin{equation}\label{chainrule2}
    H^{\eps''}_\mi(AB)_\rho=H_\mi(AB)_{\rho''}\leq H_\ma(A|B)_\rho+H_\mi(B)_{\rho''}\,.
\end{equation}

The proof of the chain rules \eqref{chainrule1} and \eqref{chainrule2} are given in the Appendix~\ref{sec:proof} and we shall see them in action in section~\ref{sec:derivation}. \\

 We shall finish this section by stating the general AEP theorem~\cite{tomamichel2015quantum} with some remarks. Given a quantum state $\rho_{A^mB^m}=\rho^{\otimes m}_{AB}$ and some $0<\varepsilon<1$, we have
 \begin{equation}\label{AEP2}
     \begin{aligned}
           \lim_{m\rightarrow\infty} \frac1m H^\varepsilon_{\mathrm{min}}(B^m|A^m)_\rho = H(B|A)_\rho\,,\\
            \lim_{m\rightarrow\infty} \frac1m H^\varepsilon_{\mathrm{max}}(B^m|A^m)_\rho = H(B|A)_\rho\,.
     \end{aligned}
 \end{equation}
 Note that the exact value of $\eps$ does not matter here as long as it is small but finite. Since we only want to evaluate the von Neumann entropy for a given boundary region rather than the conditional entropy for two boundary regions, the system $A$ in the AEP above can be set as trivial, which is the version we introduce in \eqref{AEP}. Nonetheless, we believe this formula can be useful whenever the boundary conditional entropy is relevant. It's worth mentioning that the original AEP theorem also entails taking the limit $\varepsilon\rightarrow 0$.  This was essentially needed for the converse part~\cite{tomamichel2009fully}:
 \begin{align}
     \lim_{\varepsilon\rightarrow 0}\lim_{m\rightarrow\infty} \frac1m H^\varepsilon_{\mathrm{min}}(B^m|A^m)_\rho \leq H(B|A)_\rho\,,\\
     \lim_{\varepsilon\rightarrow 0}\lim_{m\rightarrow\infty} \frac1m H^\varepsilon_{\mathrm{max}}(B^m|A^m)_\rho \geq H(B|A)_\rho\,.
 \end{align}
Note that the subtlety here is that for the smooth entropies, it is not true that for arbitrary $\varepsilon$ and  $\rho_{AB}$,  $H^\varepsilon_\mathrm{min}(B|A)_\rho\leq H(B|A)_\rho\leq H^\varepsilon_{\mathrm{max}}(B|A)_\rho$ holds. By continuity, we only expect this to hold for small enough $\varepsilon$. In~\cite{tomamichel2009fully}, a continuity bound was derived with a remainder term that only vanishes when one sends $\varepsilon\rightarrow 0$. Nevertheless, in the asymptotic limit, it turns out one can derive a better continuity bound with a vanishing remainder term in the limit $n\rightarrow\infty$ while fixing some finite $\varepsilon$ (cf. Corollary 6.3 in~\cite{tomamichel2015quantum} for details). We shall therefore use the more general AEP theorem without varying $\eps$. Also, the min-entropy AEP is sufficient for our derivation. We shall comment on the potential applications of the max-entropy AEP as well in section~\ref{sec:discussion}.

\subsection{The gravitational replica trick}\label{sec:replicatrick}

The fact that the gravitational entropy can be formally extracted from an Euclidean path integral was first proposed by Gibbons-Hartle-Hawking in deriving the black hole entropy~\cite{gibbons1977action,hartle1976path}. It is then generalized by Lewkowycz-Maldacena (LM) to AdS/CFT~\cite{lewkowycz2013generalized} using the \emph{replica trick}. It is a widely used method to compute the entanglement entropy in field theory, which is originated from the spin glass theory~\cite{mezard1987spin}. It was then successfully applied to 2d CFTs leading to some landmark results~\cite{calabrese2004entanglement,calabrese2006entanglement}. \\ 

Now we review some essential basics about path integrals in AdS/CFT in Euclidean signature.  One starts with an Euclidean path integral over a compact domain $\mathcal{B}$ that prepares the state of interest $\ket{\rho}$  supported on $\partial \mathcal{B}=B\overline{B}$, which in our case is some holographic CFT boundary state that has a time-symmetric bulk dual. Formally, the path integral defines for us a wave functional $\ket{\rho}$, which maps some field configuration $\phi$ at the boundary to $\braket{\phi}{\rho}$, 
\begin{equation}
    \braket{\phi}{\rho} = \int_\mathcal{B}^{\phi} \mathcal{D}\tilde{\phi}\, 2^{-I_{\text{boundary}}[\tilde{\phi}]}\,,
\end{equation}
where $I_{\text{boundary}}$ is the boundary Euclidean CFT action, and the sources that specifies $\ket{\rho}$ are implicitly captured in $\mathcal{B}$, which is schematically illustrated as the blue shell in Figure~\ref{fig:path1}.  The state is defined on its boundary, depicted as a circled partitioned into subregions $B$ and $\overline{B}$. Note that we shall use base-$2$ for all the exponentials and logarithms in the partition function.

Then the density operator supported on some subregion $B$ is formally given by tracing out the subregion $\overline{B}$. In the path integral, it corresponds to integrating over the bra and ket domains with the boundary conditions at $\overline{B}$ identified. 
\begin{equation}
   \rho_B=\int_{\overline{B}} \dd \phi_{\overline{B}}  \braket{\phi_{\overline{B}}}{\rho}\braket{\rho}{\phi_{\overline{B}}}\,,
\end{equation}
and we refer the $B$ region that is not integrated as a \emph{cut} on the path integral $Z_{\text{boundary}}$, which computes the trace of $\rho_B$ or equivalently the norm of $\ket{\rho}$,
\begin{equation}\label{eq:Z1}
   Z_{\text{boundary}}[\mathcal{B}^\dagger\mathcal{B}]:=\braket{\rho}{\rho}=\tr\,\rho_B=\int_{\mathcal{B}^\dagger\mathcal{B}} \mathcal{D}\tilde{\phi}\, 2^{-I_{\text{boundary}}[\tilde{\phi}]}\,
\end{equation}
where we denote the bra domain by $\mathcal{B}^\dagger$, and the notation $\mathcal{B}^\dagger\mathcal{B}$ stands for the combined path integral domain glued at $\overline{B}$.

We can insert observables $O[\tilde\phi]$ at $B$, then the path integral computes its expectation value,
\begin{equation}
    \tr O\rho_B =  Z_{\text{boundary}}[\mathcal{B}^\dagger\mathcal{B},O]:= \int_{\mathcal{B}^\dagger\mathcal{B}} \mathcal{D}\tilde{\phi}\,O[\tilde\phi] 2^{-I_{\text{boundary}}[\tilde{\phi}]}\,.
\end{equation}\\

The AdS/CFT correspondence allows us to translate the partition function on the boundary CFT to the bulk gravitational theory, 
\begin{equation}
    Z_{\text{boundary}}[\mathcal{B}^\dagger\mathcal{B},O] = Z_{\text{bulk}}[\mathcal{B}^\dagger\mathcal{B},O]
\end{equation}
The bulk partition function is defined using the bulk path integral that involves both the geometry and the matter field supported on it,
\begin{equation}
    Z_{\text{bulk}}[\mathcal{B}^\dagger\mathcal{B},O]:= \int_{\mathcal{M}} \mathcal{D}\tilde g \mathcal{D} \tilde \psi \,  2^{-I_{\text{bulk}}[\tilde g,\tilde \psi]} O[\tilde\psi]
\end{equation}
where $\mathcal{M}$ is the bulk asymptotically hyperbolic manifold with one extra dimension and satisfies the boundary condition $\mathcal{B}^\dagger\mathcal{B}$, $I_{\text{bulk}}[\tilde g,\tilde \psi]$ is the action of the dual bulk theory with the bulk metric $\tilde g$ and  bulk fields $\tilde \psi$ as variables, and $O[\tilde\psi]$ represents a functional of bulk fields that satisfies the constraint known as the extrapolate dictionary imposed by the boundary operator~$O$.

We work in the large $N$ and strong coupling limit of the boundary CFT, for which the bulk dual is described by a weakly coupled string theory that reduces to an effective
supergravity theory. In the semiclassical limit, we expect the bulk gravity to behave like the Einstein-Hilbert gravity in the leading order plus some higher derivative corrections. We leave out the corrections for our discussion here but they can all be handled systematically~\cite{wald1993black,iyer1994some,dong2014holographic,camps2014generalized}.  Then we can use the saddle-point approximation in the bulk to simplify the gravitational path integral .
\begin{equation}
    Z_\mathrm{bulk}[\mathcal{B}^\dagger\mathcal{B},O]\approx\sum_g \bar Z[g]Z^{\text{matter}}[g,\mathcal{B}^\dagger\mathcal{B},O] = \sum_g 2^{-I_{\text{grav}}[g]}Z^{\text{matter}}[g,\mathcal{B}^\dagger\mathcal{B},O]
\end{equation}
where $\bar Z[g]:=2^{-I_{\text{grav}}[g]}$ denotes the gravitational partition function for a classical saddle-point geometry $g$ that satisfies the boundary conditions, and $Z^{(\text{matter})}[\mathcal{B}^\dagger\mathcal{B},O]$ is the field theory partition on the background $g$, which is the bulk counterpart to \eqref{eq:Z1}.

In Figure~\ref{fig:path3}, the Euclidean bulk is depicted for simplicity as a disk. It prepares the geometry and the bulk state on the time-symmetric hypersurface $M$, which are determined by the dominant saddle to the bulk action subject to the boundary conditions.  The essence of the RT formula is that one can compute the entanglement entropy of the boundary subregion $B$ using the area of the minimal surface homologous to $B$ on this slice $M$.\\
\begin{figure}
\centering
\begin{subfigure}{.45\textwidth}
  \centering
  \includegraphics[width=.8\linewidth]{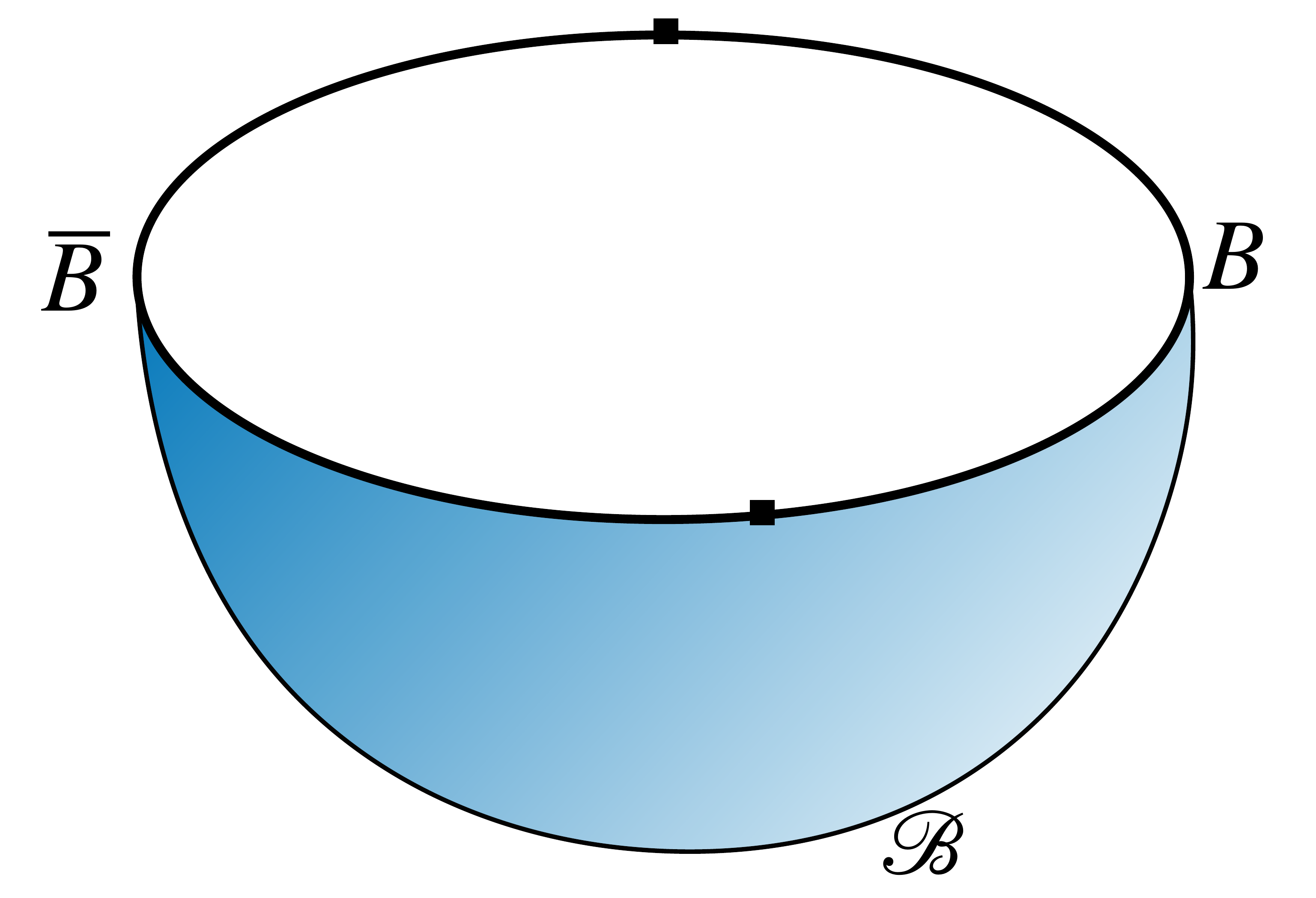}
  \caption{\centering}\label{fig:path1}
\end{subfigure}%
\begin{subfigure}{.45\textwidth}
  \centering
  \includegraphics[width=.7\linewidth]{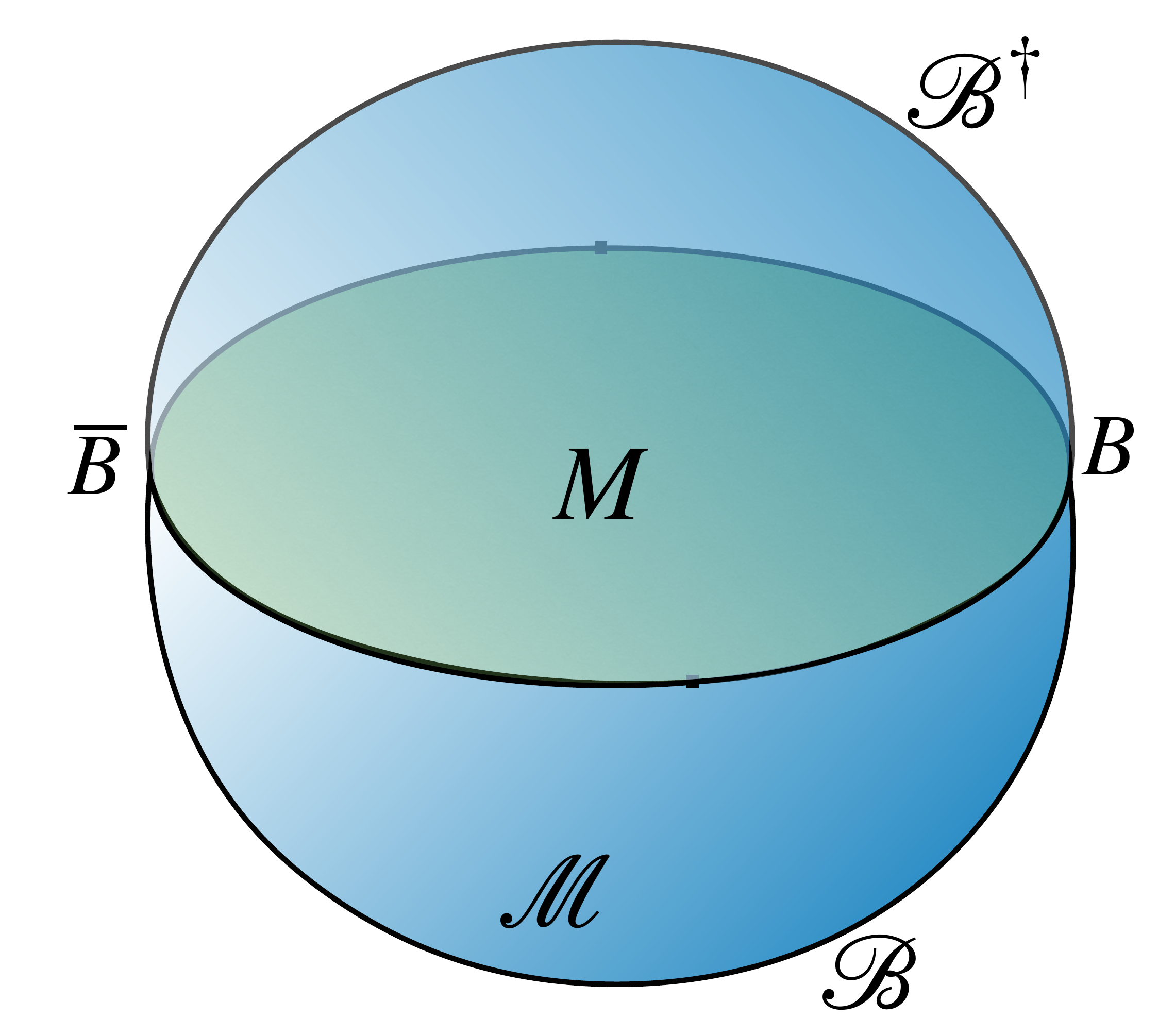}
   \caption{\centering}\label{fig:path2}
\end{subfigure}
\begin{subfigure}{.45\textwidth}
  \centering
  \includegraphics[width=.7\linewidth]{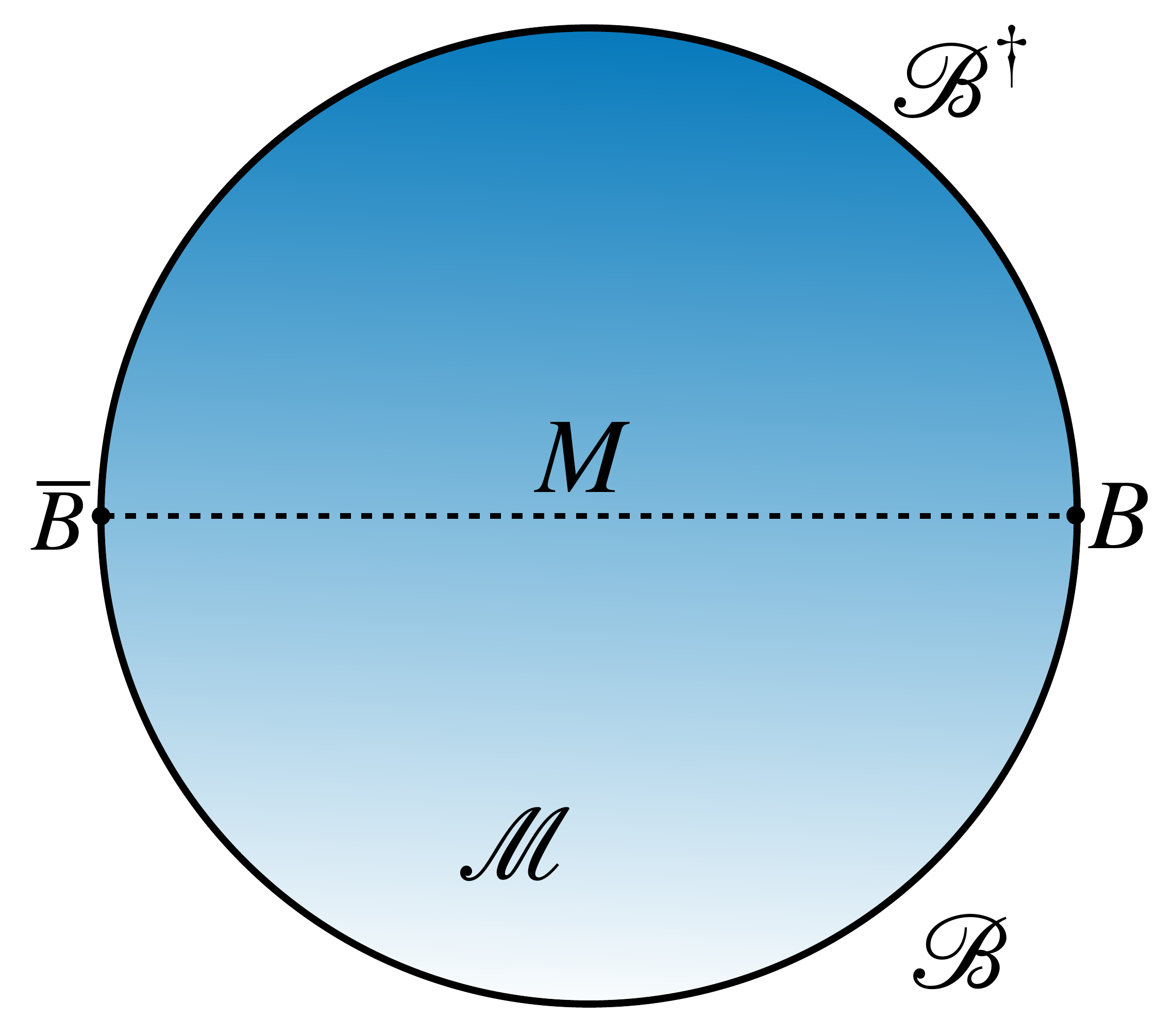}
   \caption{\centering}\label{fig:path3}
\end{subfigure}%
\begin{subfigure}{.45\textwidth}
  \centering
  \includegraphics[width=.7\linewidth]{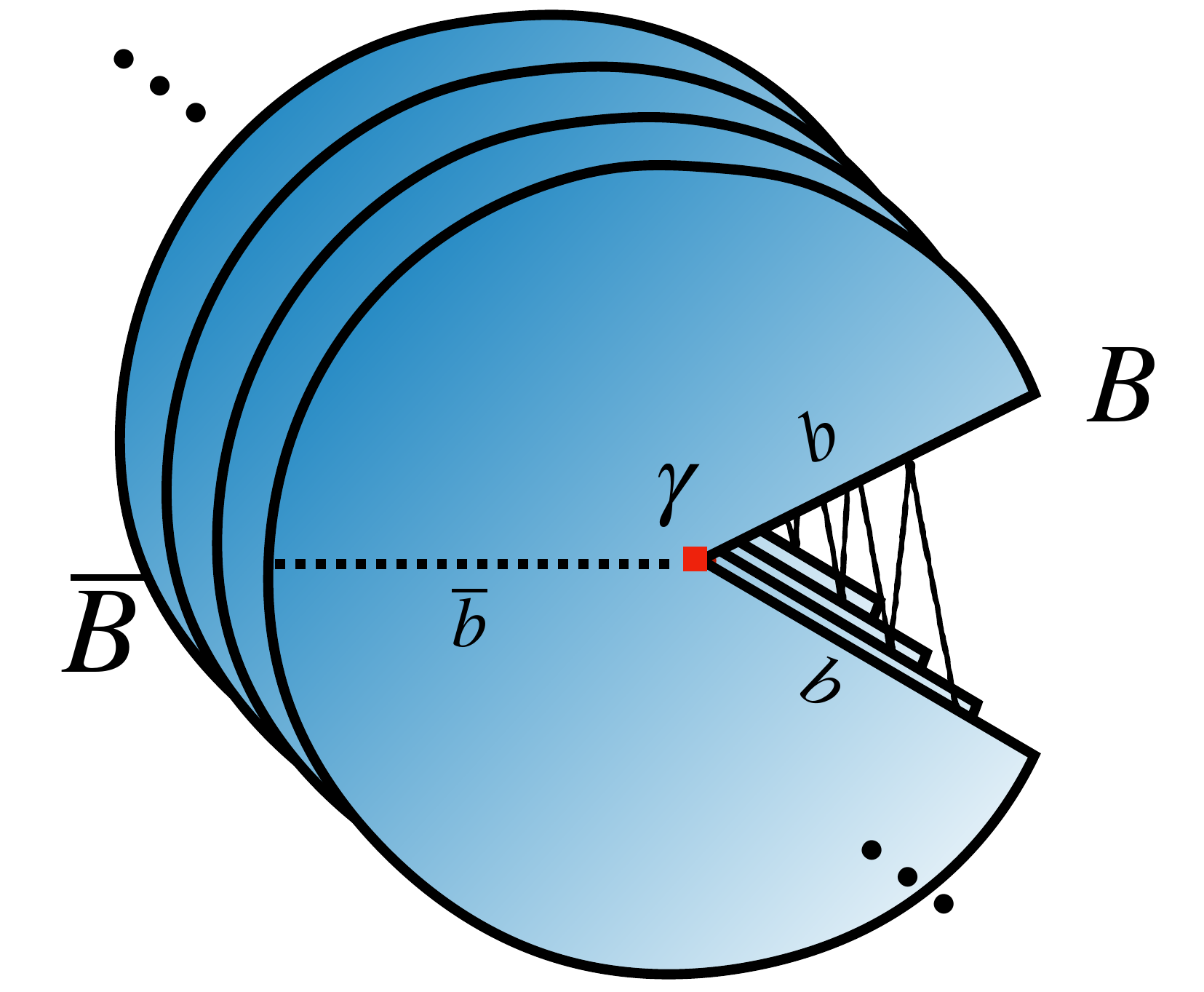}
   \caption{\centering}\label{fig:path4}
\end{subfigure}
\caption{\textbf{Path Integral and Replica Trick}.The figure (a) depicts the Euclidean path integral that prepares a ket state $\ket{\rho}$ supported on $\partial\mathcal{B}=B\overline{B}$. Note that the figure is only schematic and the boundary domain can take any appropriate topology depending on the boundary state one would like to prepare. Figure (b) shows the gravitational path integral that computes $\tr\rho_B=Z_\mathrm{boundary}[\mathcal{B}^\dagger\mathcal{B}]=Z_\mathrm{bulk}[\mathcal{B}^\dagger\mathcal{B}]$ in both the boundary and the bulk. The boundary is a path integral over both the ket domain $\mathcal{B}$ and the bra domain $\mathcal{B}^\dagger$, and the bulk path integral is over the manifold $\mathcal{M}$ filled in with the boundary conditions on the $\mathcal{B}^\dagger\mathcal{B}$. $M$ is the time-symmetric slice where the bra and ket are glued over.Again, the domain $\mathcal{B}^\dagger\mathcal{B}$ depicted here is schematic. One could have for instance a torus that is cut at $M$.  Figure (c) is a cartoon representation of $Z_\mathrm{bulk}[\mathcal{B}^\dagger\mathcal{B}]$ with one dimension suppressed and $M$ is depicted as the dotted line. Note that $\overline{B}$ and $B$ should still be connected albeit it looks disconnected on the figure. Figure (d) illustrates the calculation of $Z_n$ via the replica trick, where $n$-replicas of the path integral in (c) is glued up (shown as the zigzag lines) along the cuts at $b$ and the red surface $\gamma$ is the fixed-point of the $\mathbb{Z}_n$ symmetry.}
\label{fig:pathintegral}
\end{figure}

Let us discuss the replica trick calculation of the entanglement entropy $S(\rho)_B$ in holography due to LM~\cite{lewkowycz2013generalized}, which in turn shows the validity of the RT formula. A general characteristic of the replica trick is that one can compute the entanglement entropy $S(B)_\rho$ without knowing the explicit density matrix $\rho_B$. The replica trick relies on two general facts about entropies. Firstly, the von Neumann entropy of a quantum system equals the $n \to 1$ limit of the R\'enyi entropy of order $n>1$. Secondly, the R\'enyi entropy of order $n$, for integers $n \geq 2$,\footnote{Once the R\'enyi entropies of order $n$ are known for positive integers $n \in \mathbb{N}$, the values for non-integer orders are obtained by analytic continuation.} corresponds to the expectation value of an observable~$U_{\tau|B^n}$ on $n$ copies of the system, the \emph{replicas}. This expectation value can be evaluated by path integrals without an explicit description of the quantum state.

Starting in the boundary field theory, we are instructed to make $n$ copies of the manifold and glue them up cyclically along the cuts at subregion $B$, and then evaluate the partition function over the $n$-replicas~\cite{calabrese2004entanglement}. Let's use the shorthand $\mathcal{B}$ for the boundary domain $\mathcal{B}^\dagger\mathcal{B}$ and denote the cyclic-shift operator as $U_{\tau|B^n}$. We have
the boundary partition function $Z_\mathrm{boundary}[\mathcal{B}^{\times n},U_{\tau|B^n}]$, where the boundary conditions $[\mathcal{B}^{\times n},U_{\tau|B^n}]$ describe a stack of $n$ replicas of $\mathcal{B}$ that is cyclically glued at the cuts $B$. The replica trick then computes the R\'enyi entropies for integers $n\ge 2$ using
\begin{equation}\label{eq:renyi}
    S_n(B)_\rho := \frac{1}{1-n}\log\tr \rho_B^n= \frac{1}{1-n}\log\tr \rho^{\otimes n}U_{\tau|B^n} = \frac{1}{1-n}\log\frac{Z_\mathrm{boundary}[\mathcal{B}^{\times n},U_{\tau|B^n}]}{Z_\mathrm{boundary}[\mathcal{B}]^n}
\end{equation}
where the normalization $Z_\mathrm{boundary}[\mathcal{B}]=Z_\mathrm{boundary}[\mathcal{B},1]$ is the $n=1$ case of $Z_\mathrm{boundary}[\mathcal{B}^{\times n},U_{\tau|B^n}]$.

Then the entanglement entropy is given by
\begin{equation}\label{eq:vn}
    S(B)_\rho=\lim_{n\rightarrow 1}S_n(B)_\rho = -\partial_n (\log Z_\mathrm{boundary}[\mathcal{B}^{\times n},U_{\tau|B^n}]-n\log Z_\mathrm{boundary}[\mathcal{B}])\big |_{n=1}\,.
\end{equation}

Now we replace the boundary partition functions with the bulk partition functions denoted as $Z_n$ and $Z_1$. In the saddle-point approximation, we have
\begin{equation}\label{eq:flm}
    Z_n:= Z_\mathrm{bulk}[\mathcal{B}^{\times n},U_{\tau|B^n}] = \sum_{g} \bar Z[g]Z^\mathrm{matter}[g,\mathcal{B}^{\times n},U_{\tau|B^n}]
\end{equation}
where the sum is over classical geometries that are compatible with the boundary conditions $[\mathcal{B}^{\times n},U_{\tau|B^n}]$, and the bulk matter partition function evaluates the replica trick partition function for each background geometry.

Since the boundary condition has a $\mathbb{Z}_n$ symmetry for identical $\mathcal{B}$'s, LM assumes that this replica symmetry is carried over to the dominant dual bulk geometry as well.  We shall call it the \emph{replica symmetry} assumption.  Suppose there is only a single dominant geometry that is replica symmetric, and we denote it as $g_n$. 

With the replica symmetry, we can figuratively depict the bulk (Figure~\ref{fig:path4}) as cyclically gluing up $n$ identical replicas of the original geometry along some bulk cuts over the dotted line that extends from the boundary cuts $B$ into the bulk. Since we also know that the boundary field configurations are summed over at $\overline{B}$, this demands that in the bulk we should also glue up the geometry in each replica itself over the dotted line that extends from $\overline{B}$. Somewhere in the middle of the dotted line, there should be a turnover codimension-two surface $\gamma$, that is actually a fixed-point of the $\mathbb{Z}_n$ symmetry. Then the subregion bounded by $\gamma$ and $B$ is our bulk region $b$. This is also known as the \emph{entanglement wedge} (EW) of $B$. We also denote its complement as $\b$. 

Altogether these imply that the gravitational action evaluated on the replicated bulk geometry $g_n$ can be evaluated in the saddle point approximation using a quotient geometry $\hat g := g_n/\mathbb{Z}_n$,
\begin{equation}
    \bar Z[g_n]\approx 2^{-I[g_n]} = 2^{-nI[\hat g]}
\end{equation}
where the equality is due to the locality of the gravitational action. Since the quotient gives an opening angle of $2\pi/n$ at the fixed point $\gamma$, the action $I[\hat g]$ is different from $I[g]$ due to the conical singularity at $\gamma$. One can try to reproduce this effect by adding a cosmic brane action at $\gamma$ to the GR action. $\hat g$ can be then treated as a solution to an effective action with an additional $I_{\text{brane}}=T_n\int_\gamma \dd V$ with tension $T_n=\frac{n-1}{4G_N n}$. We still cannot directly compute $I[\hat g]$, but nonetheless the equation of motion demands the vanishing of first-order variation of the on-shell action with respect to $n$,
\begin{equation}\label{eq:effectiveaction}
   \partial_n I[\hat g]\big |_{n=1} =\partial_n\left( T_n\int_\gamma \dd V)\right)\bigg |_{n=1} = \frac{A[\gamma]}{4G_N}\,.
\end{equation}
Let us leave the matter partition function aside for a moment, and consider only the gravitational partition function $\bar Z[g_n]$ in~\eqref{eq:flm}. Then we have, according to \eqref{eq:vn},
\begin{equation}
    S(\rho)_B =-\partial_n(-n(I[\hat g]-I[g]))|_{n=1}=(I[\hat g]-I[g])|_{n=1}+n\partial_n I[\hat g]\big |_{n=1} = \frac{A[\gamma]}{4G_N}\,.
\end{equation}
 As the tension vanishes in the limit, the original geometry is restored, $I[\hat g]|_{n=1}=I[g]$. $\gamma$ is the bulk minimal surface homologous to $B$, and its minimality follows from the least action principle applied to the on-shell action with respect to varying its location. Therefore, $\gamma$ is the RT surface. This completes our summary of the LM derivation of the RT formula. 
 
Now we add back the matter partition function as in~\eqref{eq:flm}. Picking the dominant saddle geometry $g_n$ gives
\begin{equation}
    Z_n\approx \bar Z_n[g_n] Z^{\text{matter}}_{n}[g_n,\mathcal{B}^{\times n},U_{\tau|B^n}]
\end{equation}
Plugging $Z_n$ into~\eqref{eq:renyi} and then \eqref{eq:vn} gives a sum of two terms, in which the first term is $A[\gamma]/4G_N$ and the matter partition function yields the bulk matter entropy on the gluing region (entanglement wedge) $b$ in the bulk, denoted as $H(b)_\rho$. Overall, we have
\begin{equation}
    S(\rho)_B=\frac{A[\gamma]}{4G_N}+H(b)_\rho\ .
\end{equation}
This is the \emph{quantum RT formula}\footnote{The quantum RT formula is later superseded by the QES prescription, but we sometimes use this term in cases where there is only one QES candidate to consider.} proposed by Faulkner-Lewkowycz-Maldacena (FLM)~\cite{faulkner2013quantum}. which generalizes the Bekenstein's generalized entropy to holography. FLM only showed that this formula is correct up to one-loop order. A more careful variational analysis~\cite{engelhardt2015quantum,dong2018entropy} shows that the QES prescription is valid for all orders in $G_N$ in perturbation theory, and this finally established the QES prescription \eqref{nettaaron}.

\subsection{Replica trick with fixed-area states}\label{sec:fixedarea}
We are also interested in R\'enyi entropies, especially the min/max-entropies, as we will approach the entanglement entropy differently via the AEP. However, an obstacle is that the R\'enyi entropies do not have simple bulk duals as the von Neumann entropy does, but instead the so-called modular R\'enyi entropies satisfy a RT-like area law~\cite{dong2016gravity}.\footnote{On a technical level, this is because one needs to argue about the variations of the on-shell action to obtain the area law as in \eqref{eq:effectiveaction}. The modular R\'enyi entropy is defined with a derivative $\partial_n$, which is lacking in the standard R\'enyi entropy. } The holographic area law they follow is more complicated as it involves considering a cosmic brane with tension that backreacts and changes the original geometry. The brane location also varies from the RT surface and depends on $n$. More importantly, the modular R\'enyi entropies do not match with the standard R\'enyi entropies except at the $n\rightarrow 1$ limit. 

In order to calculate the min-entropy holographically to implement the AEP \eqref{minAEP}, we shall restrict to a special family of states, the \emph{fixed-area states}~\cite{dong2019flat,akers2019holographic}, which literally are the eigenstates for the area operator in the RT formula.\footnote{The ``area'' here should be interpreted as the generalized geometric quantity defined on $\gamma$ obtained from including all the higher curvature corrections~\cite{wald1993black,iyer1994some,dong2014holographic,camps2014generalized}} It helps facilitate the path-integral calculations for the R\'enyi entropies, by forcing them to be evaluated on the same bulk RT surface.  Thanks to the area-fixing, we can now work with the standard R\'enyi entropies instead of the modular R\'enyi entropies, and we don't have to explicitly account for the backreaction due to the cosmic brane tension. The same strategy is also taken by AP and we shall comment on going beyond the fixed-area states by the end of Section~\ref{sec:derivation}. 

Suppose we look for the RT surface on the bulk geometry $g$ following \eqref{eq:naive_QES} and find the QES surface $\gamma$ for some boundary subregion $B$. We choose to fix its area by restricting the metric integral to those with a fixed value $A[\gamma]=A^*$. It turns the $A[\gamma]$, interpreted as the expectation value of the area operator, to some fixed c-number. On the reduced boundary state $\rho_B$, it formally corresponds to a projection onto the $\gamma$-area eigenstates. The fixed-area bulk partition function for the expectation value of some boundary observable $O$ on $B$ can be written as
\begin{multline}
    Z[A^*,\mathcal{B},O]:=\tr O\Pi_{A[\gamma]=A^*}\rho_B\Pi_{A[\gamma]=A^*} \\
    = \int_\mathcal{M} \mathcal{D}g\mathcal{D}\psi\,\dd \mu\,2^{-I[g,\psi]-i\mu(A[\gamma]-A^*)} O[g,\psi] \approx 2^{-I[g^*]}Z^\mathrm{matter}[g^*,\mathcal{B},O]
\end{multline}
where the bulk integrals over manifold $\mathcal{M}$ are subject to the boundary condition $\mathcal{B}$. We've imposed the area-fixing at the QES via a Lagrange multiplier. One can restore the standard partition function we saw in the LM derivation by integrating over $A^*$. Using the saddle-point approximation to estimate the GPI and suppose for simplicity there is only one saddle geometry $g^*$ that has the fixed area of $\gamma$, we reduce the metric integral to the on-shell action $2^{-I[g^*]}$.

Because of the area-fixing, the classical solution $g^*$ to this action contains conical singularities at $\gamma$, corresponding to the tension $\mu^*$ that solves the equation of motion for $\mu$. This tension gives rise to a conical deficit, and we denote the opening angle as $\phi$. We note that the exact value of $\phi$ is irrelevant here as its on-shell value is whatever that fixes the area to be $A^*$, and $\phi$ will cancel out in the calculation of the R\'enyi entropies. The on-shell action reads
\begin{equation}\label{eq:fixedarea1}
    I[g^*] = I[g] + \frac{(\phi-2\pi)A^*}{8\pi G_N}
\end{equation}
where $I[g]$ is the on-shell gravitational action without the Lagrange multiplier. As compared to the LM derivation reviewed in the previous section, here we have an extra conical term due to area-fixing.

Thanks to the area-fixing, the saddle point geometry of the $n$-fold replica is extremely simple: the $n$-fold replica manifold can be constructed from gluing $n$ copies of the original manifold along the cut associated with $\gamma$. The conical opening angle is simply $n\phi$. As compared to the LM derivation, in this case we know the $n$-replica on-shell action explicitly. 
\begin{equation}\label{eq:fixedarea2}
    Z_{n}[A^*]:= Z[\mathcal{B}^{\times n},U_{\tau|B^n}]  = 2^{-I[g^*_n]}Z^\mathrm{matter}[\mathcal{B}^{\times n},U_{\tau|B^n}] =2^{nI[g] +  \frac{(n\phi-2\pi )A^*}{8\pi G_N}}\tr(\rho_b^{n})\,.
\end{equation}
From \eqref{eq:fixedarea1} and \eqref{eq:fixedarea2}, we can immediately infer that the R\'enyi entropies are given by
\begin{equation}
    S_n(\rho)_B = \frac{1}{1-n}\log \frac{Z_{n}[A^*]}{Z_{1}[A^*]^n} = \frac{A^*}{4G_N}+H_n(\rho)_b
\end{equation}
where $\phi,I[g]$ cancel out, and indeed all $S_n$ share the same area term.\\

Let us demonstrate how the standard replica trick can be applied to the scenario of two completing QES candidates $\gamma_{1,2}$ with fixed-areas $A_{1,2}$ (cf. Figure~\ref{fig:twoqes}).  We shall see what can be extracted out of it as well as its shortcomings. 

Following the same steps above, we find the saddle metric $g_n^*$ by gluing up $n$ copies of $g^*$. The dominant saddle-point solution $g_n^*$ is everywhere the same as $g^*$ at each replica except at $\gamma_{1,2}$, where the opening angles are $n\phi_{1,2}$. Here, an important distinction from the single-QES case is that we have an extra bulk subregion $b'$. The three bulk regions are glued together with different ordering: the $\b$ region is contracted at each replica corresponding to no permutation among the replicas $\mathbf{1}$ as before, and the $b$ region is contracted cyclically corresponding to a cyclic permutation among the replicas $\tau$, whereas the middle region $b'$ can be contracted with any permutation $\pi$. Therefore, as compared to \eqref{eq:fixedarea2}, here we need to account for the conical excesses at $\gamma_{1,2}$ due to the permutation $\pi$. Given some $\pi$, the on-shell action then reads\footnote{As you will see, we will not make explicit use of these nontrivial permutations except for $\pi=\mathbf{1}$ or $\tau$. You make refer to AP~\cite{akers2021leading} section 4.2 for a detailed account of these permutations $\pi$ and their relevance in the resolvent calculation. }
\begin{equation}
    I[g^*_n(\pi)] = nI[g] +  \frac{(n\phi_1-2\pi |\pi|)A_1}{8\pi G_N}+ \frac{(n\phi_2-2\pi |\pi^{-1}\circ\tau|)A_2}{8\pi G_N}\,
\end{equation}
where $|\pi|$ denotes the number of cycles in a permutation $\pi$. Considering all such permutations goes beyond the replica symmetry assumption, which only considers the replica-symmetric saddles $\pi=\mathbf{1}$ or $\tau$. Thus,
\begin{equation} \label{mainZn}
    \begin{aligned}
    \frac{Z_{n}[A_{1,2}]}{Z_{1}[A_{1,2}]^n}=& \sum_{\pi\in S_n} 2^{-I[g^*_n(\pi)]+nI[g^*]}\tr(\rho^{\otimes n}_{bb'} U_{\tau|b}\otimes U_{\pi|b'})\\
    =& \sum_{\pi\in S_n} 2^{(|\pi|-n)\frac{A_1}{4G_N} + (|\pi^{-1}\circ\tau|-n)\frac{A_2}{4G_N}}\tr(\rho^{\otimes n}_{bb'} U_{\tau|b}\otimes U_{\pi|b'})\,.
    \end{aligned}
\end{equation}
where the opening angles $\phi_{1,2}$ are canceled out by the normalizations.

\begin{figure}
\centering
\begin{subfigure}{.5\textwidth}
  \centering
  \includegraphics[width=.75\linewidth]{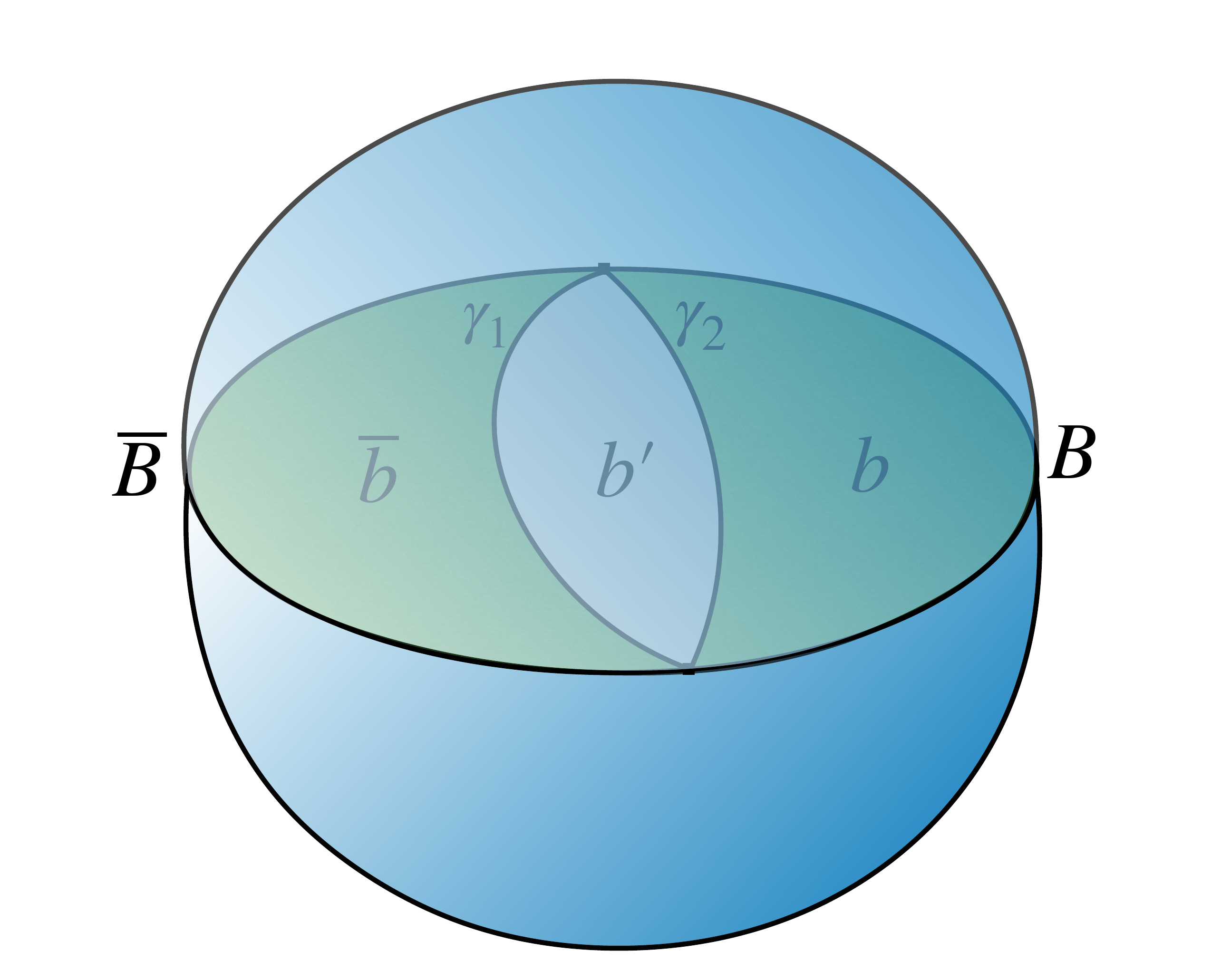}
\end{subfigure}%
\begin{subfigure}{.5\textwidth}
  \centering
  \includegraphics[width=.8\linewidth]{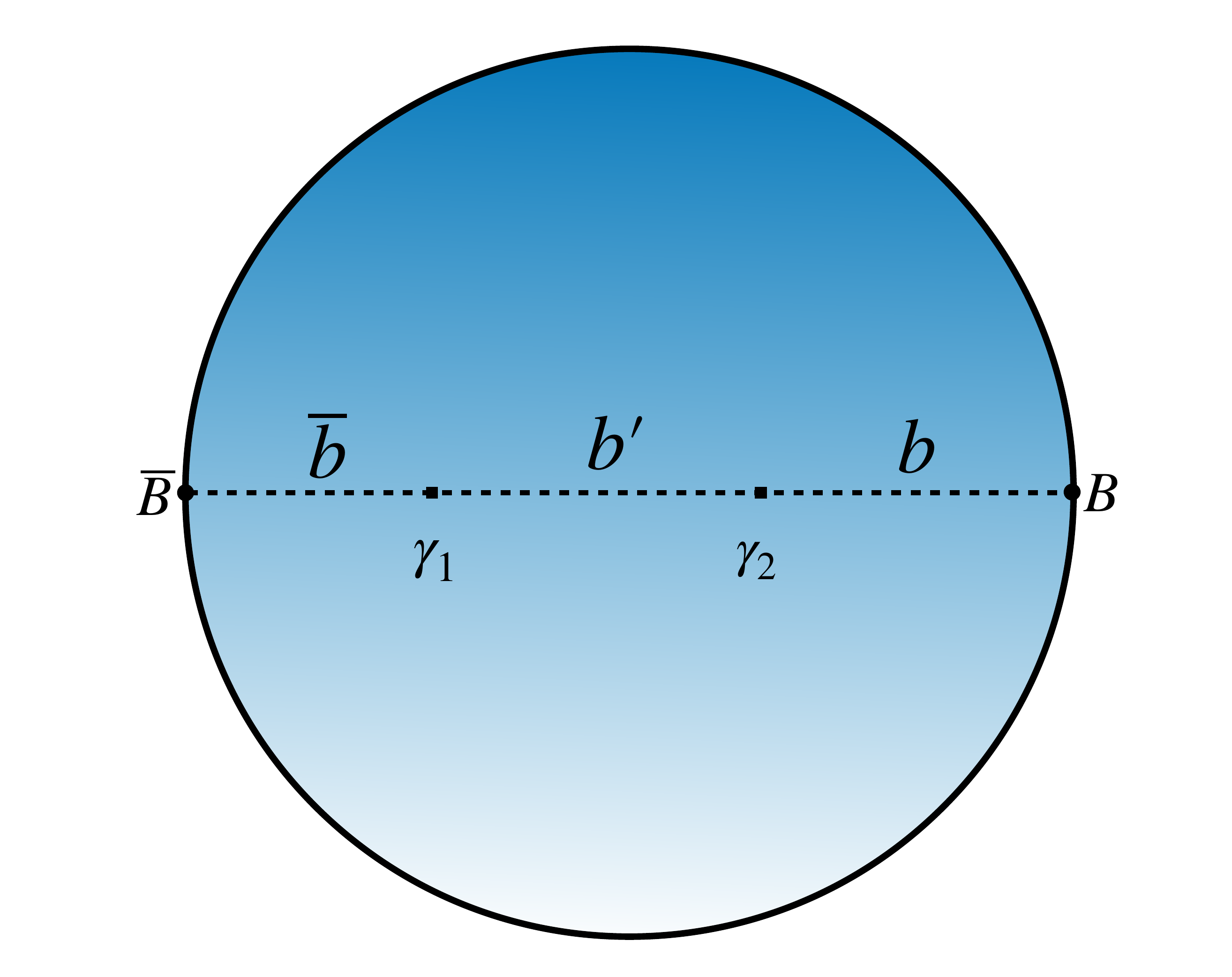}
\end{subfigure}
\caption{{\bf The partition function $Z_1[A_{1,2}]$. } The bulk slice $M$ is divided by $\gamma_{1,2}$ into three subregions $b,b',\b$. }
\label{fig:twoqes}
\end{figure}

Now the standard procedure tells us to analytically continue ($\lim n\rightarrow 1$) to the von Neumann entropy. However, this is tricky due to all these non-trivial permutations. Instead of evaluating \eqref{mainZn} $\lim n\rightarrow 1$, we can nevertheless obtain simple upper bounds for $S(B)$. It follows from  \eqref{mainZn} that ignoring the contributions from all the permutations besides the replica-symmetric ones ( $\pi=\mathbf{1}$ or $\tau$) yields a underestimated partition function and thus increases the value of $S(B)$ obtained via the analytic continuation. Thus,
\begin{equation} 
    \frac{Z_{n}[A_{1,2}]}{Z_{1}[A_{1,2}]^n} \geq 2^{\left(1-n\right)\frac{A_1}{4G_N}}\tr(\rho^{\otimes n}_{bb'} U_{\tau|bb'}),\quad\frac{Z_{n}[A_{1,2}]}{Z_{1}[A_{1,2}]^n} \geq 2^{\left(1-n\right)\frac{A_2}{4G_N}}\tr(\rho^{\otimes n}_{b} U_{\tau|b})\,.
\end{equation} 
Therefore, the standard replica trick only gives us the upper bounds if we hold the replica symmetry assumption,
\begin{equation}\label{upperbound}
    S(B)_\rho \leq  \mi \{\frac{A_1}{4G_N}+H(bb')_\rho\ , \frac{A_2}{4G_N}+H(b)_\rho\}\,.
\end{equation}\\

To put things in context and orient ourselves into the following sections, we note that much of what we cover in Section~\ref{sec:entropies}, such as the chain rules, is going to be applied to the bulk states, whereas the AEP will be applied to both bulk and boundary states. For our purposes, the field theory is assumed to be regularized rendering a finite-dimensional Hilbert space, so we can proceed working with density operators and the tools introduced in Section~\ref{sec:entropies}. Nevertheless, it is known that the one-shot entropies also generalize well to infinite dimensions~\cite{furrer2011min} and von Neumann algebras~\cite{berta2016smooth}. It would be interesting to extend the analysis rigorously to algebraic QFT, where one perhaps needs to work with min/max-divergences instead. 

We also emphasize that the reasons we choose to work with the fixed-area states henceforth are both conceptual and technical. Conceptually, we already know that the area fluctuations will induce subleading corrections to the RT formula near the transition~\cite{dong2020enhanced,marolf2020probing}, so we better turn off such corrections to manifest the corrections due to nontrivial bulk entanglement in a clean way. Technically, the fixed-area states allows us to simply glue up replica manifolds when dealing with integer R\'enyi entropies without worrying about the backreaction.

\section{The AEP replica trick}\label{sec:aepreplica}

The standard replica trick doesn't make it obvious at all why the min/max-entropies should be of any relevance in the refined QES prescription \eqref{eq:refined_QES}. On the other hand, the AEP replica trick naturally incorporates these one-shot quantities. 

Let us introduce additional $m$ replicas for each copy of $\rho$ we already have, yielding a partition function $Z_{n,m}$ defined for $n\times m$ replicas.\footnote{Such $n\times m$ replicas are also considered previously in the literature~\cite{dutta2019canonical,dong2020effective,engelhardt2021free}, where different limits of $n,m$ are taken for different purposes. On the level of replica tricks, the key distinction of our AEP replica trick is that we take the limits of both $n,m$ to infinity as integer sequences, which doesn't involve analytic continuation. } We shall refer them as \emph{$m$-families of $n$-replicas}. For the $m$-fold tensor product state, the dominant contribution to the disconnected $m$-replica boundaries should be a disconnected geometry. Therefore, the partition function is
\begin{equation}\label{trivialproduct}
    \frac{Z_{n,m}}{Z_1^{nm}} = \tr \left(\left[\rho^{\otimes m}\right]^{\otimes n} U_{\tau|B_1}\otimes\cdots\otimes U_{\tau|B_m}\right) = \left[\tr \rho^{\otimes n} U_{\tau|B}\right]^m = \left(\frac{Z_{n}}{Z_1^{n}}\right)^m.
\end{equation}

Instead, the formula \eqref{minAEP} tells us to look for the optimal state $\rho_*$ in the $\eps$-ball centered at $\rho^{\otimes m}_B$ for some arbitrary $0<\eps<1$, such that $\tr\rho_*^n$ is maximized. Now an obstacle is that the optimal state $\rho_*$ may not admit a convenient gravitational description, or for that matter, have a classical geometric dual at all in the first place. Therefore, here we pick particular feasible states and aim to obtain lower bounds for $S(B)$. As we shall see in the next section, we can carefully construct these states such that the lower bounds can precisely match with the upper bounds \eqref{upperbound} in regimes where the RT formula holds. 

More specifically, the chosen feasible state $\rho'$ shares the \emph{same} geometry $g$ as $\rho^{\otimes m}$ but only differs in its bulk state $\rho'_{b^mb'^m\b^m}$.  Since the bulk logical state is isometrically embedded in the code subspace of the boundary CFT, $\rho'$ is within the $\eps$-ball as long as the bulk state $\rho'_{b^mb'^m\b^m}$ is $\eps$-close to $\rho_{bb'\b}^{\otimes m}$. 

More precisely, we assume that there exists a set of boundary conditions for the $m$ families of replicas such that the gravitational path integral is dominated by bulk  a disconnected geometry among the $m$ families of replicas. We think this is a rather weak assumption, because the desirable states of our choice would be close to the product states corresponding to identical and independent boundary conditions. This means that any wormhole geometries that connect the $m$ families of replicas will give a negligible contribution, such that the entanglement is only manifested in the bulk state supported on the disconnected geometries among the $m$ families of replicas. We further assume that we can slightly perturb the bulk state to some nearby states of our choice with desirable properties (cf. Section~\ref{sec:outline}), by appropriately tuning the boundary conditions for the field configurations.

The partition function then reads

\begin{equation}\label{partitionfunction}
\begin{aligned}
    \frac{Z_{n,m}[A_{1,2}]}{Z_1[A_{1,2}]^{nm}} = \sum_{\{\pi^i\}^m_{i=1} \in S_n}& 2^{\left(\sum_{i=1}^m |\pi^i|-n\right)\frac{A_1}{4G_N} + \left(\sum_{i=1}^m |(\pi^i)^{-1}\circ \tau|-n\right)\frac{A_2}{4G_N}}\\
    &\cdot\tr(\rho'^{\otimes n}_{b^mb'^m}\cdot U_{\tau|b_1}\otimes\cdots\otimes U_{\tau|b_m}\otimes U_{\pi^1|b'_1}\otimes\cdots\otimes U_{\pi^m|b'_m})\,.
\end{aligned}
\end{equation}

\begin{figure}[t] 
  \centering
\includegraphics[width=0.9\linewidth]{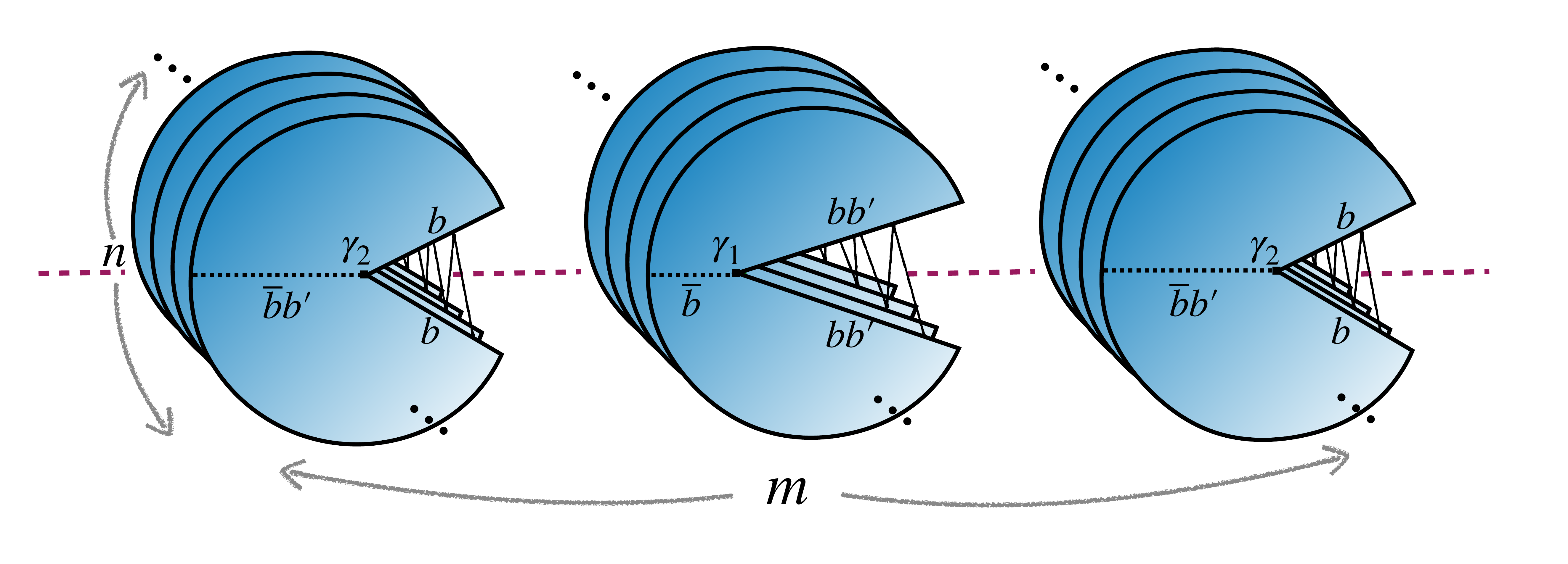}
  \caption{{\bf The AEP replica trick.} The ``Pac-Mans'' depict a generic term in the expansion of the $n\times m$-replica path integral that compute $Z_{n,m}[A_{1,2}]$. The vertically aligned $n$-replicas are contracted cyclically either along $bb'$ or $b$, as indicated by the zigzag lines. The disconnected $m$ families of replicas supports an entangled bulk state that is $\eps$-close to $\rho^{\otimes m}$, and the entanglement is indicated as the purple dashed line. }
  \label{fig:pacman}
\end{figure}

To reduce the expression further, we also use the fact that the dominant saddles for large $n$ are the ones that preserve the replica symmetry. They are $\pi^{i}=\mathbf{1}$ or $\pi^{i}=\tau$ for each of the $m$ replicas. This is not the case for the entanglement entropy at $n\rightarrow 1 $. For the min-entropy, however, we expect the replica symmetry assumption is legit and the effect of replica symmetry breaking should be insignificant. We later prove in Section~\ref{sec:holorenyi} that for the higher holographic R\'enyi-entropies of fixed-area states, the replica-symmetric saddles are dominant, and ignoring the other saddles only gives an over-estimation error of at most two bits. Since our proof is limited to pure bulk marginal states $\rho_{bb'}$, we shall keep the replica symmetry assumption for the min-entropy in our derivation.\footnote{Nonetheless, sometimes it is enough to consider pure bulk marginal states $\rho_{bb'}$. An example is the black hole model we study in Section~\ref{sec:pagecurve}.}

Now we have $2^m$ such terms instead of only two. Suppose $\{\pi^{i}\}_{i\in\Gamma}=\mathbf{1}$ for some selection $\Gamma\subset [m]:=\{1,\ldots,m\}$ of $k:=|\Gamma|$ replicas, then the corresponding term schematically reads
\begin{equation}\label{terms}
\begin{aligned}
    2^{(m-k)(1-n)\frac{A_1}{4G_N} + k(1-n)\frac{A_2}{4G_N}}\tr\left(\rho^{\otimes n}_{b^mb'^{m-k}}\cdot U_{\tau|b_1}\otimes\cdots\otimes U_{\tau|b_m}\otimes U_{\tau|b'_1}\otimes\cdots\otimes U_{\tau|b'_{m-k}}\right)
\end{aligned}
\end{equation}
where $k$ of the $b'$ systems are traced out, leaving the rest $m-k$ of the $b'$ systems cyclically contracting among the $n$-replicas. The path integral is illustrated in Figure~\ref{fig:pacman}. There are $\binom{n}{k}$ such terms with the same exponential area prefactor. Note that we do not distinguish them in $\eqref{terms}$ for different $\Gamma$'s in order to make the notations simpler, but we should bear in mind that $b'^{m-k}$ represents one particular collection $\{b'_i\}_{i\in [m]\setminus\Gamma}$ of the $b'$ systems. Besdies its cardinality $k$, the detail of exactly which systems $\Gamma$ contains is irrelevant for our derivation later.\footnote{Moreover, we shall actually consider states $\rho_{b^mb'^m}$ that are permutation-invariant, so the second line in \eqref{terms} is exactly identical for any $\Gamma$ of cardinality $k$. } We will see that the two end cases $k=0,m$ are particularly relevant, because they give the generalized entropy for the two QES candidates.

The key advantage of the AEP replica trick is that we can approach the limit as an integer sequence without assuming the analyticity in $n$, in contrast to the standard replica trick. Then we can apply \eqref{minAEP} to obtain a lower bound for the entanglement entropy,
\begin{equation}
    S(B) \geq \lim_{m\rightarrow\infty}\lim_{n\rightarrow\infty} \frac{1}{m(1-n)} \log\frac{Z_{n,m}[A_{1,2}]}{Z_1[A_{1,2}]^{nm}}.
\end{equation}
This is not yet complete because we'd like to link the right-hand side back to the bulk generalized entropy. It turns out that by carefully choosing the feasible state $\rho'_{b^mb'^m\b^m}$, we can also apply the AEP on the bulk level and eventually obtain the generalized entropy. This is how the AEP replica trick formally works. In summary, it relies on making clever choices on feasible bulk states $\rho'_{b^mb'^m\b^m}$ such that useful bounds for $S(B)$ can be obtained.

\section{Deriving the refined QES prescription}\label{sec:derivation}

We shall now demonstrate how this approach leads us to the refined QES prescription \eqref{eq:refined_QES}. Here we still consider the case where there are two competing QES saddles, and later extend the derivation to the general multi-QES case in section~\ref{sec:multiqes}.  For convenience, we assume without loss of generality that the global state on $\rho_{B\overline{B}}$ is pure,\footnote{Suppose however we are given a mixed CFT state on $B\overline{B}$, it can always be purified by some reference system $R$ and the global purity can be restored by including $R$ to be part of $\overline{B}$.} and so is the bulk state $\rho_{bb'\b}$. The refined QES prescription distinguishes from the naive one whenever there is nontrivial entanglement spectrum among the relevant bulk regions in $\rho_{bb'\b}$, encapsulated by an $\O(G^{-1})$ gap between the conditional min/max-entropies. 

\subsection{A sketch of the derivation strategy}\label{sec:outline}
\begin{figure}
\centering
\begin{subfigure}{.5\textwidth}
  \centering
  \includegraphics[width=1.\linewidth]{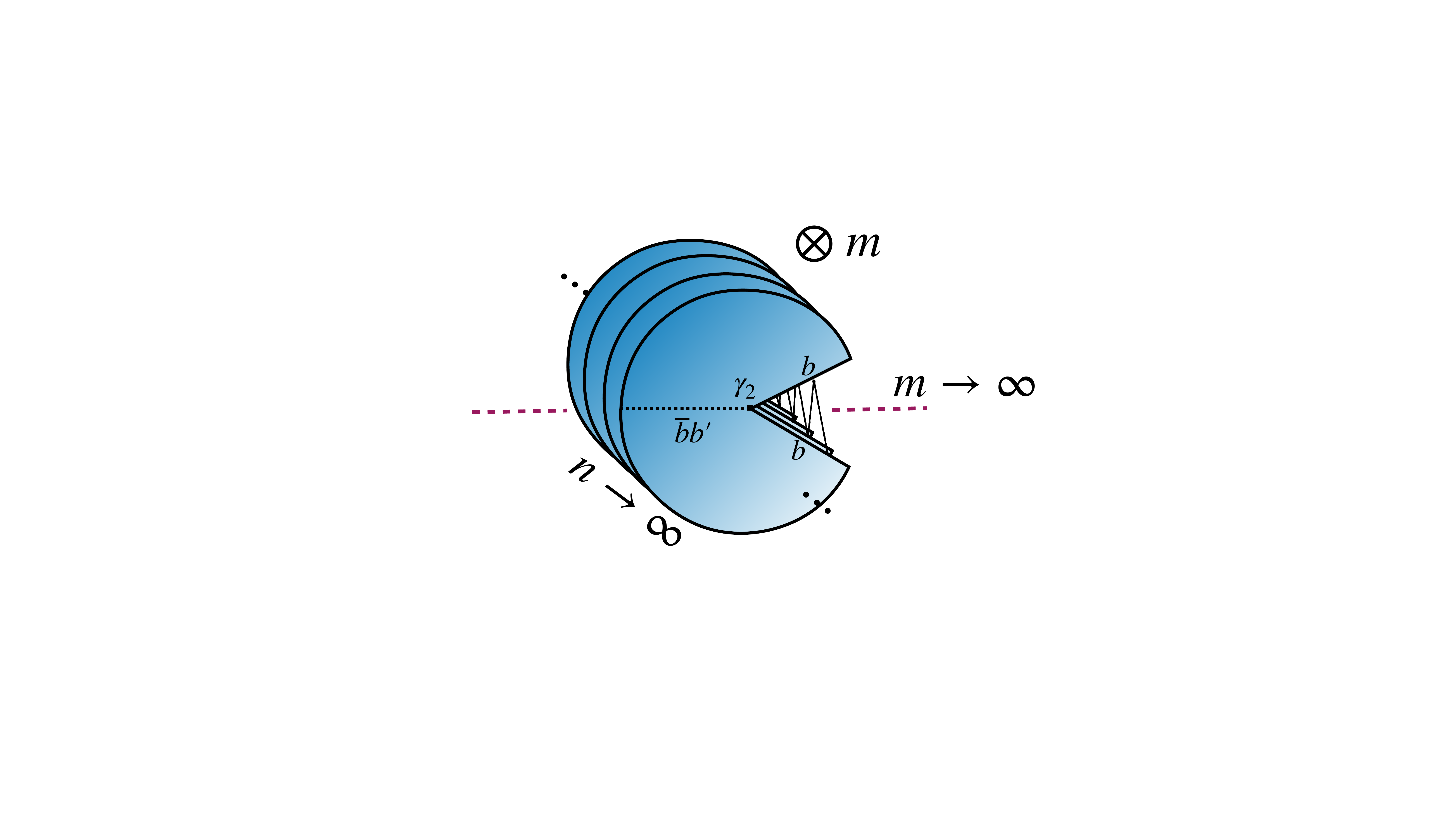}
  \caption{\centering QES=$\gamma_2$}
\end{subfigure}%
\begin{subfigure}{.5\textwidth}
  \centering
  \includegraphics[width=1.\linewidth]{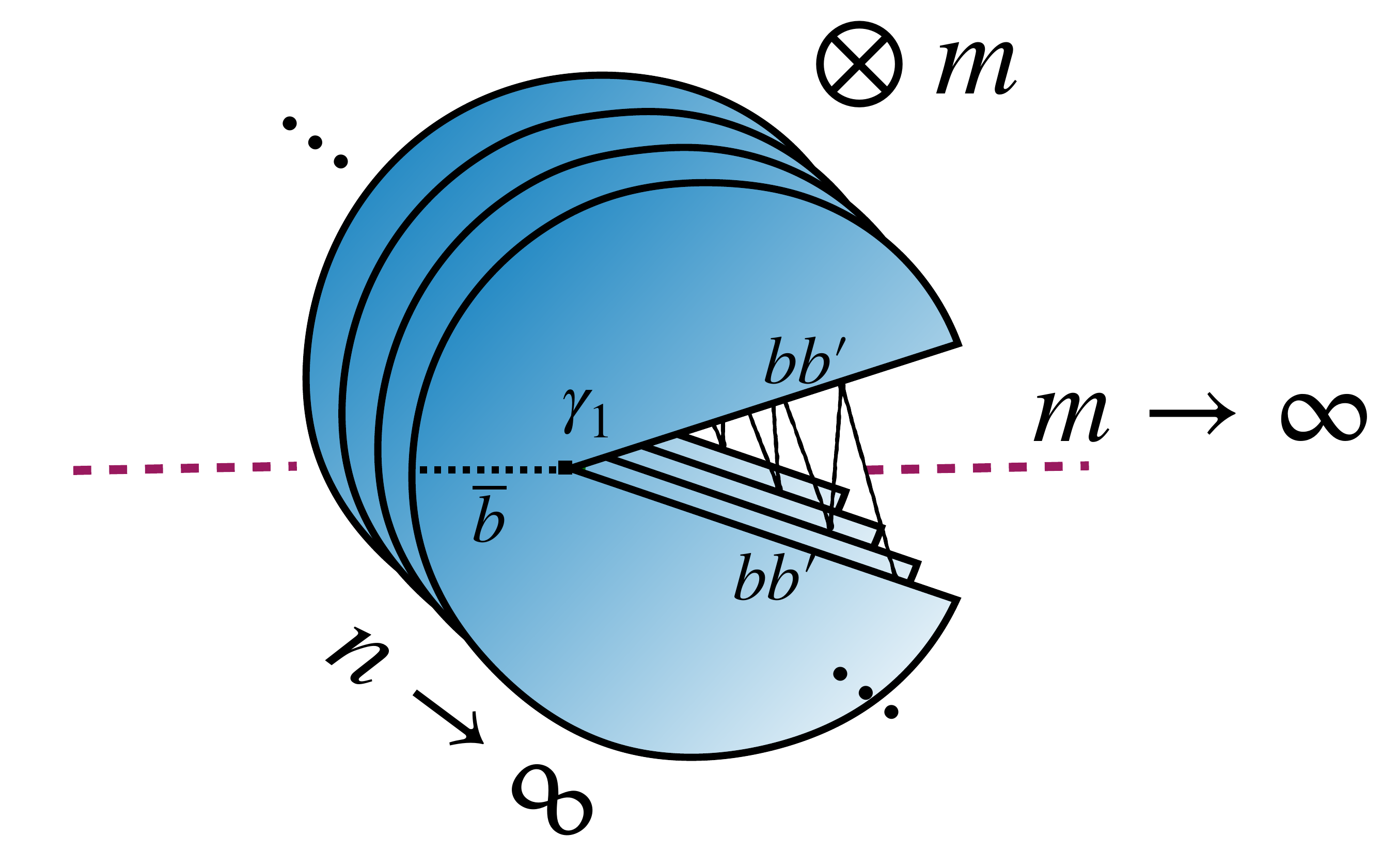}
  \caption{\centering QES=$\gamma_1$}
\end{subfigure}
\caption{{\bf The dominant saddles.} The left figure (a) denotes the dominant saddle configuration in the limit $n,m\rightarrow\infty$ when $H_\mi(b'|b)_\rho>\frac{A_2-A_1}{4G_N}$; and the configuration (b) is dominant when $H_\ma(b'|b)_\rho<\frac{A_2-A_1}{4G_N}$.}
\label{fig:dominantterms}
\end{figure}
To apply the AEP replica trick, we carefully choose feasible states $\rho', \rho''$ on the $m$ families of replicas, such that we can use the criteria in \eqref{eq:refined_QES} to pick out the dominant saddle that corresponds to the generalized entropy of some QES. Our goal is to show that for any integer $m$, there exists bulk states $\rho'_{b^mb'^m\b^m}, \rho''_{b^mb'^m\b^m}$ that satisfy the following desirable properties\footnote{In fact, we can further show that the supported bulk state $\rho'_{b^mb'^m\b^m}, \rho''_{b^mb'^m\b^m}$ we shall choose is actually permutation-invariant among the $m$ replicas (cf. Section~\ref{sec:proof}), but it is not essential for our derivation.}:
\begin{itemize}
    \item $\rho', \rho''$ share the same dual geometry $g$ as $\rho^{\otimes m}$, and their bulk states are $\eps$-close to $\rho^{\otimes m}_{bb'\b}$;
    \item The bulk states satisfy the chain rules \eqref{chainrule1} and \eqref{chainrule2} respectively;
    \item Their marginal state $\rho'_{b^m}(\rho''_{b^m})$ is the maximizer(minimizer) for the smooth min(max)-entropy $H^\eps_\mi(b^m)_{\rho^{\otimes m}}(H^\eps_\ma(b^m)_{\rho^{\otimes m}})$ respectively.
\end{itemize}
The first property ensures they are indeed feasible choices for the smooth min-entropy, so evaluating \eqref{minAEP} for $\rho'$ and $\rho''$ shall give us the lower bounds. To pick the dominant one among the $2^m$ saddle contributions, we need to compare the differences between the areas and various min-entropies. The second property links the difference between the min-entropies with the conditional min-entropy condition in \eqref{eq:refined_QES}, and it's the key technical tool that leads to the condition. More specifically, we will show that
\begin{equation}\label{claim1}
    \frac{Z_{n,m}[A_{1,2}]}{Z_1[A_{1,2}]^{nm}} = 2^{m(1-n)\frac{A_2}{4G_N}}\tr(\rho'^{\otimes n}_{b^m}\cdot U_{\tau|b_1}\otimes\cdots\otimes U_{\tau|b_m})
\end{equation}
if $H_\mi(b'|b)_\rho>\frac{A_2-A_1}{4G_N}$, corresponding to all $\pi_{i}$'s in \eqref{partitionfunction} being the identities. The other relatively smaller terms will be exponentially suppressed in the limit of large $m,n$. On the other hand, with a different feasible state $\rho''_{b^mb'^m\b^m}$, the dominant contribution in the partition function \eqref{partitionfunction} is given by
\begin{equation}\label{claim2}
\begin{aligned}
    \frac{Z_{n,m}[A_{1,2}]}{Z_1[A_{1,2}]^{nm}} = 2^{m(1-n)\frac{A_1}{4G_N}}\tr(\rho''^{\otimes n}_{b^mb'^m}\cdot U_{\tau|b_1}\otimes\cdots\otimes U_{\tau|b_m}\otimes U_{\tau|b'_1}\otimes\cdots\otimes U_{\tau|b'_m})
\end{aligned}
\end{equation}
if $H_\ma(b'|b)_\rho<\frac{A_2-A_1}{4G_N}$, corresponding to all $\pi_{i}$'s being the cyclic permutations in \eqref{partitionfunction}. The two cases are depicted in Figure~\ref{fig:dominantterms}. The last property ensures that for some fixed $\eps',\eps''$,
\begin{equation}
    H_\mi(b^m)_{\rho'} = H^{\eps'}_\mi(b^m)_{\rho'},\quad H_\mi(\b^m)_{\rho''} = H^{\eps''}_\mi(\b^m)_{\rho''}\,.
\end{equation}
This allows us to perform the AEP on the relevant bulk marginal states such that we recover the generalized entropy eventually. 

In the complementary regime, when $H_\mi(b'|b)_\rho\leq\frac{A_2-A_1}{4G_N}\leq H_\ma(b'|b)_\rho$, there is no simple criterion to determine the dominant contribution, and the resulting entanglement entropy is generally lower than the minimum of the generalized entropies evaluated on the two QES candidates according to the upper bound \eqref{upperbound}.

Bear in mind that any feasible bulk state $\rrho_{b^mb'^m\b^m}$ gives a lower bound for $S(B)$ via the AEP \eqref{minAEP}, but the bound generally does not have the form being the generalized entropy of any bulk surface. Therefore, we make these particular choices $\rho'_{b^mb'^m\b^m}$ and $\rho''_{b^mb'^m\b^m}$ here to show that the regimes of the validity of the naive QES prescription are refined and characterized by the conditional min/max-entropies. The states we choose serve only as technical tools in our derivation, and we do not make an attempt here to interpret them physically.  \\

\subsection{Deriving the min-entropy condition}

We are given the state $\left(\rho_{bb'\b}^{\otimes m}\right)^{\otimes n}$ supported on $n\times m$ replicas, and we are going to replace the bulk state by some chosen state $\left(\rho'_{b^mb'^m\b^m}\right)^{\otimes n}$. According to \eqref{chainrule1}, for any integer $m$, there exists a state $\rho'_{b^mb'^m\b^m}$, within the $\eps$-ball of $\rho_{bb'\b}^{\otimes m}$ with some $\eps$,  such that for any $k$ out of $m$ replicas and some fixed $0<\eps'<\eps$,
\begin{equation}\label{chainrule11}
    H_\mi(b^mb'^{m-k})_{\rho'}\geq H_\mi(b'^{m-k}|b^m)_{\rho^{\otimes m}} + H^{\eps'}_\mi(b^m)_{\rho^{\otimes m}}=H_\mi(b'^{m-k}|b^m)_{\rho^{\otimes m}} + H_\mi(b^m)_{\rho'}.
\end{equation}
Here $\rho'_{b^mb'^{m-k}}$ denotes any marginal state obtained from $\rho'_{b^mb'^m\b^m}$ by tracing out some arbitrary $k$ out of $m$ replicas of the $b'$ system. \JZ{Although here one can pick different sets of $k$ replicas, it is unambiguous in the above equation to not distinguish them, because the state $\rho'_{b^mb'^m\b^m}$ here can be chosen to be permutation-invariant among the $m$ replicas (cf. Section~\ref{sec:proof}). }Note that the equality above says the marginal $\rho'_{b^m}$ maximizes $ H^{\eps'}_\mi(b^m)_{\rho^{\otimes m}}$.

How is \eqref{chainrule11} gonna be useful to us? Note that the conditional min-entropy $H_\mi(b'^{m-k}|b^m)_{\rho^{\otimes m}}$ is evaluated on a product state $\rho_{bb'}^{\otimes\, m-k}\otimes\rho_{b}^{\otimes k}$. It reduces to
\begin{equation}
    H_\mi(b'^{m-k}|b^m)_{\rho^{\otimes m}} = H_\mi(b'^{m-k}|b^{m-k})_{\rho^{\otimes m}} = (m-k) H_\mi(b'|b)_\rho
\end{equation}
where in the first equality, we drop the redundant $k$ of $b$ systems thanks to the tensor factorization; and we use the tensor-additivity of the min-entropy in the second equality.

Now suppose $H_\mi(b'|b)_\rho>\frac{A_2-A_1}{4G_N}$, \eqref{chainrule11} implies
\begin{equation}
    H_\mi(b^mb'^{m-k})_{\rho'}>(m-k)\frac{A_2-A_1}{4G_N}+ H_\mi(b^m)_{\rho'}.
\end{equation}
Exponentiating both sides implies that for large enough $n$,
\begin{equation}
   \tr\left(\rho'_{b^m b'^{m-k}}\right)^n<2^{(1-n)(m-k)\frac{A_2-A_1}{4G_N}}\tr \left(\rho'_{b^m}\right)^n.
\end{equation}
Rearranging,
\begin{equation}
   2^{(1-n)(m-k)\frac{A_1}{4G_N}+(1-n)k\frac{A_2}{4G_N}}\tr\left(\rho'_{b^m b'^{m-k}}\right)^n < 2^{m(1-n)\frac{A_2}{4G_N}}\tr \left(\rho'_{b^m}\right)^n.
\end{equation}
Therefore, we can conclude that, when we take the limit $n\rightarrow\infty$, the term with all the $b'$ systems traced out exponentially dominates over any other terms in the expansion of the partition function \eqref{terms}. \JZ{This is not yet enough as we want the this term to dominate over the sum of all the rest $2^m-1$ terms. Note that in the AEP, one takes the limit of $n\rightarrow\infty$ then the limit $m\rightarrow\infty$. For any integer $m$, we have the total contributions of the rest $2^m-1$ terms being bounded by
\begin{equation}
      \sum_{k=0}^{m-1}\binom{n}{k}2^{(1-n)(m-k)\frac{A_1}{4G_N}+(1-n)k\frac{A_2}{4G_N}}\tr\left(\rho'_{b^m b'^{m-k}}\right)^n
      \leq (2^m-1)2^{(1-n)(m-k^*)\frac{A_1}{4G_N}+(1-n)k^*\frac{A_2}{4G_N}}\tr\left(\rho'_{b^m b'^{m-k^*}}\right)^n
\end{equation}
where $k^*$ corresponds to the maximum term in the sum. Now we compare it with the dominant term after taking $\frac{1}{1-n}\log[\cdot]$ on both sides. Clearly the extra factor here only gives a contribution $\frac{\log(2^m-1)}{1-n}$ that vanishes in the limit $n\rightarrow\infty$. Hence, one can conclude this stronger dominance statement from the chain rule.}
This proves our claim \eqref{claim1}. \\

Since the feasible state is chosen such that $H_\mi(b^m)_{\rho'}= H_\mi^\eps(b^m)_{\rho_b^{\otimes m}}$, we have when $H_\mi(b'|b)_\rho>\frac{A_2-A_1}{4G_N}$,
\begin{equation}\label{bulkaep1}
\begin{aligned}
    S(B)_\rho &\geq \lim_{m\rightarrow\infty}\lim_{n\rightarrow\infty}\frac{1}{m(1-n)}\log\frac{Z_{n,m}[A_{1,2}]}{Z_1[A_{1,2}]^{nm}} = \frac{A_2}{4G_N} + \lim_{m\rightarrow\infty}\lim_{n\rightarrow\infty}\frac{1}{m(1-n)}\tr(\rho'_{b^m})^n\\ &= \frac{A_2}{4G_N} + \lim_{m\rightarrow\infty}\frac{1}{m}H_\mi^{\eps'}(b^m)_{\rho^{\otimes m}} = \frac{A_2}{4G_N} + H(b)_\rho\,
\end{aligned}    
\end{equation}
where we've used the min-AEP \eqref{minAEP} on the bulk region $b^m$ in the last step.\\

\subsection{Deriving the max-entropy condition}

Now we turn to \eqref{claim2} and shall make use of the second chain rule \eqref{chainrule2}. We choose the bulk state $\rho''_{b^mb'^m\b^m}$, such that for any $k$ out of $m$ replicas and some fixed $0<\eps''<\eps$,
\begin{equation}\label{chainrule22}
    H_\mi(b^mb'^m)_{\rho''}=H^{\eps''}_\mi(b^mb'^m)_{\rho''}\leq H_\ma(b'^k|b'^{m-k}b^m)_{\rho^{\otimes m}} + H_\mi(b^mb'^{m-k})_{\rho''}.
\end{equation}
We can reduce the conditional max-entropy term as above,
\begin{equation}
    H_\ma(b'^k|b'^{m-k}b^m)_{\rho^{\otimes m}} = H_\ma(b'^k|b^k)_{\rho^{\otimes m}} = k H_\ma(b'|b)_\rho,
\end{equation}
and then if $H_\ma(b'|b)_\rho<\frac{A_2-A_1}{4G_N}$,
\begin{equation}
    H_\mi(b^mb'^m)_{\rho''}<k\frac{A_2-A_1}{4G_N}+ H_\mi(b^mb'^{m-k})_{\rho''}.
\end{equation}
Exponentiating both sides implies that for large enough $n$,
\begin{equation}
   \tr\left(\rho''_{b^m b'^m}\right)^n>2^{k(1-n)\frac{A_2-A_1}{4G_N}}\tr \left(\rho''_{b^mb'^{m-k}}\right)^n.
\end{equation}
Rearranging,
\begin{equation}\label{A1dominant}
   2^{m(1-n)\frac{A_1}{4G_N}}\tr\left(\rho''_{b^m b'^m}\right)^n > 2^{(m-k)(1-n)\frac{A_1}{4G_N}+k(1-n)\frac{A_2}{4G_N}}\tr \left(\rho''_{b^mb'^{m-k}}\right)^n
\end{equation}
where the RHS are exactly the subleading contributions \eqref{terms}. \JZ{Again, one can argue that the dominant term is larger than the sum of all the rest $2^m-1$ terms as we did in the previous subsection.} This proves our claim \eqref{claim2}.\\

Since the feasible state is chosen such that $H_\mi(b^mb'^m)_{\rho''}= H_\mi^{\eps''}(b^mb'^m)_{\rho^{\otimes m}}$, we have when $H_\ma(b'|b)_\rho<\frac{A_2-A_1}{4G_N}$, 
\begin{equation}\label{bulkaep2}
    \begin{aligned}
    S(B)_\rho &\geq \lim_{m\rightarrow\infty}\lim_{n\rightarrow\infty}\frac{1}{m(1-n)}\log\frac{Z_{n,m}[A_{1,2}]}{Z_1[A_{1,2}]^{nm}}= \frac{A_1}{4G_N} + \lim_{m\rightarrow\infty}\lim_{n\rightarrow\infty}\frac{1}{m(1-n)}\tr(\rho''_{b^mb'^{m}})^n\\
    & = \frac{A_1}{4G_N} + \lim_{m\rightarrow\infty}\frac{1}{m}H_\mi^{\eps''}(b^mb'^m)_{\rho^{\otimes m}} = \frac{A_1}{4G_N} + H(bb')_\rho\,.
\end{aligned}
\end{equation}

There is a much simpler alternative to obtain the same result. The global purity of $\rho_{bb'\b}^{\otimes m}$ suggests we can simply follow the same steps for the complementary region $\overline{B}$ and replace $b$  with $\b$ in the chosen feasible state $\rho'_{b^mb'^m\b^m}$ above. We can eventually reach the same conclusion: when $H_\ma(b'|b)_\rho = - H_\ma(b'|\b)_\rho <\frac{A_2-A_1}{4G_N}$,
\begin{equation}
    \begin{aligned}
    S(B)_\rho=S(\overline{B})_\rho &\geq \lim_{m\rightarrow\infty}\lim_{n\rightarrow\infty}\frac{1}{m(1-n)}\log\frac{\overline{Z}_{n,m}(\rho')}{\overline{Z}_1^{nm}}= \frac{A_1}{4G_N} + \lim_{m\rightarrow\infty}\lim_{n\rightarrow\infty}\frac{1}{m(1-n)}\tr(\rho'_{\b^m})^n\\
    & = \frac{A_1}{4G_N} + \lim_{m\rightarrow\infty}\frac{1}{m}H_\mi^{\eps'}(\b^m)_{\rho^{\otimes m}} = \frac{A_1}{4G_N} + H(\b)_\rho = \frac{A_1}{4G_N} + H(bb')_\rho\,
\end{aligned}
\end{equation}
where we've used the global purity in the first and the last equality.\\

The two arguments shown above are basically equivalent as one can see from the proof of the second chain rule in the Section~\ref{sec:proof}. We prefer the first argument because it demonstrates the change of the dominant contribution in the \emph{same} partition function \eqref{partitionfunction} rather than switching to the complementary one. It also generalizes better to the multiple QES situations discussed in the next section.\\

Recall that the upper bounds \eqref{upperbound} and \eqref{threeentropies} imply
\begin{equation}
    S(B)_\rho \leq
\begin{cases}
      \frac{A_1}{4G_N} + H(bb')_\rho\,,\quad &\text{when}\quad H_\ma(b'|b)_\rho<\frac{A_2-A_1}{4G_N}\,;\\
     \frac{A_2}{4G_N} + H(b)_\rho\,,\quad &\text{when}\quad H_\mi(b'|b)_\rho>\frac{A_2-A_1}{4G_N}\,,
\end{cases}
\end{equation}

Altogether with the lower bounds we obtain
\begin{equation}
S(B)_\rho =
\begin{cases}
      \frac{A_1}{4G_N} + H(bb')_\rho\,,\quad &\text{when}\quad H_\ma(b'|b)_\rho<\frac{A_2-A_1}{4G_N}\,;\\
     (\mathrm{indefinite})\,, &\mathrm{when}\quad H_\mi(b'|b)_\rho\leq \frac{A_2-A_1}{4G_N}\leq H_\ma(b'|b)_\rho\,;\\
     \frac{A_2}{4G_N} + H(b)_\rho\,,\quad &\text{when}\quad H_\mi(b'|b)_\rho>\frac{A_2-A_1}{4G_N}\,.
\end{cases}
\end{equation}\\

Lastly, we need to add back the smoothing to the conditional min/max-entropies to yield  \eqref{eq:refined_QES}, matching the original proposal by AP (cf. equation (1.12) in~\cite{akers2021leading}). It seems that we have now two versions of the QES prescription, depending on whether we add the smoothing. They are in fact equivalent in practise because the entropy values the two versions prescribe have negligible differences at the leading orders. For instance, we can consider two different bulk states $\rho_{bb'\b}, \rrho_{bb'\b}$, and $\rrho_{bb'}$ is the minimizer for $H^\eps_\ma(b'|b)_\rho$. The difference between the smooth and non-smooth conditions can be manifested in applying the (non-)smoothed refined QES prescription to $\rho (\rrho)$ respectively, such that the conditions that determine regime 1 are identical in the two cases. Then we can ask what is the difference in the generalized entropy calculated for these two states. By invoking the Fannes-Audenaert continuity bound for the von Neumann entropy~\cite{fannes1973continuity,audenaert2007sharp}, we have
\beq
|H(bb')_{\rrho} - H(bb')_\rho|\leq \delta(\rrho,\rho)\log d + h(\delta(\rrho,\rho))
\eeq
where $\delta(\rrho,\rho):=\frac12||\rrho-\rho||_1$ is the trace distance and the $h(\cdot)$ is the binary entropy function. By definition, we have  $\delta(\rrho,\rho)\leq P(\rrho,\rho)\leq \eps$, so the bound roughly scales linearly with $\eps\sim\O(\mathrm{Poly}(G_N))$:
\beq
|H(bb')_{\rrho} - H(bb')_\rho|\leq\O(\mathrm{Poly}(G_N)).
\eeq
 By exploiting this wiggle room,\footnote{Though AP did not explicitly put the $\O(\eps)$ correction in their statement of the prescription (equation (1.12)), it is indeed noted in their random tensor network derivation (cf. equation (5.15)).} we can easily convert our derived prescription above to the smoothed version \eqref{eq:refined_QES} and the difference can be neglected at the leading orders. This completes our derivation.

\subsection{Remarks}

\emph{\quad\quad The indispensable conditional min/max-entropies}. We shall emphasize that the chain rule used is crucial in our derivation to link the conditional min-entropy criterion, with the difference between the dominant term and \emph{any} other terms in the expansion. Had we only needed to pick among the two extreme terms with all $\pi_{b'_i}$'s being $\tau$ or $I$, we won't need the chain rule and the condition reduces to  
\begin{equation}
    H_\mi(b^mb'^m)-H_\mi(b^m)>\frac{A_2-A_1}{4G_N}.
\end{equation} 
Also note that suppose we do not have the additional $m$ replicas or the they also have the product form as in \eqref{trivialproduct}, then the condition determining the dominant contribution for the min-entropy is exactly to compare their differences as above. Due to the presence of the extra terms with mixed $\pi_{b'_i}$'s and also the correlation in the state $\rho'_{b^mb'^m\b^m}$, the differences between min-entropies cannot serve as the criteria to pick out the dominant saddle, and one must resort to the conditional min-entropy via the chain rule instead. \\

\emph{Why do smoothing?}. It would seem unnecessary to add the smoothing in the end of the derivation while the entropy values they prescribe are more or less the same. As far as computing the entanglement entropy is concerned, one motivation to use the smooth entropies is that they could be more robustly defined than the non-smoothed counterparts in quantum field theory. For example, when considering the one-interval entanglement entropy in 2D CFT, though the von Neumann entropy and min-entropy are finite after regularization, the (Hartley) max-entropy ($H_0$) is still infinite. The smoothing then helps to tame the infinity by truncating out the small eigenvalues~\cite{czech2015information,bao2019beyond}. 

Moreover, the smooth entropies are more favoured in that they also characterize how robust a protocol is under approximation. Operationally, the $\eps$ that appears in the smooth entropy is the error tolerance in achieving an information processing task.  The smoothed quantities definitely serve as better criteria when it comes to characterizing operational tasks like the entanglement wedge reconstruction, which is anyways approximate in holography~\cite{hayden2019learning}. We thus expect the smooth min/max-entropies will also be the more fundamental quantities to use in determining the entanglement wedge and the reconstructable region,  due to the close connection between the RT formula and EWR~\cite{harlow2017ryu}. In fact, AP suggested that the task of state-dependent EWR should be understood as the one-shot state merging, and there $H^\eps_\ma(b'|b)$ determines the bits one can transfer over from $b'$ to $b$, and $\eps$ is the error tolerance of achieving the task.  We shall come back to this discussion later. \\

\emph{Beyond fixed-area states}. So far we have only considered the fixed-area states, such that we have a fixed area contribution to the entanglement entropy. It allows us to directly work with R\'enyi entropies and facilitates a clean separation between the area and the bulk field contribution when doing the smoothing. To argue for a generic holographic state, one can lift the area-fixing constraint under certain \emph{diagonal assumption} in the semiclassical limit. One can argue that the marginal of a generic holographic state is approximately diagonal in the basis of fixed-area states, and therefore its entropy is the average of the entropies of the fixed-area states. Due to the semiclassical Gaussian fluctuation in area, such an averaging brings a $\O(G^{-\frac12})$ correction to the naive QES prescription near the transition window of size $\O(G^{-\frac12})$. This is seen in the JT black hole Page curve calculation~\cite{penington2019replica} and also generally argued in~\cite{akers2021leading,marolf2020probing}. On the other hand, AP used the dustball example to show that the correction can also occur for fixed-area states at the leading order within the indefinite regime 2. Then the correction can be carried over to a generic state also after the averaging, rendering the $\O(G^{-\frac12})$ correction a subleading effect. Similarly, the ``safe zones'', i.e. regimes 1 and 3, bordered by the conditional min/max-entropies are as intact for generic holopgrahic states as for fixed-area states at the leading order. Overall, the statement is that the refined QES prescription \eqref{eq:refined_QES} also holds for general holographic states as far as the leading order is concerned. Since our derivation does not provide any new insights beyond the diagonal approximation, we recommend the readers the detailed arguments in~\cite{akers2021leading,marolf2020probing}. \\

\emph{The indefinite regime}. We see that the refined QES prescription essentially restricts the validity regimes of the naive QES prescription, opening an indefinite regime where the value of $S(B)$ may not be simply given by the generalized entropy of some QES. Admittedly, this is not the full version of the desired refined QES prescription as a description of the indefinite regime is still lacking, in which we cannot tell which of the $2^m$ terms in the partition function expansion dominates. It then implies that the upper and lower bounds do not generally match as in regimes 1 and 3, leaving an $\O(1)$ room (in terms of the area difference) for deviations from the naive QES prescription.  In this regime, we can hardly identify a bulk surface whose associated generalized entropy measures the entanglement entropy. Nevertheless, we can still hope to obtain useful bounds for $S(B)$ via this approach even in this indefinite regime.  We expect $S(B)_\rho$ in the indefinite regime to be \emph{generically} smaller than what the naive QES prescription predicts at the leading order, whenever the entanglement spectrum of the bulk state gives a $\O(G_N^{-1})$ window between the condition min/max-entropies. However, here we shall not attempt to rigorously formulate and prove such a claim, and leave the thorough investigation to future works.

\section{The holographic R\'enyi entropy of fixed area states: a theorem}\label{sec:holorenyi}

We consider here the problem of computing the holographic  R\'enyi entropy in the bulk. Generally, the holographic (modular) R\'enyi entropies are given by the cosmic brane prescription~\cite{dong2016gravity}, where a brane of finite tension is inserted into the spacetime, and backreaction needs to be considered. In fixed-area states, the problem significantly simplifies as all the R\'enyi entropies are evaluated on one of the QES surfaces whose classical areas are fixed. The problem then reduces to picking the right QES surface to evaluate the generalized R\'enyi entropies\footnote{Here, the notion of generalized R\'enyi entropies should be understood in a restrictive sense because we have fixed the area term. } associated with it. 

In the case of the holographic von Neumann entropy we studied, we saw that the usual rule of picking the minimal QES surface fails near the transition regime. It begs the same question if the non-trivial entanglement spectrum of the bulk state shall also yield a complicated rule for computing the holographic R\'enyi entropies. It was claimed that the higher holographic R\'enyi entropies have sharp QES transitions~\cite{akers2021leading} based on some examples.

When restricted to pure bulk marginal states $\rho_{bb'}$, we prove a theorem (cf. Theorem~\ref{thm:dominance}) implying that the simple rule of picking the QES surface with the minimal generalized bulk R\'enyi entropy still works well for computing the holographic R\'enyi entropy of fixed area states. This result provides evidence that, unlike the von Neumann entropy, R\'enyi entropies tend to have sharp transitions between QES surfaces without an indefinite regime. \\

We first need to introduce some preliminary notions. 
\begin{defn}\label{def:nc}
A partition $\pi$ of the set $[n]:=\{1,\ldots,n\}$ is defined as $\pi=\{X_i\}_i$ such that the blocks $X_i$ satisfy $X_i\ne\emptyset\,\,\forall i$, $X_i\cap X_j=0\,\, \forall i\ne j$ and $\cup_i X_i=[n]$. A partition $\pi$ is non-crossing if we do not have $a_1<b_1<a_2<b_2$ for some $a_1,a_2\in X_i$ and $b_1,b_2\in X_j$ with $i\ne j$.
\end{defn}
The non-crossing (NC) partitions of $[n]$, of which the set is denoted as $NC(n)$, are in one-to-one correspondence with the non-crossing permutations: for each NC partition $\pi=\{X_i\}_i$, we let the NC permutation be defined with cycles $(X_i)_i$ such that each cycle $(X_i)$ consists of the same elements in the block with an increasing order. We use $|\pi|$ to denote both the cardinality of the partition $\pi$ or the number of cycles in the corresponding permutation as we did before. Because of the isomorphism, we shall use the same notation $NC(n)$ for both NC partitions and NC permutations. It should be clear from the context which object we are referring to. 

Given a set $[n]$, consider now doubling it with additional elements $\bar 1,\ldots,\bar n$ and interlace them with $1,\ldots,n$ in an alternating way, $1,\bar 1,\ldots,n,\bar n$. For any $\pi\in NC(n)$, its \emph{Kreweras complement}~\cite{kreweras1972partitions}, $\bar\pi\in NC(\bar 1,\ldots,\bar n)\cong NC(n)$, is defined as the biggest\footnote{There is a natural partial order defined for NC partitions, called the reverse refinement order. We say $\pi_1\le\pi_2$ if each block of $\pi_1$ is contained in the blocks of $\pi_2$. } element in $NC(\bar 1,\ldots,\bar n)$ such that $\pi\cup\bar\pi\in NC(1,\bar 1,\ldots,n,\bar n)$. The corresponding permutation of $\bar\pi$ is equivalently given by $\pi^{-1}\circ\tau$ where $\tau$ is the cyclic permutation $\tau=(1, \cdots, n)$. We quote the following Lemma for NC partitions/permutations.
\begin{lem}[Biane~\cite{biane1997some}]\label{lem:nc}
Let $\pi\in NC(n)$ be a NC partition, we have
\begin{equation}
    |\pi|+|\bar\pi|=|\pi|+|\pi^{-1}\circ\tau|=n+1\ .
\end{equation}
\end{lem}
Generally for any two partitions $\pi_1,\pi_2\in\mathcal{P}(n)$, we have $|\pi_1|+|\pi_2|\le |\pi_1\circ\pi_2|+n$.\footnote{One can define a metric on the permutations known as the \emph{Cayley distance}, $d(\pi_1,\pi_2):=n-|\pi_1^{-1}\circ\pi_2|$, which also counts the minimal number of transpositions needed to transform $\pi_1$ to $\pi_2$. Then equation $|\pi_1|+|\pi_2|\le |\pi_1\circ\pi_2|+n$ follows from the triangle inequality of the Cayley distance. Furthermore, Lemma~\ref{lem:nc} can be equivalently stated for Cayley distance as $d(\pi,I)+d(\pi,\tau)=n-1$.}

Without loss of generality, we consider the case of two QES candidates, and the general case also follows by the same arguments. Consider the $n$-th moment of a density operator $\rho_B$. Recall that the standard replica gravitational path-integral for the fixed area states \eqref{mainZn} gives us
\beq\label{eq:partitionfunction}
\tr\rho^n_B=\frac{Z_n[A_{1,2}]}{Z_1[A_{1,2}]^n} = \sum_{\pi\in S_n} 2^{(|\pi|-n)A_1/4G_N+(|\pi^{-1}\circ\tau|-n)A_2/4G_N}\cdot\tr\left(\rho_{bb'}^{\otimes n}\cdot U_{\tau|b^n}\otimes U_{\pi|{b'}^n}\right)
\eeq 

To proceed, we shall first simplify the problem by discarding the crossing permutations in the sum of \eqref{eq:partitionfunction}, which is justified~\cite{penington2019replica,akers2021leading} by the fact that the areas are implicitly IR divergent, so the leading contributions should maximize the sum $|\pi|+|\pi^{-1}\circ\tau|$, which corresponds to the $\pi\in NC(n)$.  Therefore, we have
\begin{equation}\label{eq:partitionfunction2}
       \frac{Z_n[A_{1,2}]}{Z_1[A_{1,2}]^n} = \sum_{\pi\in NC(n)} 2^{(|\pi|-n)A_1/4G_N+(1-|\pi|)A_2/4G_N}\cdot\tr\left(\rho_{bb'}^{\otimes n}\cdot U_{\tau|b^n}\otimes U_{\pi|{b'}^n}\right)\,.
\end{equation}

As mentioned, we do not yet obtain the general result for a mixed state $\rho_{bb'}$. Instead, we shall proceed and consider the simpler case of $\rho_{bb'}$ being pure.  As a first step, we aim to establish a weaker statement, that is for any $n\ge 2$ the saddle corresponds to $\pi$ equal to $\mathbf{1}$ (or $\tau$) is the largest among all the saddles in \eqref{eq:partitionfunction}, when the R\'enyi entropy of the marginal state on $b$, $H_n(b)_\rho$ is smaller (or larger) than the area difference $(A_1-A_2)/4G_N$. More precisely, we shall prove the following result.

\begin{thm}\label{thm:dominance}
Let $\rho_{bb'}$ be a pure state and $n\in\mathbb{N}, n\ge 2$, and $\Delta:=(A_1-A_2)/4G_N$ be the area difference between the two QES. Then we have a partial order among the summands of \eqref{eq:partitionfunction2}: for any non-crossing permutation $\pi$,
\begin{equation}\label{eq:weak3}
    2^{(n-|\pi|) \Delta}\cdot\tr\rho_b^n\ge \tr\left(\rho_{bb'}^{\otimes n}\cdot U_{\tau|b^n}\otimes U_{\pi|{b'}^n}\right)\,,\quad\quad \text{if}\quad H_n(b)_\rho\le\Delta
\end{equation}
\begin{equation}\label{eq:weak4}
    2^{(1-|\pi|) \Delta}\ge \tr\left(\rho_{bb'}^{\otimes n}\cdot U_{\tau|b^n}\otimes U_{\pi|{b'}^n}\right)\,,\;\;\quad\quad\quad\quad\text{if}\quad H_n(b)_\rho\ge\Delta
\end{equation}
\end{thm}

The key is to understand the RHS of \eqref{eq:weak3} and \eqref{eq:weak4}, for which we would need the following lemma.

\begin{lem}\label{lem:kreweras}
Let $\rho_{AB}$ be a pure state on $\h_{A}\otimes\h_B$. For any $n\in\mathbb{N}^+$ and any NC permutation $\pi\in NC(n)$ and its Kreweras complement $\bar\pi$, we have
\begin{equation}
    \tr\left(\rho^{\otimes n}_{AB}\cdot U_{\tau|A^n}\otimes U_{\pi|B^n}\right)=\prod_{\bar X\in \bar\pi}\tr\rho_A^{|\bar X|}\ .
\end{equation}
\end{lem}

\begin{proof}
A \emph{vectorization} of a linear operator $X_A$ is defined as $X_A\otimes I_{A'}\ket{\Phi}_{AA'}$  where $A'\cong A$ and $\ket{\Phi}_{AA'}:=\sum_i\ket{i}_A\ket{i}_{A'}$ is the unnormalized Bell state. It is denoted as $\ket{X}_A$. In particular, the canonical purification of a density matrix $\rho_A$ is defined as $\ket{\sqrt{\rho}}_{AA'}$. Let us also denote the state vector of the pure state $\rho_{AB}$ as $\ket{\rho}_{AB}$. Since all purifications are equivalent up to a local isometry on the purifying system, we have
\begin{equation}
    \ket{\rho}_{AB}=V_{A'\to B}\ket{\sqrt{\rho}}_{AA'}
\end{equation}
where $V_{A'\to B}$ is a unitary, an isometry or a partial isometry depending on the dimensions $|A|, |B|$. Then it follows that 
\begin{equation}
    U_{\tau|A^n}\otimes U_{\pi|B^n} \ket{\rho}_{AB}^{\otimes n}=V_{A'\to B}^{\otimes n}U_{\tau|A^n}\otimes U_{\pi|A'^n}\ket{\sqrt{\rho}^{\otimes n}}_{A^nA'^n}
\end{equation}
where the unitaries commute because the permutations commute with any local unitary/isometry. Then we use the identity (cf. for example, eq.(2.79) in~\cite{watrous2018theory}), for any linear operators $P_A$ and $Q_{A'}$, we have
\begin{equation}
    P_A\otimes Q_{A'} \ket{\sqrt{X}}_{AA'}= \ket{P\sqrt{X}Q^\intercal}_{AA'}
\end{equation}
to obtain 
\begin{equation}
    V_{A'\to B}^{\otimes n}U_{\tau|A^n}\otimes U_{\pi|A'^n}\ket{\sqrt{\rho}^{\otimes n}}_{A^nA'^n} =  V_{A'\to B}^{\otimes n}\ket{U_{\tau}\sqrt{\rho}^{\otimes n}U_{\pi}^{-1}}_{A^nA'^n}=V_{A'\to B}^{\otimes n}\ket{\sqrt{\rho}^{\otimes n}U_{\tau}U_{\pi}^{-1}}_{A^nA'^n} 
\end{equation}
where the first equality use the fact that the unitary representations of the permutation group on $\h_A^{\otimes n}$ have real entries and thus symmetric; and the second equality follows from $U_\tau\rho^{\otimes n}U_\tau^*=\rho^{\otimes n}$. Then we have
\begin{equation}
\begin{aligned}
    \tr\left(\rho^{\otimes n}_{AB}\cdot U_{\tau|A^n}\otimes U_{\pi|B^n}\right)&=\bra{\sqrt{\rho^{\otimes n}}}V^{*\otimes n}_{A'\to B}V_{A'\to B}^{\otimes n}\ket{\sqrt{\rho^{\otimes n}}U_{\tau}U_{\pi}^{-1}}_{A^nA'^n}\\
    &=\tr(\rho_A^{\otimes n} U_\tau U_{\pi}^{-1})=\tr( \rho_A^{\otimes n} U_{\tau\circ\pi^{-1}})=\prod_{\bar X\in \bar\pi}\tr\rho_A^{|\bar X|}\ 
\end{aligned}
\end{equation}
where the last step follows from the definition of the Kreweras complement.
\end{proof}

\begin{proof}[Proof of Theorem~\ref{thm:dominance}]
Using Lemma~\ref{lem:kreweras}, we can unpack the trace function
\begin{equation}
\tr\left(\rho_{bb'}^{\otimes n}\cdot U_{\tau|b^n}\otimes U_{\pi|{b'}^n}\right)=\prod_{\bar X\in\bar\pi}\tr\rho_b^{|\bar X|}=\prod_{i=1}^{n+1-k}\tr\rho_b^{n_i}. 
\end{equation}
where $n_i$ denotes the size of each cycle $\bar X$ in the Kreweras complement $\bar\pi$ of $\pi$.

Next, we want to show that
\begin{equation}\label{eq:subleading}
\tr\rho_b^k\ge\prod_{i=1}^{n+1-k}\tr\rho_b^{n_i}\,,\quad\forall \{n_i\}\vdash [n]\,.
\end{equation}
where 
\begin{equation}\label{eq:ncnn}
    \tr\rho^k_b=\tr\left(\rho_{bb'}^{\otimes n}\cdot U_{\tau|b^n}\otimes U_{\pi|{b'}^n}\right)
\end{equation}
for a $k$-cycle non-crossing and \emph{non-nesting} (NN) permutation $\pi$. These are permutations whose cycles do not nest each other. Then \eqref{eq:ncnn} follows for there are $k$ cycles. 

Taking the $\frac{1}{1-k}\log$ gives
\begin{equation}
    \sum_{i}\frac{1}{1-k}\log\tr\rho_b^{n_i+1}=\sum_{i}\frac{n_i}{k-1}H_{n_i+1}(b)_\rho\ge H_k(b)_\rho
\end{equation}
where we used the monotonicity of the R\'enyi entropies in the last inequality. The claim \eqref{eq:subleading} follows by undoing the $\frac{1}{1-k}\log$.\\

Having established \eqref{eq:subleading}, it suffices to check whether the permutations that are both NC and NN satisfy \eqref{eq:weak3} and \eqref{eq:weak4}, that is for any such permutations $\pi$ with $|\pi|=k$,
\begin{equation}\label{eq:weak5}
    2^{(n-k) \Delta}\cdot\tr\rho_b^n\ge \tr\rho_{b}^k\,,\quad\quad \text{if}\quad H_n(b)_\rho\le\Delta
\end{equation}
\begin{equation}\label{eq:weak6}
    2^{(1-k) \Delta}\ge \tr\rho_{b}^k\,,\;\;\quad\quad\quad\quad\text{if}\quad H_n(b)_\rho\ge\Delta
\end{equation}
which, together with \eqref{eq:subleading}, will lead to \eqref{eq:weak3} and \eqref{eq:weak4}.

They follow straightforwardly. Taking the log on the difference between two sides \eqref{eq:weak5} and follow by the division of $n-k$,
\begin{equation}
    \Delta+\frac{1-n}{n-k}H_n(b)_\rho-\frac{1-k}{n-k}H_k(b)_\rho \ge \Delta+\frac{k-n}{n-k}H_n(b)_\rho \ge 0
\end{equation}
 Similarly, \eqref{eq:weak6} follows from
\begin{equation}
    \Delta-H_k(b)_\rho \le \Delta - H_n(b)_\rho \le 0\,.
\end{equation}
\end{proof}

We have shown that the leading contributions are given by the LHS of \eqref{eq:weak3} and \eqref{eq:weak4}. This statement is weaker than what we want, because what we really need is that the dominant summand in \eqref{eq:partitionfunction2} is larger than the sum of the rest contributions. Now we make a stronger claim that the sum of the subleading contributions do not change the estimate of $S_n(B)_\rho$ much. Note that we are not saying the sum of the subleading contributions are suppressed as compared to the leading one. Rather we only need to argue that including them altogether only decreases the $S_n(B)_\rho$ value by at most an $\O(1)$ amount, so such a bounded shift by a few bits can safely be ignored when the system in consideration is large. 

The estimates go as follows. The total number of of non-crossing permutations of $n$ systems is the Catalan number,
\begin{equation}
    C_n=\frac{1}{n+1}\binom{2n}{n}\sim \frac{4^n}{n^{3/2}}\,.
\end{equation}

The partition function \eqref{eq:partitionfunction2} reads
\beq
\begin{aligned}
   \tr\rho^n_B=\frac{Z_n[A_{1,2}]}{Z_1[A_{1,2}]^n}=& \sum_{\pi\in NC(n)} 2^{(|\pi|-n)A_1/4G_N+(1-|\pi|)A_2/4G_N}\cdot\tr\left(\rho_{bb'}^{\otimes n}\cdot U_{\tau|b^n}\otimes U_{\pi|{b'}^n}\right)\,.
\end{aligned}
\eeq 

When $H_n(b)_\rho\ge \Delta$, \eqref{eq:weak3} implies
\begin{equation}
    \begin{aligned}
       \tr \rho_B^n \le& \sum_{\pi\in NC(n)}2^{(1-n)A_1/4G_N}=C_n\cdot2^{(1-n)A_1/4G_N}\le 2^{2n+(1-n)A_1/4G_N},
    \end{aligned}
\end{equation}
and the R\'enyi entropy is given by
\begin{equation}
    S_n(B)_\rho\ge \frac{1}{1-n}\left(\log\tr\rho^n_{bb'}+2n+(1-n)\frac{A_1}{4G_N}\right)=\frac{A_1}{4G_N}(+H_n(bb')_\rho)+\frac{2n}{1-n}\,.
\end{equation}

When $H_n(b)_\rho\le\Delta$, \eqref{eq:weak4} implies
\begin{equation}
    \begin{aligned}
       \tr \rho_B^n \le& \sum_{\pi\in NC(n)}2^{(1-n)A_2/4G_N}\tr\rho^n_{b}=C_n\cdot2^{(1-n)A_2/4G_N}\tr\rho^n_{b}
       \le2^{2n+(1-n)A_2/4G_N}\tr\rho^n_{b}\,,
    \end{aligned}
\end{equation}
and the R\'enyi entropy is given by
\begin{equation}
    S_n(B)_\rho\ge \frac{1}{1-n}\left(\log\tr\rho^n_{b}+2n+(1-n)\frac{A_2}{4G_N}\right)=\frac{A_2}{4G_N}+H_n(b)_\rho+\frac{2n}{1-n}\,.
\end{equation}

We know that the R\'enyi entropy $S_n(B)_\rho$ is bounded from above by $A_1/4G_N+H_n(bb')_\rho$ or $A_2/4G_N+H_n(b)_\rho$ when discarding all the replica asymmetric saddles, we see that including them induces \emph{at most} a negligible four-bits shift in the final entropy for R\'enyi indices $n\ge 2$. For $n\rightarrow\infty$, the min-entropy is only off by two-bits.

On the other hand, near $n=1+\eps$, the lower bound we obtained can be arbitrarily bad as $\eps$ goes to zero, suggesting that the von Neumann entropy has a qualitatively different behaviour near the transition. We know that this is indeed the case from the examples in~\cite{akers2021leading} and the refined Page curve in Section~\ref{sec:pagecurve}. The effect of replica symmetry breaking is clearly significant for von Neumann entropy and it becomes less relevant as the R\'enyi index increases. 

The result shown here partially justifies the replica symmetry assumption used for the refined QES derivation in the previous section. For the application in computing the black hole Page curve in next section, we are confident that the refined QES prescription holds as there the bulk state is pure.  We expect the result can be generalized, at least for the min-entropy, to mixed bulk states $\rho_{bb'}$, but currently a proof is still lacking.

\section{Example: a refined black hole Page curve}\label{sec:pagecurve}

Let us recall that the LM derivation of the RT formula is deemed by the authors as leading to a ``generalized gravitational entropy''~\cite{lewkowycz2013generalized}. Hence, the LM derivation really implies that the RT formula applies beyond the scope of AdS/CFT, and the recent developments of the black hole Page curve calculation confirm their insight. Similarly, the fact that a direct path integral derivation works also for the refined QES prescriptions implies that the correction should apply more generally. Essentially, as long as we are dealing with a gravity theory that enables us to calculate partition functions like $\tr \rho^n$ via a gravitational path integral, the same derivation should basically go through (with caveats of course). Here we confirm this generalization by giving an example of a \emph{refined} black hole Page curve.

We choose the simple model of the Jackiw–Teitelboim (JT) gravity with an End of the World (EOW) brane, also known as the Penington-Shenker-Stanford-Yang (PSSY) model~\cite{penington2019replica}, where the gravitational path integral can be done explicitly such that the entanglement spectrum can be resolved. Therefore, by assigning a bulk state that is generic enough to have distinct conditional min, max and von Neumann entropies, we can perform the sum over geometries and directly compute the entanglement entropy to verify that if the page curve behaves consistently with the prediction of the refined QES prescription. As we shall see, the Page curve does deviate largely in the indefinite regime from the RT formula. It demonstrates a refinement to the \emph{island formula}, which is an incarnation of the QES formula applied to the Hawking radiation. To avoid diverting away from the main theme of this work, we shall only give an example here and more comprehensive analysis of the Page curve refinement will be provided in a future work.

\subsection{Review of the PSSY model}
We start with a minimal review of the PSSY model. The model describes a 2D eternal AdS black hole whose interior degrees of freedom living on the EOW brane are entangled with the early Hawking radiation on some auxiliary system outside of the spacetime (cf. Figure~\ref{fig:setup}). Such a black hole does not evaporate so we model the Page curve evolution by tuning the amount of entanglement by hand. Then if we compute the von Neumann entropy of such an eternal black hole, its value should follow a Page curve, which initially increases and then flattens out, as one tunes up the entanglement. The simplest scenario as studied by~\cite{penington2019replica} is to consider the following maximally entangled state:
\begin{equation}\label{eq:entangled}
    \ket{\rho}_{BR}:=\frac{1}{\sqrt{k}}\sum_{i=1}^k \ket{\psi_i}_B\ket{i}_R
\end{equation}
where $\ket{\psi_i}_B$ is the black hole state with the brane in state $i$, and the $\ket{i}_R$ is the state of the ``Hawking partner'' in the radiation. Here the Schmidt coefficients are chosen to be flat with the Schmidt rank $k$. The entropy of the Radiation system is given by the island formula,
\begin{equation}
    S(R)_\rho=\frac{A[\gamma]}{4G_N}+H(I\cup R)_\rho.
\end{equation}
In JT gravity, the ``area'' term is given by a topological contribution in a crude approximation, and we can take it to be some constant $S_{\text{BH}}$. Since the global state is pure, the formula also computes the entropy of the black hole system. The Page curve depicts how does this entanglement entropy changes as one tunes $k$. According to the QES prescription, the Page curve has two regimes with $k< 2^{S_{\text{BH}}}$ and $k> 2^{S_{\text{BH}}}$ (cf. Figure~\ref{fig:setup}). In the former case, the QES surface is empty and so it the island, so we have no area term but only a bulk entropy term of $H(I\cup R)_\rho=H(R)_\rho=\log k$  evaluated in the state \eqref{eq:entangled}; and in the latter case, the QES surface sits at the horizon and we have no bulk entropy contribution $H(I\cup R)_\rho=H(BR)_\rho=0$ but only a constant area term $S_{\text{BH}}$. The island formula therefore provides us a crude sketch of the Page curve as we tune up $k$.

\begin{figure}
\centering
\begin{subfigure}{.5\textwidth}
  \centering
  \includegraphics[width=.8\linewidth]{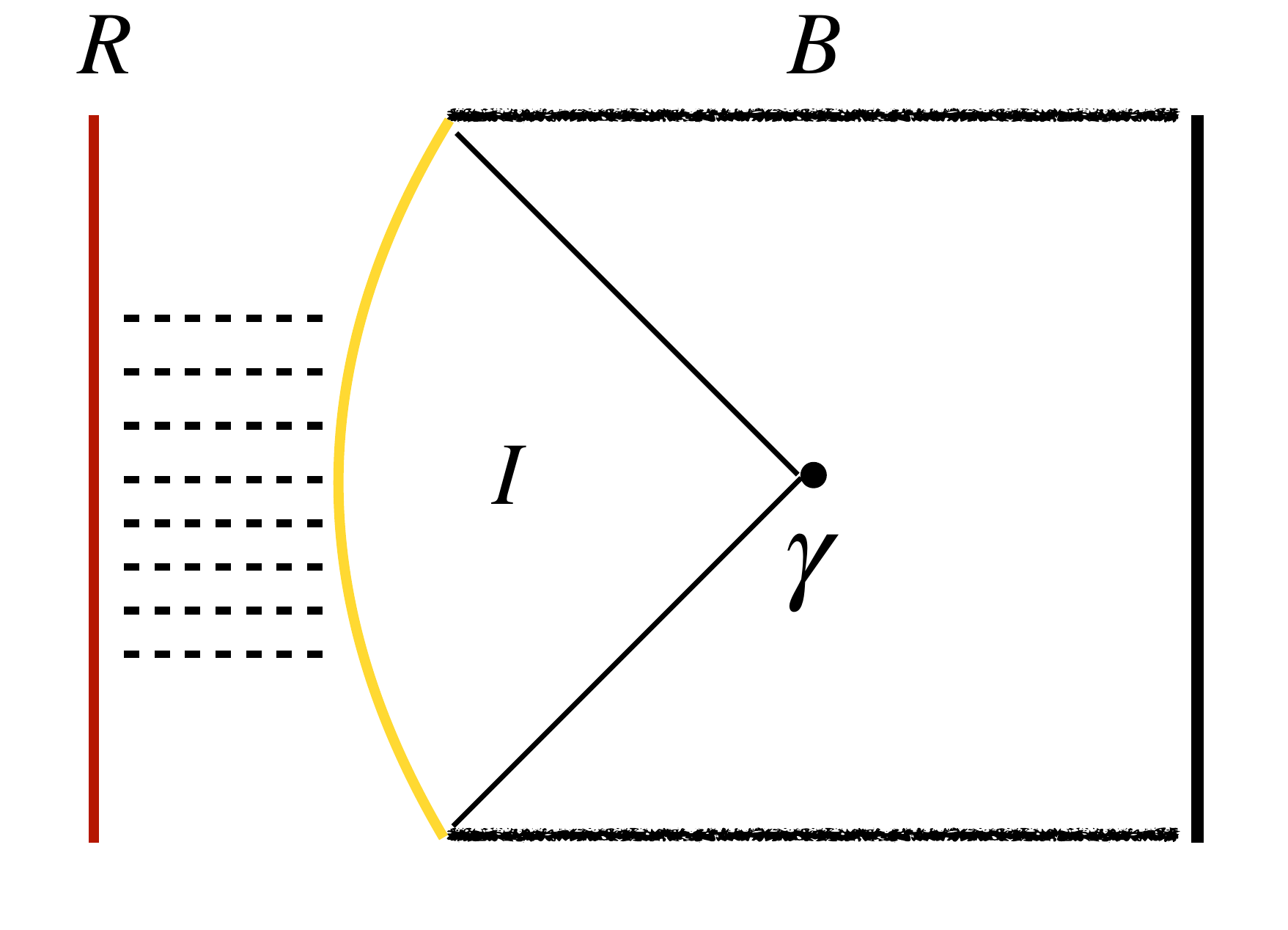}
  \caption{\centering $k>2^{S_{\text{BH}}}$}
\end{subfigure}%
\begin{subfigure}{.5\textwidth}
  \centering
  \includegraphics[width=0.8\linewidth]{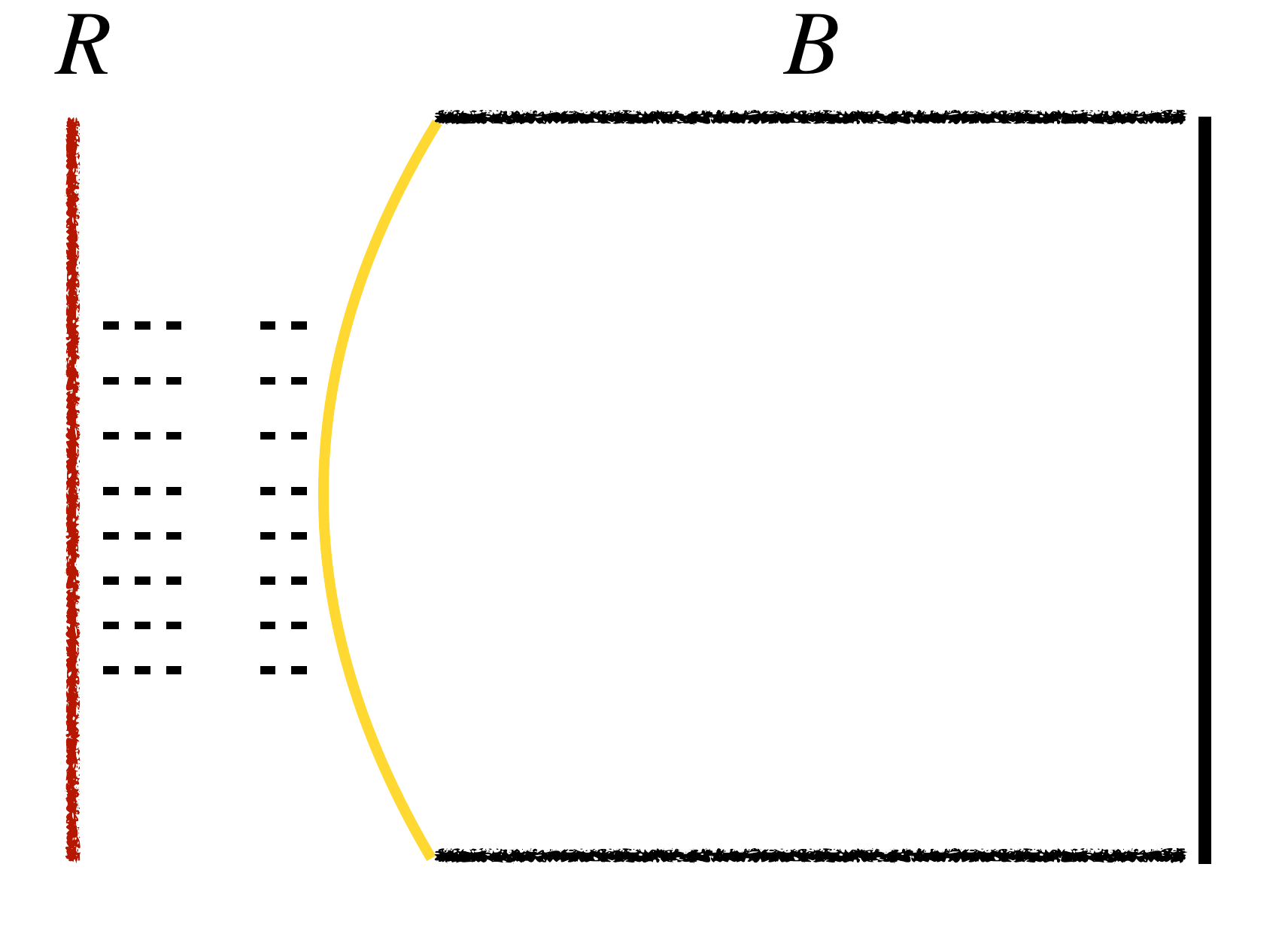}
  \caption{\centering $k<2^{S_{\text{BH}}}$}
\end{subfigure}
\caption{The auxiliary system storing Hawking radiation is depicted red, and the EOW brane is in yellow. The entanglement between them is by the dash line. (a) The island region is contained the entanglement wedge of $R$ while QES is at $\gamma$. The entanglement and purity are preserved, so the entropy $S(R)_\rho$ is given by the area term $2^{S_{\text{BH}}}$.(b) The QES and the island are empty, and the brane is no longer contained in the entanglement wedge of $R$. The entanglement is broken in the entanglement wedge, so the entropy $S(R)_\rho$ is given by the bulk entropy term $H(R)_\rho=\log k$ .}
\label{fig:setup}
\end{figure}

However, the formula looks nonsensical at a first sight: Since \eqref{eq:entangled} is given in the form of a Schmidt decomposition, then it seems straightforward that the entropy is simply $\log k$ according to \eqref{eq:entangled}. The crux is to note that the inner product $\braket{\psi_i}{\psi_j}$ is defined via the gravitation path integral, so one needs to be careful about whether it is legit to use the Dirac notation, which implicitly assumes the standard Hilbert space treatment. In fact, the path integral calculation indicates that the orthogonality among $\ket{\psi_i}$'s is not respected when we considered higher moments like $\tr\,\rho^n_R$ for $n>1$ while we compute the R\'enyi entropies. This is a curious feature of the gravitational path integral that causes the factorization problem~\cite{maldacena2004wormholes,arkani2007}, and such deviations from orthogonality are essentially due to the replica wormhole contributions in the semiclassical calculation. A common interpretation is that the gravitational path integral is actually computing the inner product in an ensemble average of theories with disorder~\cite{saad2018semiclassical,saad2019jt,saad2019late,stanford2019jt,bousso2020unitarity,bousso2020gravity,pollack2020eigenstate,marolf2020transcending,marolf2021observations,stanford2020more,cotler2021ads3} (see also~\cite{engelhardt2021free,liu2021entanglement,afkhami2021free,chen2021bra,belin2021random,maloney2020averaging,blommaert2020dissecting,giddings2020wormhole,van2021comments,eberhardt2021summing,saad2021wormholes,verlinde2021deconstructing,kudler2021relative}).  The non-orthogonality is crucial in obtaining the unitary Page curve from the path integral, as one can basically think of the phase transition occurs due to the accumulated effect of the non-orthogonality. Ignoring such deviations and simply computing the entropy for $\rho_B$ gives the paradoxical Hawking result~\cite{hawking1976breakdown}. It is thus also important to distinguish that $S(R)_\rho$ denotes the exact fine-grained entropy of the radiation computed from the replica trick, whereas $H(I\cup R)_\rho$ is computed in the semiclassical description \eqref{eq:entangled}, where we denote by $H$ to avoid any confusion.

When implementing the gravitational path integral for $\tr\rho_R^n$ in the semiclassical regime, we can assume that the integral reduces to a sum over saddle point geometries. This simple model has the advantage that we can do the sum explicitly without approximations, and we can extract the entire spectrum of $\rho_B$. Given the spectrum, there is no need to invoke the analytic continuation or AEP to compute the entanglement entropy. Note that because of the non-orthogonality, the spectrum is not the flat one we put in \eqref{eq:entangled}. Rather, we expect that the relevant spectrum is also random due to the disorder. Therefore, we should aim for the distribution density for the spectral values, and invoke tools from the random matrix theory to facilitate the calculation. We shall use the \emph{resolvent method}, as pioneered by PSSY in gravity~\cite{penington2019replica}. 
 
 The resolvent $k\times k$ matrix of the a positive semidefinite Hermitian operator $\rho$, parameterized by some complex variable $\lambda$, is defined as
\begin{equation}\label{eq:resolvent}
    \mathbf{R}(\lambda;\rho):=\left(\lambda\mathbf{I}_{k\times k}-\rho\right)^{-1}=\frac{\mathbf{I}_{k\times k}}{\lambda}+\sum_{n=1}^\infty \frac{\rho^n}{\lambda^{n+1}}
\end{equation}
where the second equality is an asymptotic expansion. The resolvent has poles at the spectrum of $\rho$ so we can think of it as a green function, and all the spectral functions, such as the entropies, can be evaluated using this green function. Explicitly, denoting the trace as $R(\lambda;\rho)=\tr\,\mathbf{R}(\lambda;\rho)$, we can extract the spectral density $D(\lambda;\rho)$ from the imaginary part of $R(\lambda;\rho)$ via the inverse Stieltjes transformation
\begin{equation}
    D(\lambda;\rho) = \lim_{\eps\rightarrow 0^+}\frac{R(\lambda-i\eps)-R(\lambda+i\eps)}{2\pi i}\,,
\end{equation}
and the density function is normalized
\begin{equation}
    \int_{\mathbb{R}^+}\dd\lambda\; D(\lambda;\rho)\lambda = 1, \quad \int_{\mathbb{R}^+}\dd\lambda\; D(\lambda;\rho) = k\,
\end{equation}
where the first equation says $\rho$ has unit trace and the second says they are $k$ spectral values.

The von Neumann entropy and the R\'enyi entropies are given by
\begin{equation}\label{formulaentropy}
   S(\rho) = -\int_{\mathbb{R}^+} \dd\lambda\; D(\lambda;\rho)\lambda\log\lambda , \quad S_n(\rho) = \frac{1}{1-n}\log\int_{\mathbb{R}^+} \dd\lambda\; D(\lambda;\rho)\lambda^n\,.
\end{equation}
The task is to compute the trace of the resolvent matrix in order to obtain the density function. This was done via a ``Feynman diagram" re-summation trick (cf. page 11 in~\cite{penington2019replica}), and one can obtain the following recurrence relation: 
\begin{equation}\label{eq:resum0}
    \lambda R_{ij}(\lambda;\rho) = \delta_{ij}+\sum_{n=1}^\infty \frac{Z_n R(\lambda;\rho)^{n-1}R_{ij}(\lambda;\rho)}{k^n Z_1^n}.
\end{equation}
The partition functions $Z_n$ can be worked out explicitly. In the case of a microcanonical ensemble, the result is very simple: $Z_n/Z_1^n=2^{-(n-1)S}$ where $S$ is the entropy of the microcanonical ensemble. One can take the trace, compute the geometric series and obtain a quadratic equation:
\beq
\lambda R(\lambda;\rho) = k+\frac{2^S R(\lambda;\rho)}{k 2^S - R(\lambda;\rho)}.
\eeq
Solving for $R(\lambda;\rho)$ and then $D(\lambda;\rho)$ yields the \emph{Marchenko–Pastur} distribution for the singular values of large random matrices~\cite{marvcenko1967distribution}. The page curve follows from evaluating \eqref{formulaentropy}, and there is a relatively sharp transition around $S\approx H(R)_\rho$ (cf. the blue curve in Figure~\ref{fig:pagecurve}.) Note that even for the microcanonical ensemble, the sum over saddles introduces an $\O(1)$ (upper bounded by one half) deviation near the transition. This is essentially the same phenomenon as in Page's original result using Haar random states~\cite{page1993average}. This effect is not associated with any parameter and negligible if we consider large entropy values. For more details on this model and the Page curve calculation, refer to section 2 in~\cite{penington2019replica}. 

\subsection{A black hole in superposition}
We'd like to show that the transition is not generally characterized by the von Neumann entropy but rather the one-shot entropies. For simplicity, we leave out the smoothing here, as the transition is sharp enough. To show this, the original setup studied in~\cite{penington2019replica} has to be modified because the bulk state there has a flat spectrum so that the min, max and von Neumann entropies all coincide and thus the refinement is not manifested. We therefore consider modifying the Schmidt coefficients $\{c_i\}_{i=1}^k$,
\begin{equation}\label{eq:generalspectrum}
    \ket{\rho}_{BR}:=\sum_{i=1}^k c_i\ket{\psi_i}_B\ket{i}_R.
\end{equation}
We shall focus on the microcanonical ensemble where the correction to the standard QES prescription is manifested the most in absence of energy fluctuation. For simplicity, we consider an ``L-shaped'' spectrum where the first $m$ coefficients share the same value $\{c_i\}_{i=1}^{m}=\sqrt{\frac{p}{m}}$ and the rest $k-m$ coefficients share the same value $\{c_i\}_{i=m+1}^k=\sqrt{\frac{1-p}{k-m}}$. We assume the coefficients are such that $\frac{p}{m}>\frac{1-p}{k-m}$. Physically, the state represents a black hole in a superposition of two stages of evaporation. In one branch, there has been $\log k$ amount of emitted radiation quanta, and $\log m$ amount of quanta in the other branch. Both branches share the same Bekenstein-Hawking entropy and  the same semiclassical geometry.

The calculation of the spectral density follows the same resolvent method and we aim to compute its trace.  Then the diagrammatic rules used to compute $R(\lambda;\rho)$ needs to be modified by the weights $\{c_i\}$. Because of the weights, it turns out that the following weighted sum
\begin{equation}\label{eq:R1}
    R^{(1)}(\lambda;\rho): = \sum_i^k c^2_i R_{ii}(\lambda;\rho)
\end{equation}
 will also be relevant in the re-summation. Following the same steps as in~\cite{penington2019replica} (cf. Appendix~\ref{sec:JTEOW} for the detailed calculation.), we obtain 
 \begin{equation}\label{eq:resum}
     \lambda R_{ij}(\lambda;\rho) = \delta_{ij}+\frac{c_i^2R_{ij}(\lambda;\rho)}{1-R^{(1)}2^{-S}}\,.
 \end{equation}
In contrast to \eqref{eq:resum0}, simply taking the trace of \eqref{eq:resum} now does not give us an equation of $R(\lambda;\rho)$ as before. We can proceed by considering another variable $R_{11}(\lambda;\rho)$ which is the first diagonal element of the resolvent matrix. Note that it shares the same value as the first $m$ diagonal elements according to \eqref{eq:resum}. Then we can obtain a system of three equations for three unknowns $R(\lambda;\rho),R^{(1)}(\lambda;\rho),R_{11}(\lambda;\rho)$,
\begin{equation}\label{eq:cubicequations}
    \begin{cases}
        \lambda R(\lambda;\rho) =& k+\frac{R^{(1)}(\lambda;\rho)}{1-R^{(1)}(\lambda;\rho)2^{-S}}\,,\\
        \lambda R_{11}(\lambda;\rho) =& 1+\frac{pR_{11}(\lambda;\rho)}{m-mR^{(1)}(\lambda;\rho)2^{-S}}\,,\\
        R^{(1)}(\lambda;\rho) =& \frac{1-p}{k-m}R(\lambda;\rho) + m\left(\frac{p}{m}-\frac{1-p}{k-m}\right)R_{11}(\lambda;\rho)\,
    \end{cases}
\end{equation}
where the first two equations follow from \eqref{eq:resum} and the last follows from the definition of $R^{(1)}(\lambda;\rho)$ \eqref{eq:R1}. We focus on solving $R(\lambda;\rho)$ and the equations reduce to a cubic equation for $R(\lambda;\rho)$ after eliminating $R^{(1)}(\lambda;\rho)$ and $R_{11}(\lambda;\rho)$. This is an elementary but still complicated equation to solve, so we omit here the general expression of the solution. We look for positive real domains of $\lambda$ such that the cubic discriminant is negative, corresponding to the part of of the spectrum with non-zero density. Then the cubic admits one real and two complex solutions conjugate to each other. The imaginary part gives us the spectral density. 
\begin{figure}
  \centering
\includegraphics[width=0.6\linewidth]{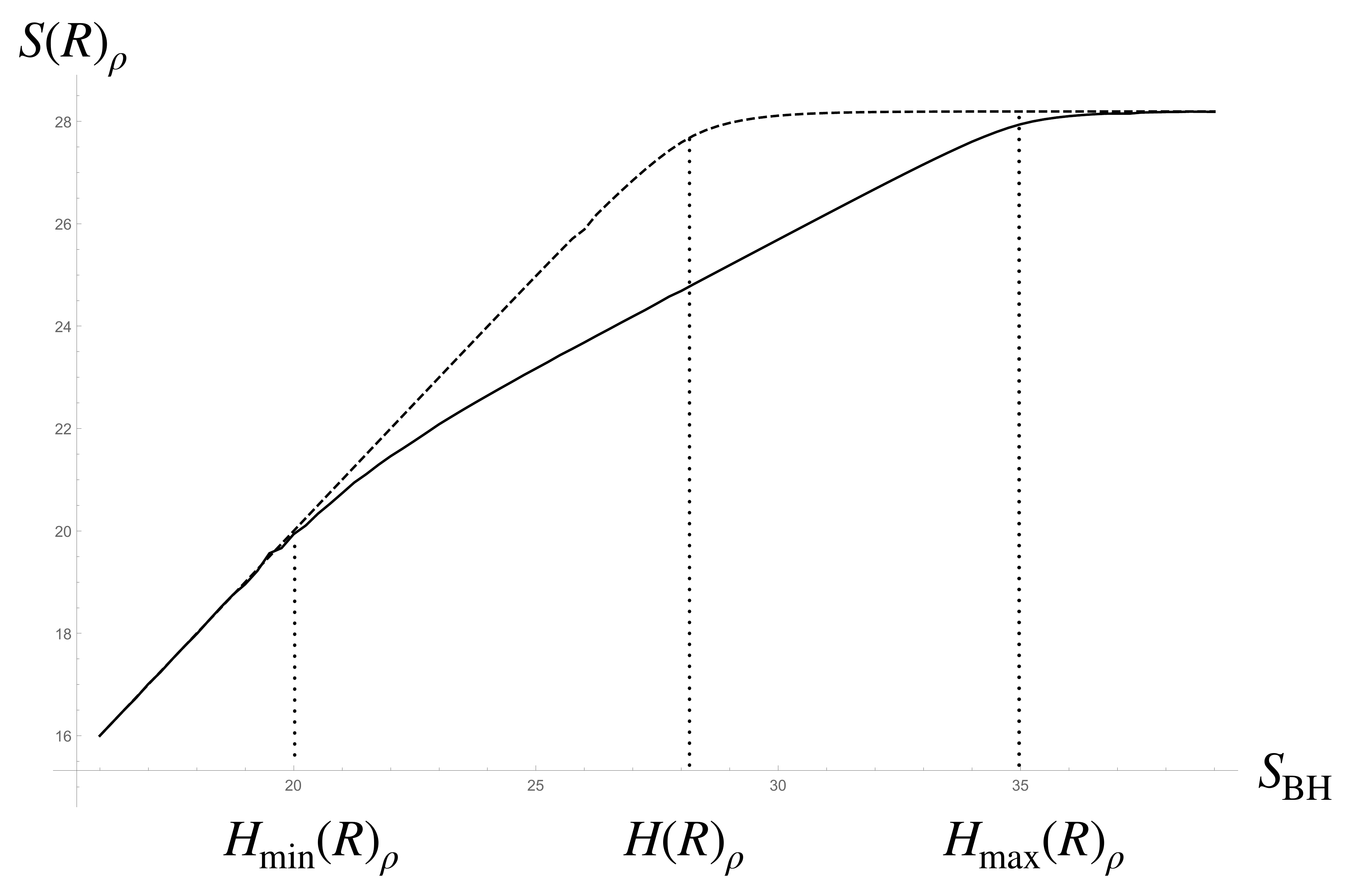}
  \caption{The Page curve we study here depicts the von Neumann entropy of the radiation $S(R)$ against different values of the entropy $S$ of the microcanonical ensemble. We plot an analytically evaluated Page curve (solid) against the prediction of the standard island rule (dashed, i.e. the scenario with a flat Schmidt spectrum). The resulting behaviour is consistent with the refined QES prescription, and here we clearly find three regimes among which two are well described by the RT formula. This example uses with $k=2^{35}$, $m=2^{20}$, $p=0.5$. }
  \label{fig:pagecurve}
\end{figure}

The spectrum is parameterized by the parameters $k,m,p,S$. We fix $k,m,p$ to be some constant ($k=2^{35}, m=2^{20}, p=0.5$) and compute the spectral density as we tune up the entropy of the microcanonical ensemble $S$.  Qualitatively, the density function has three characteristic regimes. For $S<H_{\mathrm{min}}(B)_\rho$, the density is approximately a Dirac delta peak centered at $\lambda=2^{-S}$. For $H_{\mathrm{min}}(B)_\rho\leq S\leq H_{\mathrm{max}}(B)_\rho$, the peak slits to two parts where one stays at $\frac{p}{m}$ and the other moves with $S$. For $H_{\mathrm{max}}(B)_\rho<S$, the density function has two frozen peaks at $\frac{p}{m}$ and $\frac{1-p}{k-m}$. The resulting von Neumann entropy $H$ has therefore three regimes consistent with the refined QES prescription, and we see a large deviation from the standard QES/island rule. 

We plot\footnote{Since the counterpart to the area difference here is $-S$ and conditional entropies are all negative because of the entanglement, we choose to present it in a more intuitive way by flipping every term to positive. For the entropies, we use the duality relation \eqref{eq:duality}, $H_{\mi/\ma}(R)_\rho=H_{\mi/\ma}(I)_\rho=-H_{\ma/\mi}(I|R)_\rho$, to label the transition points as entropies instead of the minus conditional entropies.} the von Neumann entropy of radiation $S(R)$ against different values for the entropy $S$ of the microcanonical ensemble. The result is depicted as the solid curve in Figure~\ref{fig:pagecurve} for some particular choices of the values $k,m,p$.  Since the model involves no dynamic evaporation anyway, we can think of it as a static Page curve capturing a collection of snapshots at different ratios between $S$ and $k$. It captures the same physics as the usual Page curve which varies $k$. 

The behaviour is not sensitive to the chosen parameters, and we can conclude that the Page curve is more consistently described by the refined QES prescription than the QES prescription (dashed curve). This example provides evidence that the entanglement entropy is generically over-estimated by the RT formula in the indefinite regime as expected. The small offsets at the min/max entropy values are negligible if we consider systems of large entropy. \\

Let us finish this section by a few remarks. Generally, given any spectrum $\{c_i\}_{i=1}^k$ with $l$ distinct values, we can associate $l+1$ resolvent variables with $l+1$ equations. The resulting equation for $R(\lambda;\rho)$ of interest would be a polynomial equation of degree $l+1$. We expect the resulting von Neumann entropy behaves the same way as predicted by the refined QES prescription. 

A similar correction in the canonical ensemble was shown in the original work~\cite{penington2019replica}. A more complicated analysis was carried out for a canonical ensemble at some fixed temperature $1/\beta$, and the result was that the Page curve features a smooth transition around a $1/\sqrt{\beta}$ window with a dent of size $ 1/\sqrt{\beta}$. This effect is parametrically subleading as compared to the correction we pointed out, which doesn't depend on the temperature even in a canonical ensemble.  In a canonical ensemble, the Page curve would behave in the same way as in Figure~\ref{fig:pagecurve}, up to the the further smearing at the transitions due to the energy fluctuation at finite temperature. 

A similar resolvent calculation is carried out by AP to analyze a dust ball in a \emph{mixed state} with also an $L$-shaped entanglement spectrum. The resulting spectrum density is very similar to what we find here. Though the conclusion is qualitatively the same, the Page curve we saw here \emph{cannot} be directly inferred from their dust ball example. The difference can be seen from comparing the resolvent equations in~\cite{akers2021leading,penington2019replica} (see also~\cite{marolf2020probing}) with ours, and they are not equivalent.  The reason is that for an coherent superposition between two branches with different amount of entangled pairs, it is a priori unclear how the entropy of the marginal state of the superposed state would be related to the ones of each branch. On the other hand, for an incoherent mixture of the same two branches, the entropy is well approximated (up to one bit of error) by the averaged entanglement entropy of the individual branch. Nonetheless, it can be shown that in the PSSY model, as far as the entanglement spectral density is concerned, the incoherent mixture can be a good approximation of the coherent superposition provided $m\ll k$. It is therefore not surprising that the Page curve we obtained here looks like the average between two standard Page curves that saturates at $S=k$ and $S=m$ respectively.  

\section{Beyond two quantum extremal surfaces}\label{sec:multiqes}

Now we return to AdS/CFT, and study what a general refined QES prescription for more than two QES' should be. We shall generalize our derivation in Section~\ref{sec:derivation} to apply it here.
\subsection{The refined multi-QES prescription}

We have studied a somewhat specific case there are two complementary boundary subregions $B,\overline{B}$ and two candidate QES $\gamma_1,\gamma_2$ that divide the bulk into three subregions $b,b',\b$. What is then the general QES prescription if we have multiple QES candidates associated with some arbitrary disconnected boundary subgregion $B$, such as in Figure~\ref{fig:3qes} ? AP also addresses the general case by defining the max-Entanglement Wedge ($\mathrm{EW}_\ma$) and min-Entanglement Wedge ($\mathrm{EW}_\mi$) respectively as
\begin{equation}\label{maxew}
    \mathrm{EW}_\ma:= \bigcup {\be\sim B}\,\, \mathrm{s.t.}\,\, \forall \be'\subset \be, H^\eps_\ma(\be\setminus \be'|\be')<\frac{A[\partial \be']-A[\partial \be]}{4G_N}
\end{equation}
where $\be\sim B$ means $\partial \be$ is homologous to $B$, and  $\be\setminus \be'$ is the complement of $\be'$ in $\be$. 

\begin{equation}\label{minew}
    \mathrm{EW}_\mi:= \bigcap {\be\sim B}\,\, \mathrm{s.t.}\,\, \forall \be'\subset \overline{\be}, H^\eps_\mi(\be'|\be)>\frac{A[\partial \be]-A[\partial( \be\cup\be')]}{4G_N}
\end{equation}
where $\overline{\be}$ denotes the complement of $\be$ in the bulk.

The existence of $\mathrm{EW}_{\mi/\ma}$ has been proved by AP, and also $\mathrm{EW}_\ma\subseteq\mathrm{EW}_\mi$ (cf. section 7 in AP). The refined QES prescription is then generally stated as: whenever the $\mathrm{EW}_{\mi}$ and $\mathrm{EW}_{\ma}$ coincide with each other and therefore can be identified as the entanglement wedge, the generalized entropy evaluated on its boundary gives the holographic entanglement entropy $S(B)$. Otherwise, it is in the indefinite regime and we cannot obtain the $S(B)$ straightforwardly, and we call the corresponding region $\mathrm{EW}_{\mi}\setminus\mathrm{EW}_{\ma}$ as \emph{``no man's land''}.\\

\begin{figure}
\centering
\begin{subfigure}{.5\textwidth}
  \centering
  \includegraphics[width=0.8\linewidth]{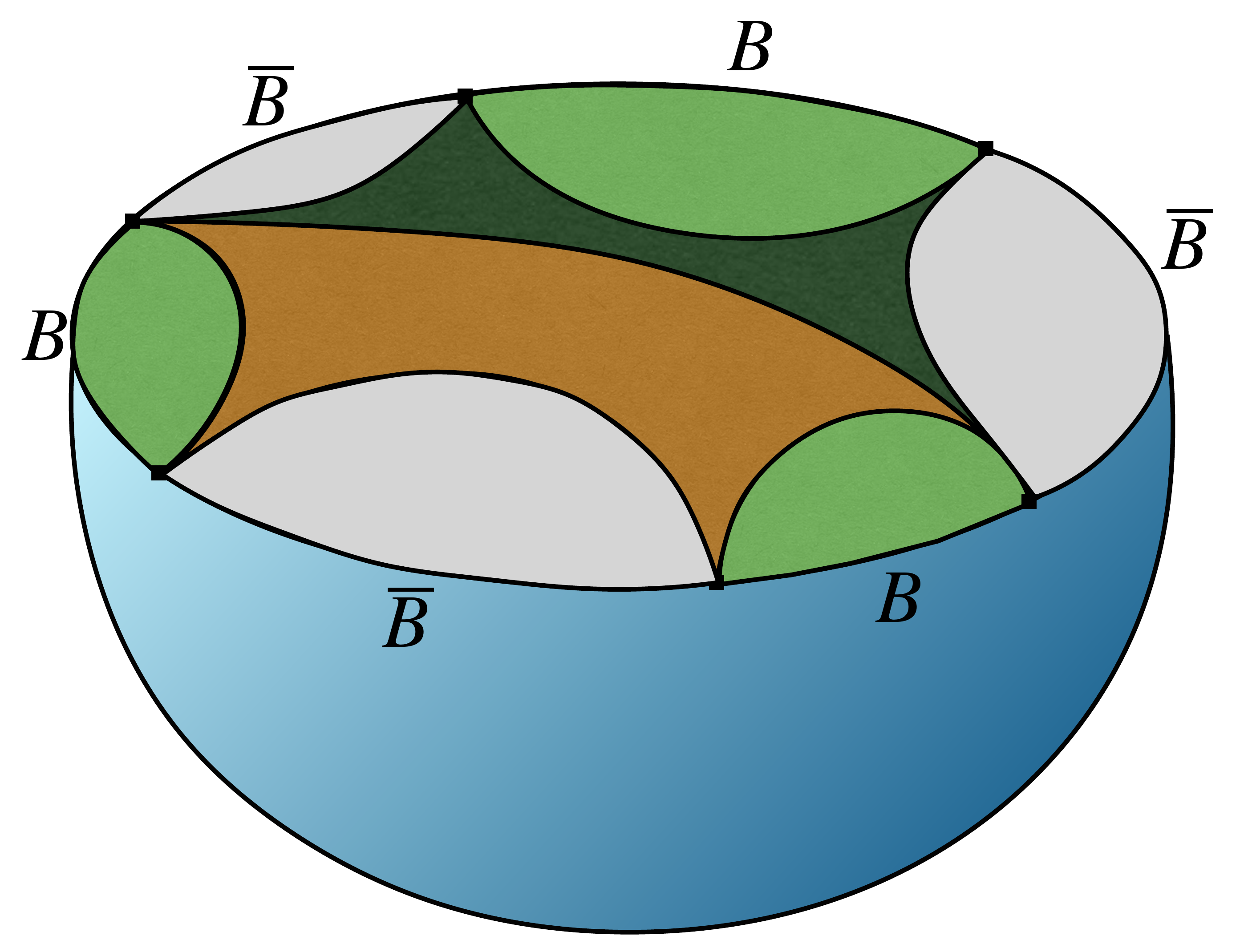}
\end{subfigure}%
\caption{{\bf A multi-QES scenario.} The figure depicts a typical multi-QES setup where the boundary consists of three disconnected components and we identify the bulk subregions as $b_0=$ green, $b_1=$ brown, $b_2=$ olive, $b_3=$ gray. }
\label{fig:3qes}
\end{figure}

Here we prove the validity of this proposal that AP left out. Since the above definitions of the $\mathrm{EW}_{\mi/\ma}$ entail checking all possible continuous variations of the bulk subregion, it is technically difficult to prove the general refined QES prescription via these definitions. Instead, we shall restrict the variational domain by only considering a finite set of non-overlapping QES candidates with fixed areas to simplify the problem. Furthermore, we assume the QES candidates are not crossing each other so we have a natural ordering with respect to the homology condition,\footnote{We can in principle lift this assumption and work with a partially order set of QES candidates in general. However, the derivation would be even more complicated and we shall leave it for simplicity.} as for example in Figure~\ref{fig:3qes}. Hence, given the collection of QES candidates $\Gamma:=\{\gamma_i\}_{i=0}^l$, with areas $\{A_i\}_{i=0}^l$ ($\gamma_l$ is set to be empty and thus $A_l=0$), homologous to $B$ and the associated homology regions $\Sigma:=\{\be_i\}_{i=0}^l$, we can then naturally order the set $\Sigma$ (and also $\Gamma$) as follows: $\be_i\subset \be_j \implies i<j$ and let $l:=|\Sigma|=|\Gamma|$. Then the bulk is divided into $l+1$ regions $b_0,b_1,b_2,\cdots,b_{l-1},b_l$, where
\begin{equation}
    b_i = \be_{i}\setminus \be_{i-1} ,\,\,\,\,\be_i = \bigcup_{j=0}^{i} b_j, \quad i\in\{0,1,2,\cdots,l\}.
\end{equation} 
The ``discretized" $\mathrm{EW}_{\mi/\ma}$ are defined as (cf. Figure~\ref{fig:EWs})
\begin{equation}
\begin{aligned}
    \widetilde{\mathrm{EW}}_\ma:=& \max_{j} \,\be_j\in \Sigma \,\, \mathrm{s.t.}\,\,  \forall \be_i\subset \be_j, H^\eps_\ma(\be_j\setminus \be_i|\be_i)<\frac{A_i-A_j}{4G_N}\,,\\
    \widetilde{\mathrm{EW}}_\mi:=& \min_{j}\,\be_j\in \Sigma \,\, \mathrm{s.t.}\,\, \forall \be_i\supset \be_j, H^\eps_\mi(\be_i\setminus \be_j|\be_j)>\frac{A_j-A_i}{4G_N}\,.
\end{aligned}
\end{equation} 

\begin{figure}[t]
\centering
\begin{subfigure}{.5\textwidth}
  \centering
  \includegraphics[width=0.8\linewidth]{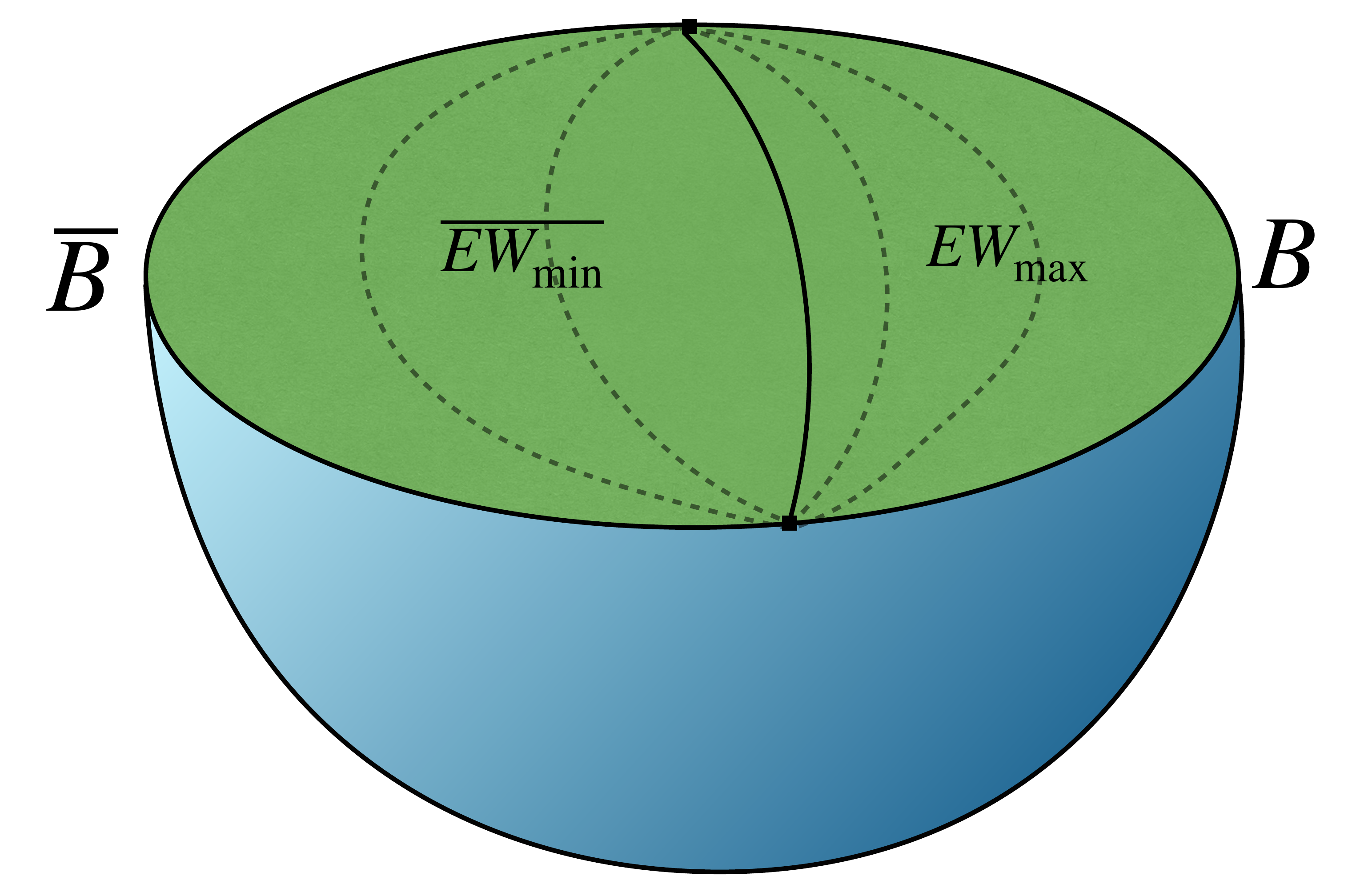}
  \caption{\centering $\text{EW}_\mi=\text{EW}_\ma=\text{EW}$.}
  \label{fig:EW1}
\end{subfigure}%
\begin{subfigure}{.5\textwidth}
  \centering
  \includegraphics[width=0.8\linewidth]{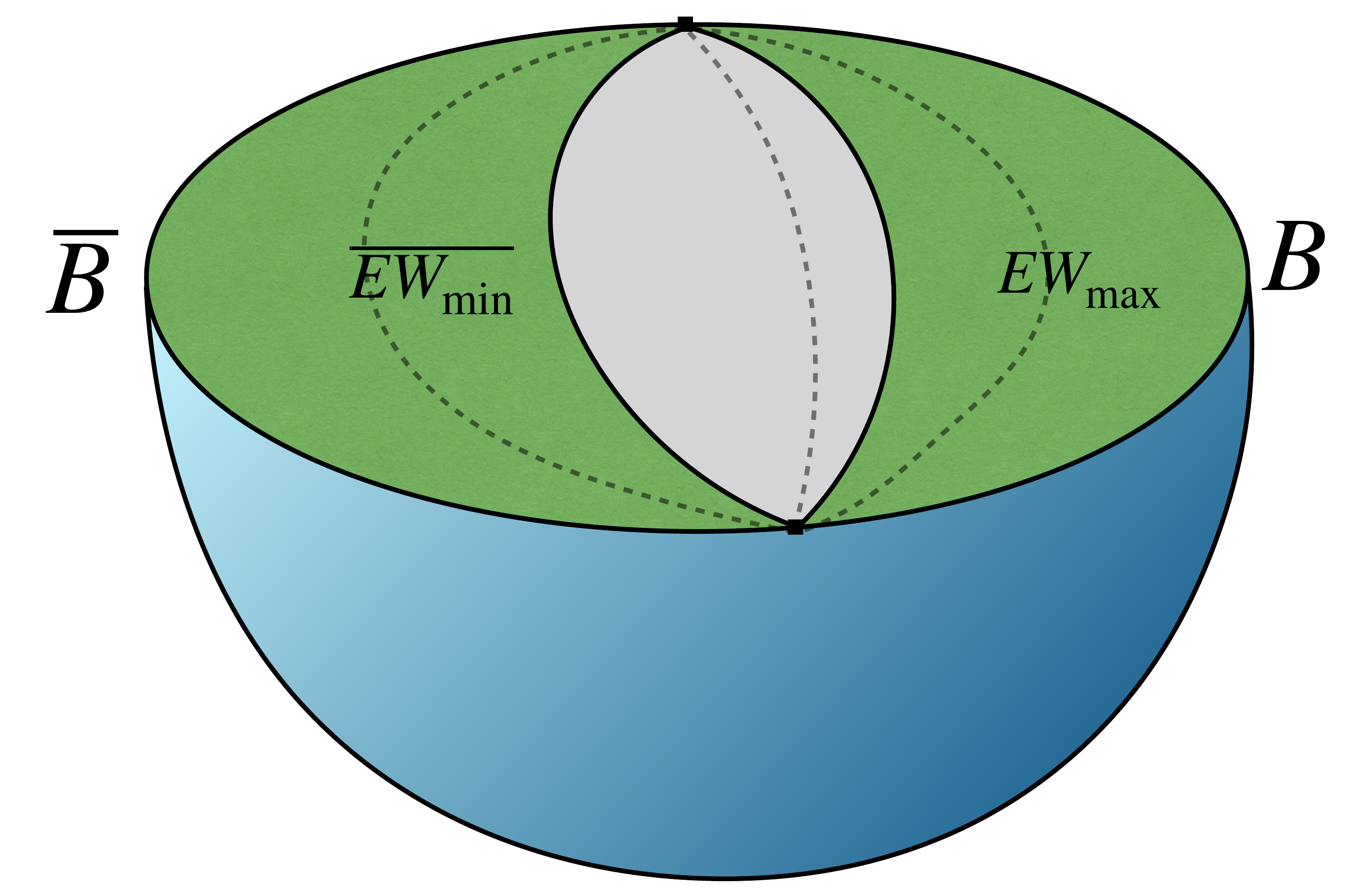}
  \caption{\centering Indefinite regime.}
  \label{fig:EW2}
\end{subfigure}
\caption{{\bf $\text{EW}_\ma$ and $\text{EW}_\mi$.} The green zones depict the $\text{EW}_\ma$ and the complement of $\text{EW}_\mi$. The left figure (a) shows the regime where $\text{EW}_\mi$ and $\text{EW}_\ma$ match, so the task of entanglement wedge reconstruction is achievable and the entanglement entropy $S(B)$ can be identified as the generalized entropy evaluated on the shared QES surface in black; and the right figure (b) depicts the indefinite regime where there is a grey ``no man's land'' region that cannot be reconstructed from either $B$ or $\overline{B}$ alone. }
\label{fig:EWs}
\end{figure}
Then the refined QES prescription claims that \\

\emph{if $\widetilde{\mathrm{EW}}_\ma=\widetilde{\mathrm{EW}}_\mi=\be_k$ for some $k$, then the holographic entanglement entropy is given by $S(B)_\rho=A[\gamma_k]/4G_N + H(\be_k)_\rho$; and is indefinite otherwise.} \\

It translates to a more direct statement as follows: 
\begin{equation}\label{generalqes}
 \mathrm{If}\quad    H^\eps_\ma\left(\bigcup_{i=j+1}^k b_i\Bigg\vert\bigcup_{i=0}^j b_i\right)_\rho < \frac{A_j-A_k}{4G_N},\,\forall j<k \quad\land\quad H^\eps_\mi\left(\bigcup_{i=k+1}^j b_i\Bigg\vert\bigcup_{i=0}^k b_i\right)_\rho > \frac{A_k-A_j}{4G_N},\, \forall j>k\,,
\end{equation}
then $S(B)=A_k/4G_N + H(\be_k)_\rho$\, . \\

Let's use the two-QES setting in Section~\ref{sec:derivation} as an example to illustrate the general refined QES prescription \eqref{generalqes} as stated above. We had two candidate QES', so $l=2$. We can identify the regions as $b_0 = b, b_1 = b', b_2 = \b$, and the areas as $A_0=A_2, A_1=A_1$.  When
\begin{equation}
    H^\eps_\mi(b_1|b_0)_\rho=H^\eps_\mi(b'|b)_\rho>\frac{A_0-A_1}{4G_N}\,,
\end{equation}
that is the second condition in \eqref{generalqes} is satisfied for $k=0$ and the first is trivially satisfied when $k=0$, then we have $S(B)_\rho=A_0/4G_N + H(b)_\rho$. On the other hand, when
\begin{equation}
    H^\eps_\ma(b_1|b_0)_\rho = H^\eps_\ma(b'|b)_\rho<\frac{A_0-A_1}{4G_N}\,,
\end{equation}
that is the first condition in \eqref{generalqes} is satisfied for $k=1$ and the second is trivially satisfied when $k=1$, then we have $S(B)_\rho=A_1/4G_N + H(bb')_\rho$. They are consistent with \eqref{eq:refined_QES}. We shall forestall some potential confusion by noting the following. It appears that each condition in \eqref{eq:refined_QES} corresponds to one regime in our two-QES example, but in fact both conditions in \eqref{eq:refined_QES} are needed for each regime as one is trivially satisfied. In the previous section, we've used two different states for two different QES regimes, whereas generally both the min-entropy and max-entropy conditions in \eqref{generalqes} are needed to single out one QES $\gamma_k$.

\subsection{Derivation}

We shall see that a generalization of the argument in Section~\ref{sec:derivation} leads us to the QES prescription for fixed area states with multiple QES candidates. Although the generalization is straightforward in principle, the actual implementation is more complicated. We shall therefore treat some parts of the derivation with extra care as we proceed. Before we do the AEP replica trick, let's first see what we can obtain from the usual replica trick to compute the R\'enyi entropy of fixed area states. The partition function reads
\begin{equation}
    \frac{Z_n[\{A_i\}^l_0]}{Z_1[\{A_i\}^l_0]^n} = \sum_{\{\pi_i\}\in S_n}2^{\sum_{i=0}^{l}\left(|\pi_{i+1}^{-1}\circ \pi_i|-n\right)\frac{A_i}{4G_N}}\cdot\tr\left(\rho^{\otimes n}_{b_0\cdots b_l}\cdot\bigotimes_{i=0}^l U_{\pi_i|b_i}\right)\,.
\end{equation}@
Under the replica symmetry assumption that is partially justified in Section~\ref{sec:holorenyi}, we only need to keep the dominant contributions from $\pi_{i} = \tau$ or $\mathbf{1}$, which together count $2^l$ terms. Each term corresponds to tracing out some region from the bulk state, with the exponential prefactor consists of different areas. We can ignore terms with more than one area in the sum, leaving only $l$ terms corresponding to $l$ distinct areas,
\begin{equation}\label{dominant}
    \frac{Z_n[\{A_i\}^l_0]}{Z_1[\{A_i\}^l_0]^n} = \sum_{i=0}^{l}2^{(1-n)\frac{A_i}{4G_N}}\cdot\tr\left(\rho^{\otimes n}_{b_0\cdots b_i}\right)\,.
\end{equation}
Similar to \eqref{upperbound}, the standard replica trick for computing the von Neumann entropy via analytic continuation tells us that $S(B)$ is upper-bounded by
\begin{equation}\label{generalupperbound}
    S(B)_\rho\leq \min_i\{A_i/4G_N+ H(b_i)_\rho\}.\\
\end{equation}

\begin{figure}
\centering
\begin{subfigure}{.58\textwidth}
  \centering
  \includegraphics[width=1\linewidth]{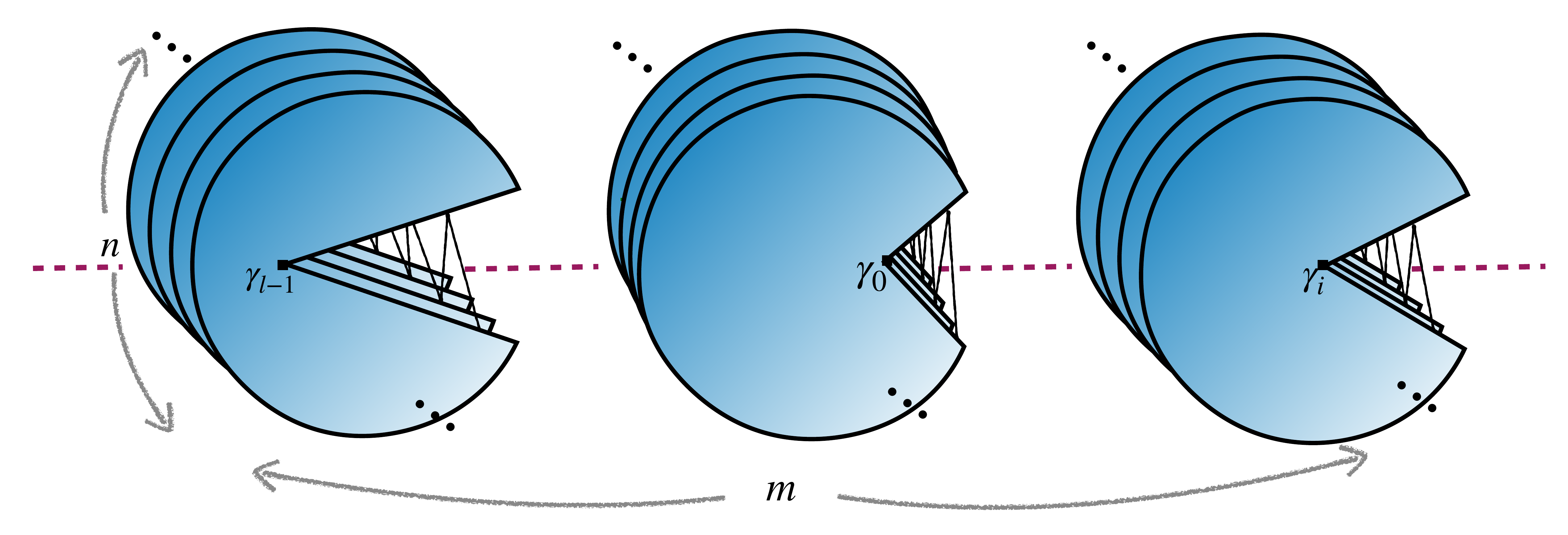}
   \caption{\centering A generic term in the $n\times m$-replica path integral}
  \label{fig:packmank}
\end{subfigure}%
\begin{subfigure}{.42\textwidth}
  \centering
  \includegraphics[width=1\linewidth]{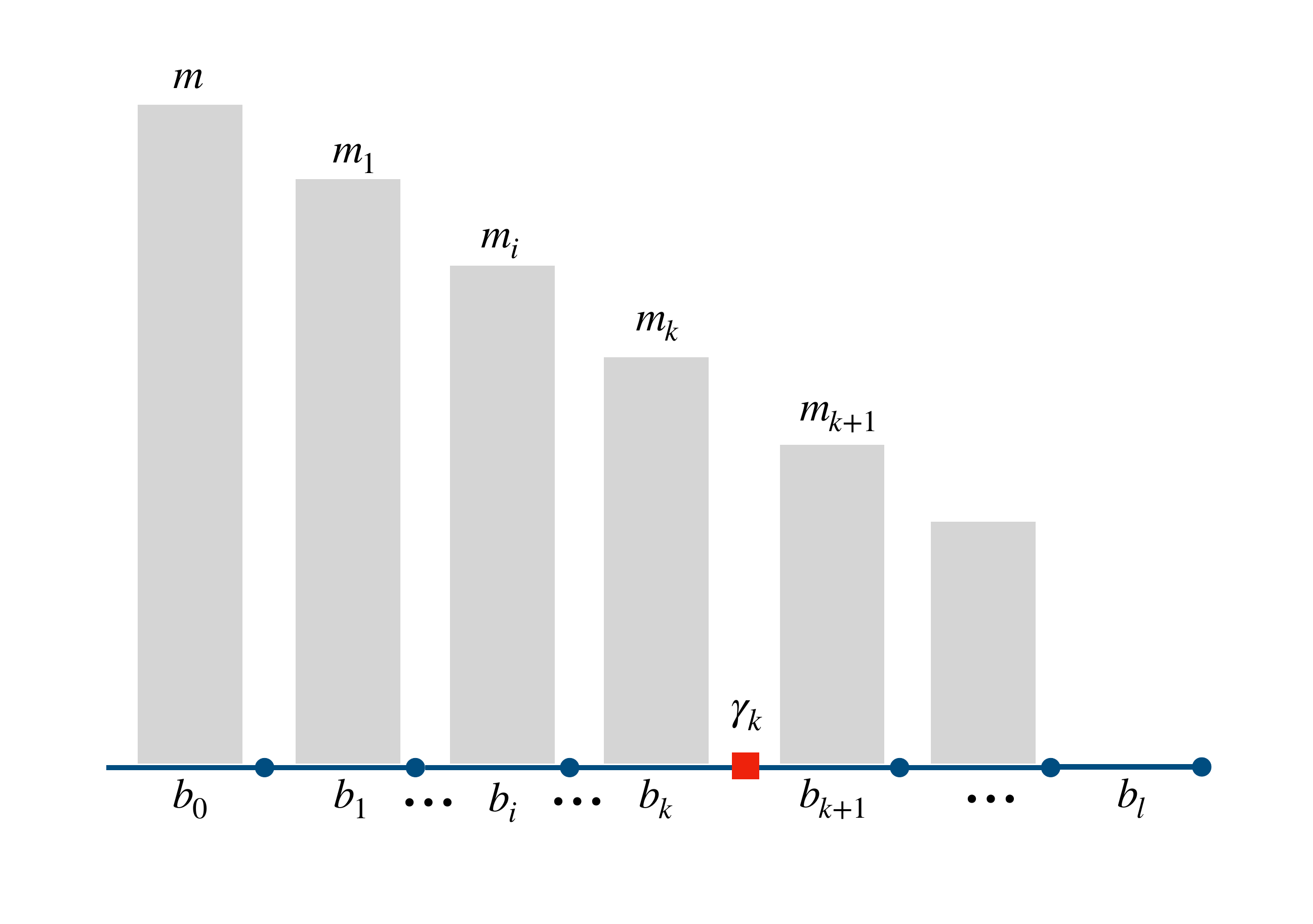}
  \caption{\centering Distribution of the $m$-weights}
  \label{fig:histoterm}
\end{subfigure}
\caption{{\bf The AEP replica trick for multi-QES and $m$-weights. } The $n\times m$-replica path integral computes $Z_{n,m}$ with multiple $\gamma$'s as an expansion over terms, and each such term can be evaluated using the replicas arranged and contracted as depicted in (a). The histogram (b) shows how many terms are cyclically contracted among the $n$-replicas in each family of $m$ families of replicas.}
  \label{fig:generalterm}
\end{figure}

The AEP replica trick is then basically raising the above expression to power $m$ while keeping the correlation in the chosen bulk state.  Now we set up $n\times m$ replicas and expand the partition function for a chosen feasible state $\rho^*_{b_0^m\cdots b_l^m}$, that is $\eps$-close to $\rho_{b_0\cdots b_l}^{\otimes m}$ and $m$-replica-symmetric, as
\begin{equation}\label{partitionfunction2}
\begin{aligned}
    \frac{Z_{n,m}[\{A_i\}^l_0]}{Z_1[\{A_i\}^l_0]^{nm}} = \sum_{\{\pi^{(j)}_{b_i}\} \in S_n}2^{\sum_{i=0}^{l}\left(\sum_{j=1}^m |(\pi^j_{i+1})^{-1}\circ \pi^j_i|-n\right)\frac{A_i}{4G_N}}\cdot\tr\left(\rho^{*\otimes n}_{b_0^m\cdots b_l^m}\cdot\bigotimes_{i=0}^l\bigotimes_{j=1}^m U_{\pi^j_i|b^{(j)}_i}\right)
\end{aligned}
\end{equation}
where we denote the $i^\mathrm{th}$ (among $l$) bulk region on the $j^\mathrm{th}$ (among $m$) replica as $b^{(j)}_i$, and the corresponding permutation that acts on it $\pi^j_i$, and let $\pi^j_0=\tau, \pi^j_l=\mathbf{1}, \,\forall j$. This equation generalizes~\eqref{partitionfunction}.\\

Keeping only the dominant contributions as in \eqref{dominant} under the replica symmetry assumption, the terms in the expansion have the form (cf. Figure~\ref{fig:packmank}),
\begin{equation}\label{terms2}
2^{(1-n)\sum_{i=0}^{l}(m_i - m_{i+1})\frac{A_i}{4G_N}}\cdot\tr\left(\rho^*_{b^{m_0}_0\cdots b^{m_i}_i\cdots b^{m_{l}}_{l}}\right)^n\,.
\end{equation}
where the descending array $\{m_i\}_{i=0}^{l}$ satisfy 
\beq
m=m_0\geq m_i\geq m_j\geq m_l=0, \forall i < j\,,
\eeq  
and we henceforth refer them as the \emph{$m$-weights} (cf. Figure~\ref{fig:histoterm}).   The expansion sums over all distinct $\{m_i\}$'s and different collections of marginal systems for a given $\{m_i\}$. Again, we shall not distinguish them in the notations as the details are irrelevant for the derivation, but we should bear in mind that $\rho^*_{b^{m_0}_0\cdots b^{m_i}_i\cdots b^{m_{l}}_{l}}$ represents one particular instance on the selected number of subsystems.

\begin{figure}
\centering
\begin{subfigure}{.5\textwidth}
  \centering
  \includegraphics[width=1\linewidth]{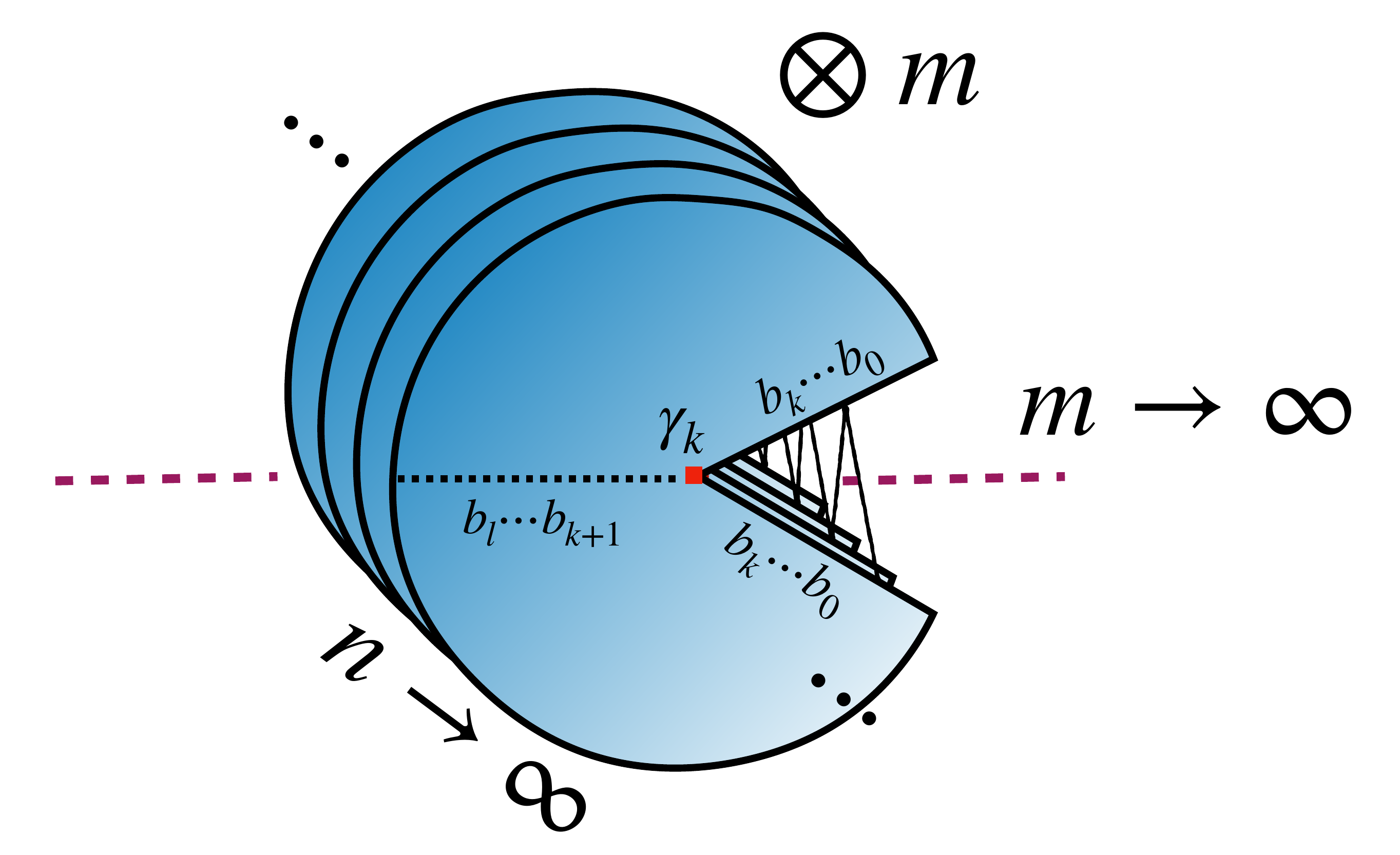}
\end{subfigure}%
\begin{subfigure}{.5\textwidth}
  \centering
  \includegraphics[width=1\linewidth]{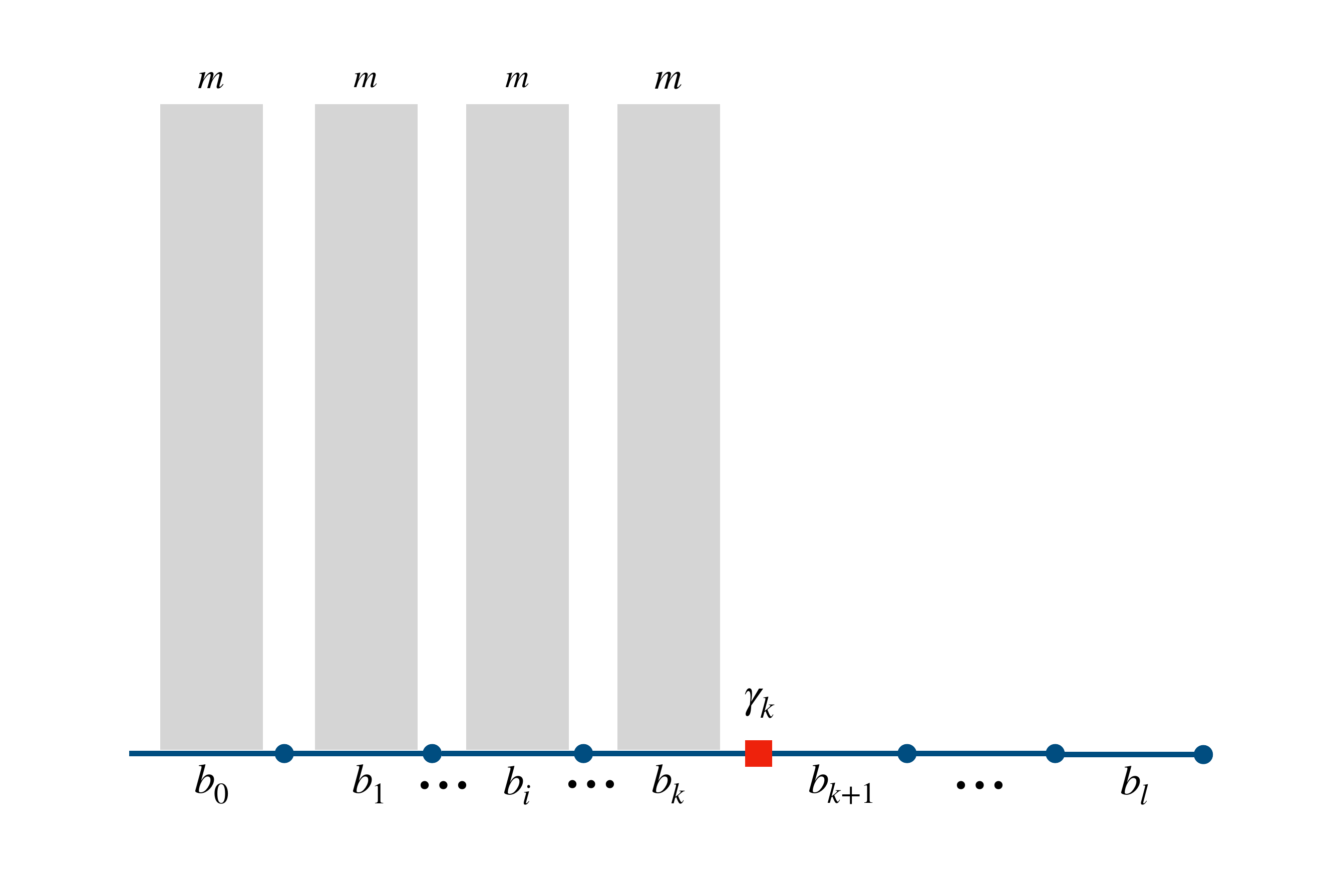}
\end{subfigure}
\caption{{\bf The dominant saddle.} The left figure depicts the dominant term when \eqref{generalqes} holds, and the right histogram shows its corresponding $m$-weights. }
\label{fig:dominantmulti}
\end{figure}

Consider some $\gamma_k$ with area $A_k$ and the associated $\be_k = b_0\cdots b_k$. We'd like to show that the particular term, with $\{m_i\}_{i=0}^k=m\,\, \mathrm{and}\,\, \{m_i\}_{i=k+1}^l=0$, dominates the partition function \eqref{partitionfunction2} if \eqref{generalqes} holds.  Previously in the two QES case, we can consider different states for regime 1 and regime 3 respectively, and in each case one of the two conditions in \eqref{generalqes} is trivially satisfied. One complication we shall deal with is that here for some $\gamma_k$, we need to construct a \emph{single} state  $\rho^*_{b^m_0\cdots b^m_l}$ that satisfies both conditions in \eqref{generalqes}, and argue the following term associated with marginal state $\rho^*_{b^m_0\cdots b^m_k}$ dominates, as depicted in Figure~\ref{fig:dominantmulti}.
\begin{equation}\label{keyterm}
    2^{(1-n)m\frac{A_k}{4G_N}}\tr\left(\rho^*_{b^{m}_0\cdots b^{m}_k}\right)^n.
\end{equation}

In order to find such a feasible state, we construct a particular bulk state $\rho^*_{b_0^m\cdots b_l^m}$ for any integer $m$, $\eps$-close to $\rho^{\otimes m}_{b_0\cdots b_l}$, such that it satisfies the two chain rules. The explicit construction is left in Section~\ref{sec:proof2}. The first chain rule, which concerns terms in \eqref{partitionfunction2} with $\{m_i\}_{i=k+1}^l=0$, states that for any fixed $0<\eps''<\eps$ and any $\{m_i\}_{i=0}^k$ with $m_k\leq\cdots\leq m_0=m$,
\begin{equation}\label{superchain1}
  H_\mi(b^m_0\cdots b^m_k)_{\rho^*} \leq H_\ma(b^{m-m_0}_0\cdots b^{m-m_k}_k|b^{m_0}_0\cdots b^{m_k}_k)_{\rho^{\otimes m}}+  H_\mi(b^{m_0}_0\cdots b^{m_k}_k)_{\rho^*}\,.
\end{equation}
and 
\begin{equation}\label{maximizer}
    H^{\eps''}_\mi(b^m_0\cdots b^m_k)_{\rho^{\otimes m}}=H_\mi(b^m_0\cdots b^m_k)_{\rho^*}.
\end{equation}

Our bulk state $\rho^*_{b^m_0\cdots b^m_l}$ also satisfies that for any $m$-weights $\{m_i\}_{i=0}^l$,
\begin{equation}\label{superchain2}
\begin{aligned}
    H_\mi(b^{m_0}_0\cdots b^{m_{l}}_{l})_{\rho^*} \geq H_\mi(b^{m_{k+1}}_{k+1}\cdots b^{m_{l}}_{l}|b^{m_0}_0\cdots b^{m_k}_k)_{\rho^{\otimes m}}+  H_\mi(b^{m_0}_0\cdots b^{m_k}_k)_{\rho^*}.
\end{aligned}
\end{equation}

Here comes the complication that we did not encounter for the two-QES case. Our goal is to argue that \eqref{keyterm}, corresponding to the bulk state $\rho^*_{b^{m}_0\cdots b^{m}_k}$, dominates over the contribution from any other marginal state of the form $\rho^*_{b^{m_0}_0\cdots b^{m_l}_l}$ with arbitrary $\{m_i\}_{i=0}^l$. Note that the latter is \emph{not} a marginal state of former as in the derivation for the two-QES setup. Consequentially, we cannot hope to establish the inequality among them using one chain rule.  The strategy is to break the problem into two parts by considering a set of intermediate states $\rho^*_{b^{m_0}_0\cdots b^{m_k}_k}$ with arbitrary $\{m_i\}_{i=0}^k$, which is the marginal state of $\rho^*_{b^{m}_0\cdots b^{m}_k}$. Then we compare it with the desired term \eqref{keyterm} using the first chain rule \eqref{superchain1}, and also with any other marginal states that shares the same $\{m_i\}_{i=0}^k$ using the second chain rule \eqref{superchain2}. This procedure is illustrated in Figure~\ref{fig:histoblocks}.  

\begin{figure}[t] 
  \centering
\includegraphics[width=0.6\linewidth]{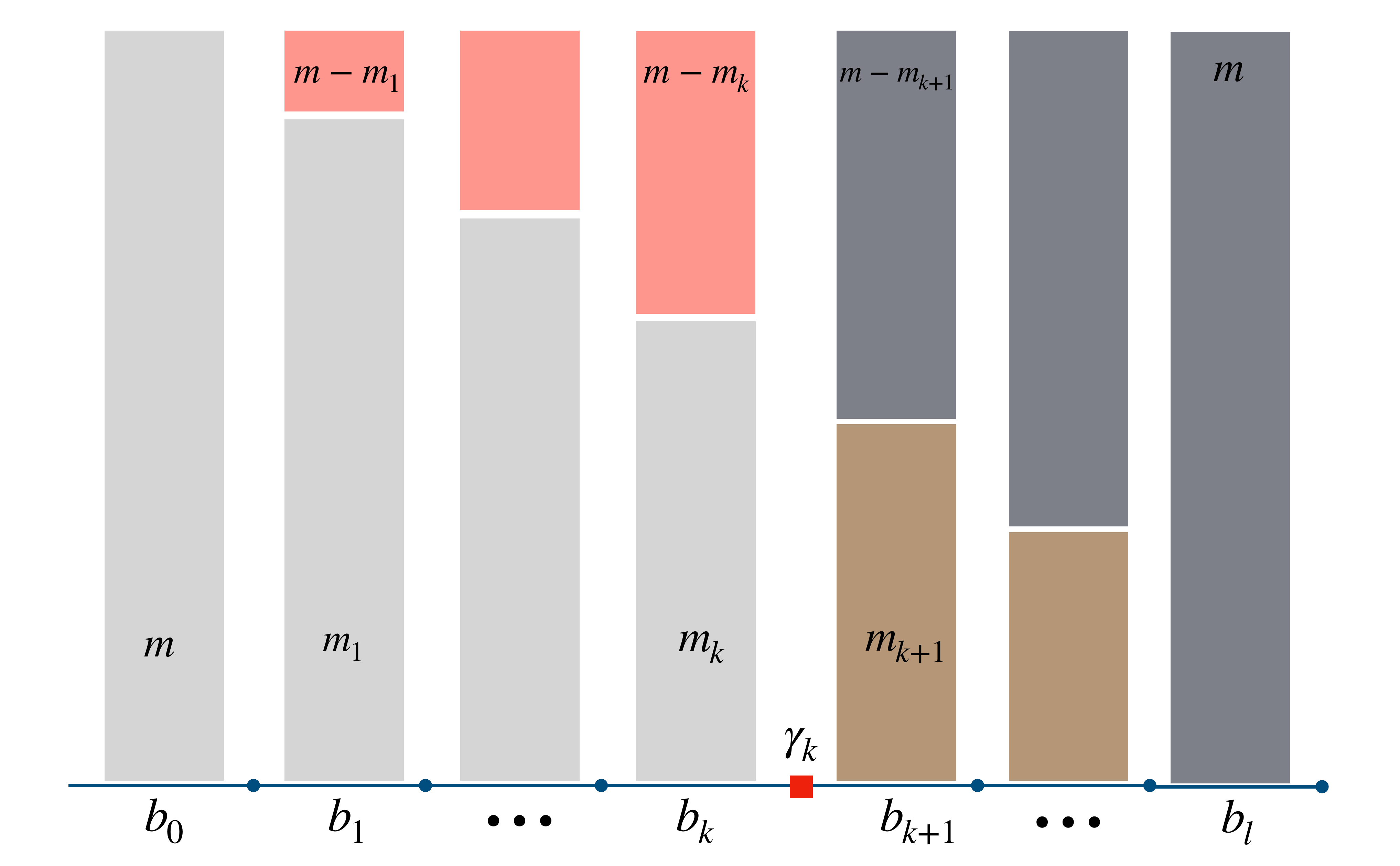}
  \caption{{Illustration of the proof strategy.} The histogram of $m$-weights helps to illustrate our proof strategy. We first show that the contributions from a generic term (brown+white) is dominated by the contribution from the white terms, and then we show that the latter is dominated by the dominant term \eqref{keyterm} represented by white+pink.}
  \label{fig:histoblocks}
\end{figure}

Let us start with the first chain rule. Given the conditional max-entropy condition in \eqref{generalqes},
\begin{equation}
    H_\ma(b_{i+1}\cdots b_k|b_0\cdots b_i)_\rho< \frac{A_i-A_k}{4G_N}\,,\quad\forall i<k\,,
\end{equation}

What we need to do is to break down the conditional max-entropy term on the composite system to parts which we can lower bound using the conditions in \eqref{generalqes}. The steps are better illustrated diagrammatically in Figure~\ref{fig:historeduction}. We have
\begin{equation}\label{longdivision1}
    \begin{aligned}
         &H_\ma(b^{m-m_0}_0\cdots b^{m-m_k}_k|b^{m_0}_0\cdots b^{m_k}_k)_{\rho^{\otimes m}}\\ 
        =&H_\ma(b^{m-m_1}_1\cdots b^{m-m_k}_k|b^{m-m_k}_0\cdots b^{m_{k-2}-m_k}_{k-2}b^{m_{k-1}-m_k}_{k-1})_{\rho^{\otimes m}} \\
        =&H_\ma(b^{m_{k-1}-m_k}_k|b^{m_{k-1}-m_k}_0\cdots b^{m_{k-1}-m_k}_{k-2}b^{m_{k-1}-m_k}_{k-1})_{\rho^{\otimes m}} \\
        &+H_\ma(b^{m-m_1}_1\cdots b^{m-m_{k-1}}_k|b^{m-m_{k-1}}_0\cdots b^{m_{k-2}-m_{k-1}}_{k-2})_{\rho^{\otimes m}}\\
        =&H_\ma(b^{m_{k-1}-m_k}_k|b^{m_{k-1}-m_k}_0\cdots b^{m_{k-1}-m_k}_{k-2}b^{m_{k-1}-m_k}_{k-1})_{\rho^{\otimes m}} \\
        &+H_\ma(b^{m_{k-2}-m_{k-1}}_{k-1} b^{m_{k-2}-m_{k-1}}_k|b^{m_{k-2}-m_{k-1}}_0\cdots b^{m_{k-2}-m_{k-1}}_{k-2})_{\rho^{\otimes m}}\\
        &+H_\ma(b^{m-m_1}_1\cdots b^{m-m_{k-2}}_{k-1}b^{m-m_{k-2}}_k|b^{m-m_{k-2}}_0\cdots b^{m_{k-3}-m_{k-2}}_{k-3})_{\rho^{\otimes m}}\\
        =&\cdots\\
        =&\sum_{i=1}^{k} H_\ma(b^{m_{i-1}-m_{i}}_{i}\cdots b^{m_{i-1}-m_{1}}_{k} | b^{m_{i-1}-m_{i}}_{0}\cdots b^{m_{i-1}-m_{i}}_{i-1})_{\rho^{\otimes m}}\\
        =&\sum_{i=1}^{k} (m_{i-1}-m_{i}) H_\ma(b_{i}\cdots b_{k} | b_{0}\cdots b_{i-1})_{\rho}\\
        <&\sum_{i=1}^{k} (m_{i-1}-m_{i})\frac{A_{i-1}-A_k}{4G_N}\,
    \end{aligned}
\end{equation}
where we applied \eqref{generalqes} in the last step.

Therefore, for large enough $n$,
\begin{equation}\label{dominance1}
\begin{aligned}
    \tr\left(\rho^*_{b^m_0\cdots b^m_k}\right)^n >& 2^{(1-n)\sum_{i=1}^{k}(m_{i-1}-m_{i})\frac{A_{i-1}-A_k}{4G_N}}\tr\left(\rho^*_{b^{m_0}_0\cdots b^{m_k}_k}\right)^n\,,\\
    2^{(1-n)m\frac{A_k}{4G_N}}\tr\left(\rho^*_{b^m_0\cdots b^m_k}\right)^n >& 2^{(1-n)\sum_{i=1}^{k}(m_{i-1}-m_{i})\frac{A_{i-1}}{4G_N}+ (1-n)m_k\frac{A_k}{4G_N}} \tr\left(\rho^*_{b^{m}_0\cdots b^{m_k}_k}\right)^n.
\end{aligned}
\end{equation}
It implies that the desired term \eqref{keyterm} with $b^{m}_0\cdots b^{m}_k$ dominates over these $b^{m_0}_0\cdots b^{m_k}_k$ terms when the bulk state is $\rho^*_{b^m_0\cdots b^m_l}$. These don't yet include all possible terms in the expansion so we need to make use of the min-entropy condition in \eqref{generalqes}. \\

\begin{figure}[t] 
  \centering
\includegraphics[width=0.78\linewidth]{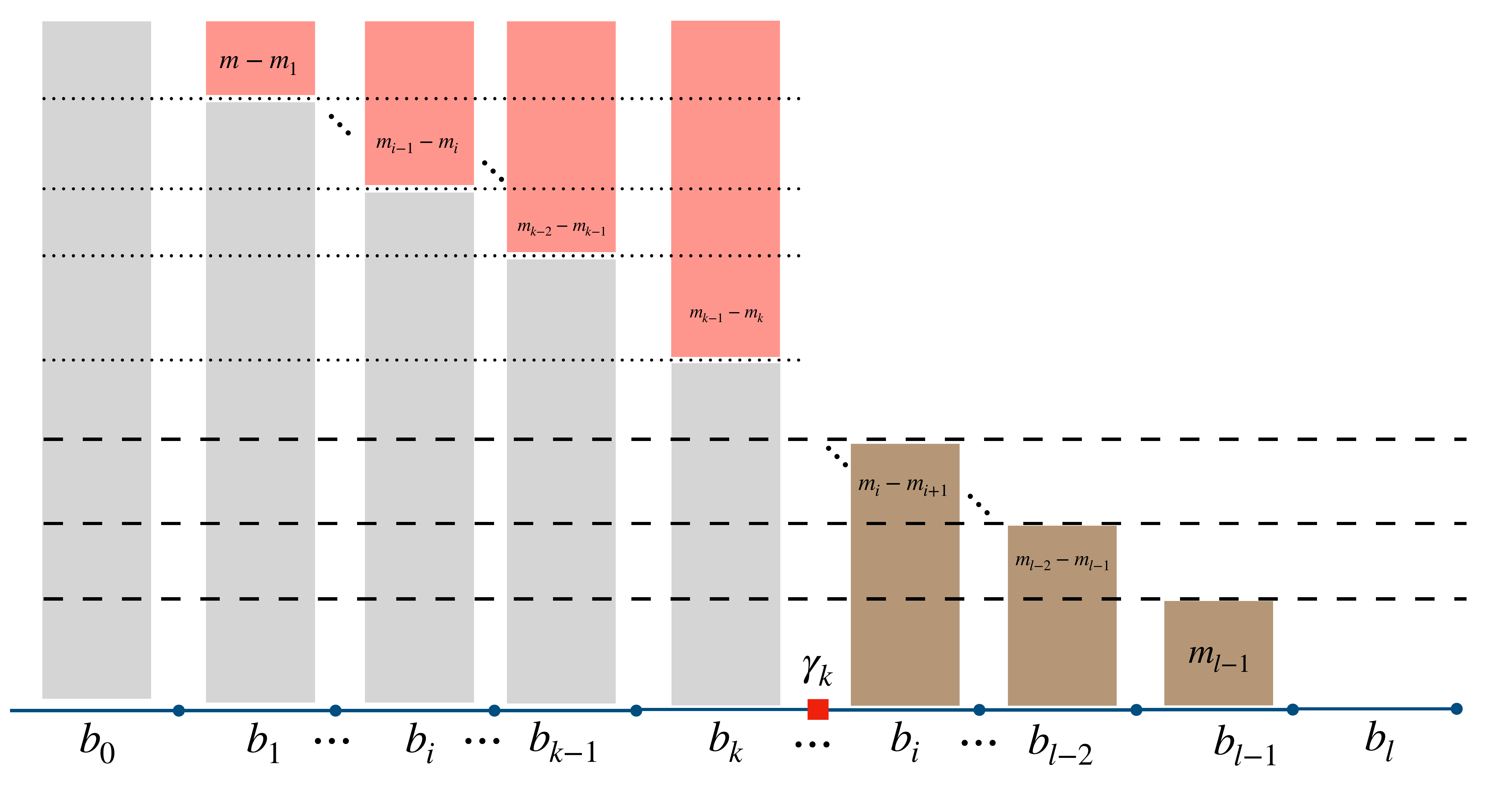}
  \caption{{\bf Illustrations of \eqref{longdivision1} and \eqref{longdivision2}.} The key fact we are using is that there are only correlations among the systems ``horizontally'' but not ``vertically''.  The conditional max-entropy of the pink conditioned on the white \eqref{longdivision1} can be reduced to a sum of the conditional max-entropies of each horizontal block divided by the dotted lines. The conditional min-entropy of the brown conditioned on the white \eqref{longdivision2} can be reduced to a sum of the conditional min-entropies of each horizontal block divided by the dashed lines.}
  \label{fig:historeduction}
  \end{figure}
Consider now the second chain rule \eqref{superchain2}. The conditional min-entropy condition in \eqref{generalqes} gives 
\begin{equation}
   H_\mi(b_{k+1}\cdots b_{j}|b_0\cdots b_k )_\rho>\frac{A_k-A_j}{4G_N}\,,\quad\forall j>k\,.
\end{equation}
Similar to \eqref{longdivision1}, we obtain
\begin{equation}\label{longdivision2}
\begin{aligned}
     &H_\mi(b^{m_{k+1}}_{k+1}\cdots b^{m_{l}}_{l}|b^{m_0}_0\cdots b^{m_k}_k)_{\rho^{\otimes m}}=H_\mi(b^{m_{k+1}}_{k+1}\cdots b^{m_{l-1}}_{l-1}|b^{m_0}_0\cdots b^{m_k}_k)_{\rho^{\otimes m}}\\
    =& H_\mi(b^{m_{k+1}-m_{l-1}}_{k+1}\cdots b^{m_{l-2}-m_{l-1}}_{l-2}|b^{m_0-m_{l-1}}_0\cdots b^{m_k-m_{l-1}}_k)_{\rho^{\otimes m}}+ H_\mi(b^{m_{l-1}}_{k+1}\cdots b^{m_{l-1}}_{l-1}|b^{m_{l-1}}_0\cdots b^{m_{l-1}}_k)_{\rho^{\otimes m}}\\
    =& H_\mi(b^{m_{k+1}-m_{l-2}}_{k+1}\cdots b^{m_{l-3}-m_{l-2}}_{l-3}|b^{m_0-m_{l-2}}_0\cdots b^{m_k-m_{l-2}}_k)_{\rho^{\otimes m}}\\
    &+ H_\mi(b^{m_{l-2}-m_{l-1}}_{k+1}\cdots b^{m_{l-2}-m_{l-1}}_{l-2}|b^{m_{l-2}-m_{l-1}}_0\cdots b^{m_{l-2}-m_{l-1}}_k)_{\rho^{\otimes m}}+ H_\mi(b^{m_{l-1}}_{k+1}\cdots b^{m_{l-1}}_{l-1}|b^{m_{l-1}}_0\cdots b^{m_{l-1}}_k)_{\rho^{\otimes m}}\\
    =&\cdots\\
    =&\sum_{i=k}^{l-1} H_\mi(b^{m_{i}-m_{i+1}}_{k+1}\cdots b^{m_{i}-m_{i+1}}_{i}|b^{m_{i}-m_{i+1}}_0\cdots b^{m_{i}-m_{i+1}}_k)_{\rho^{\otimes m}}>\sum_{i=k}^{l-1} (m_{i}-m_{i+1})\frac{A_k-A_{i}}{4G_N}\,
\end{aligned}
\end{equation}
where we applied \eqref{generalqes} in the last step.\\

Therefore, the term with $b^{m_0}_0\cdots b^{m_k}_k$ dominates over $b^{m_0}_0\cdots b^{m_{l}}_{l}$ at large enough $n$,
\begin{equation}\label{dominance2}
\begin{aligned}
    &2^{(1-n)\sum_{i=0}^{l}(m_i - m_{i+1})\frac{A_i}{4G_N}}\tr\left(\rho^*_{b^{m_0}_0\cdots b^{m_{l}}_{l}}\right)^{n}\\
    <&2^{(1-n)\sum_{i=k}^{l-1} (m_{i}-m_{i+1})\frac{A_k-A_{i}}{4G_N}+(1-n)\sum_{i=0}^{l-1}(m_i - m_{i+1})\frac{A_i}{4G_N}}\tr\left(\rho^*_{b^{m_0}_0\cdots b^{m_k}_k}\right)^n \\
    =&2^{(1-n)\sum_{i=0}^{k-1}(m_i - m_{i+1})\frac{A_i}{4G_N} + (1-n)m_k\frac{A_k}{4G_N}}\tr\left(\rho^*_{b^{m_0}_0\cdots b^{m_k}_k}\right)^n \,
\end{aligned}
\end{equation}

With \eqref{dominance2} we've shown that in the bulk state $\rho^*_{b^m_0\cdots b^m_l}$, any term in \eqref{terms2} with arbitrary array $\{m_i\}_{i=0}^l$ is dominated by the terms with the same array on $\{m_i\}_{i=0}^k$ but with $\{m_i\}_{i=k+1}^l=0$, which is then dominated by the desired term with $\{m_i\}_{i=0}^k=m$ and $\{m_i\}_{i=k+1}^l=0$ \eqref{dominance1}. Thence, we have established that \eqref{keyterm} is the dominant contribution in the partition function \eqref{partitionfunction2} under the conditions \eqref{generalqes}. This term corresponds to the generalized entropy evaluated on $\gamma_k$ and $\be_k$. Since we choose that $\rho^*_{b^m_0\cdots b^m_k}$ is the maximizer for the smooth min-entropy \eqref{maximizer}, then via the AEP \eqref{minAEP} we immediately have as in \eqref{bulkaep1} and \eqref{bulkaep2}.
\begin{equation}
    S(B)\geq \frac{A_k}{4G_N}+H(\be_k)_\rho\,.
\end{equation}

Lastly, as shown by AP, the generalize entropy on $\gamma_k$ is the minimal among all QES' when \eqref{generalqes} holds. Therefore, the upper bound \eqref{generalupperbound} is simply
\begin{equation}
    S(B)\leq \min_i\{A_i/G_N+ H(b_i)_\rho\}=\frac{A_k}{4G_N}+H(\be_k)_\rho\,.
\end{equation}
Lastly, we can smoothen the conditional min/max-entropies as we did in the last section. This concludes the proof of the refined QES prescription for fixed-areas states with multiple QES candidates.

\section{Discussion}\label{sec:discussion}

By adopting a novel AEP replica trick, we derive the refined QES prescription via the gravitational path integral. Our derivation manifests the asymptotic nature of the holographic entanglement entropy and demonstrates the essential roles played by the conditional min/max-entropies in characterizing the QES phase transition. Let us also recall that AP's argument for \eqref{eq:refined_QES} is based on the equivalence between the Haar averaging calculations in the random tensor network model and the standard replica trick calculation~\cite{hayden2016holographic,hayden2019learning}. Then they are able to apply the one-shot decoupling theorem to the tensor network and it brings in the conditional min/max-entropies. In comparison, our approach is more direct as we do not translate the calculation to random tensor network models. \JZ{Therefore, it is also more transparent that the refinement should extend to the island formula and the Page curve, as we have demonstrated.} We also provide a full derivation for the general multi-QES setting and establish the validity of the most general refined QES prescription as proposed by AP. 

It's worth noting that lifting the assumption of the $\mathbb{Z}_n$ replica symmetry as in LM's original derivation was the key to observing the QES refinement in AP's calculation. We, however, derive the refined QES prescription while still holding the replica symmetry assumption and discarding the replica asymmetric saddles nonetheless. At least for pure bulk marginal states, as shown in section~\ref{sec:holorenyi}, this is because the higher R\'enyi entropies are not so sensitive to the replica symmetry breaking as the replica-symmetric saddles provide the dominant contributions. Therefore, the AEP replica trick  circumvents the complication of directly dealing with the replica symmetry breaking. 

On the other hand, the refinement in \eqref{eq:refined_QES} is more than the claim that there exists an indefinite regime within a window determined by the conditional min/max-entropies. It also says that in this regime the entropy can indeed deviate largely from the naive QES predictions. However, the behaviour of $S(B)$ in regime 2 remains unclear and our derivation does not add anything to that. While not necessary in manifesting the relevance of the one-shot entropies in the QES prescription, lifting the replica symmetry assumption, in our point of view, is still the key to observing large corrections in the indefinite regime. \\

Let us finish by discussing some interesting future research directions:\\

\emph{max-AEP replica trick.} The smooth max-entropy also satisfies the AEP, and we can in principle perform the same AEP replica trick using the max-AEP. The difference is that the minimization, instead of the maximization, shall gives us the upper bound instead of the lower bound. The upper bound would be useful, for instance, if we want to argue that the leading order correction is generic in the indefinite regime, because the corrected entropy can only be lower than the naive QES predictions. We don't really need another set of upper bounds because the ones deduced from the standard replica trick is enough for our derivation. However, it is of independent interest to understand how the max-AEP replica trick can work as it might gives us some useful upper bounds when it comes to the indefinite regime.  One complication that makes the max-AEP  tricky to deal with is that eventually we need to perform analytic continuation to $\alpha=\frac12$, whereas for the min-AEP there is no need for analytic continuation. This is more subtle than the common $\alpha\rightarrow 1$ continuation because for $\alpha<1$, we have the sign flipped in the exponents in the partition function and we need to deal with the brane action with negative tension. How such procedure can be justified is beyond the scope of this work but we provide the calculation using the max-AEP replica trick in the Appendix~\ref{sec:maxchainrule} assuming it would actually work. Nonetheless, it turns out that eventually we obtain exactly the same upper bound as in \eqref{upperbound}. \\

\emph{$m$-replica wormholes.} In our construction, the geometries of the $m$ families of replicas are uncorrelated whereas the supported bulk state is entangled. Given a smoothing state should be $\eps$ close to a factorized state and the lack of the cyclically identifying boundary condition, the replica-wormhole geometries among the $m$ families of replicas are suppressed. We therefore pick a convenient feasible state to be such that the correlation necessary for the smoothing is only built through the bulk fields. One interesting alternative is to consider a feasible state with a connected dual geometry as opposed to $\rho^{\otimes m}$. One can consider wormhole spacetimes where the $m$ families of replicas are connected by tiny and narrow wormholes to ensure the $\eps$ distance. Since the optimizer for the smooth min-entropy is permutation-invariant, we can restrict to $m$-replica-symmetric configurations. It is interesting to find out the bounds obtained from the partition function evaluated on these saddles, and generally explore the physics of these geometries.\\

\emph{Entanglement wedge reconstruction.} Harlow has proved that the validity of the quantum Ryu-Takayanagi formula is equivalent to the achievability of the task of entanglement wedge reconstruction with respect to the RT surface~\cite{harlow2017ryu}. Harlow's structural theorem is stated for the whole code subspace, but one can generalize this task of EWR to a state-dependent task~\cite{hayden2019learning,akers2021leading}, where the boundary operator is only required to have the same action as the bulk operator on the particular state. This is referred as state-specific EWR in AP. 

In the Schr\"odinger picture, AP argued that the state-specific EWR should be understood as one-shot state merging and it's characterized by the conditional max-entropy, which matches well with the conditions in the refined QES prescription. It's therefore promising that an analogue of Harlow's theorem can be established formally connecting state-specific EWR and the refined QES prescrition at least for the max-EW and the complement of min-EW. Regarding the indefinite regime and the corresponding ``no man's land'' region ($\mathrm{EW}_\mi\setminus\mathrm{EW}_\ma$), it's important to note that the state-specific EWR could still be achievable in this region because accessing side information from the ``no man's land'' could assist the state transfer (or operator reconstruction) in the $\mathrm{EW}_\ma$, provided this side information is sufficiently ``quantum'', i.e. has sufficient entanglement with $\mathrm{EW}_\ma$. It's less clear how Harlow's theorem can be extended to this regime. 

One potential direction is to consider a different setup that does not involve infinite replicas we considered in this work. The tensor network model and the error correction picture for holography suggest that we can think of the EWR as pushing a bulk operator through layers of the tensors to the boundary. In the Schrodinger picture, this corresponds to generate a bulk state supported over many bulk legs with an input state from the boundary. The conditional min/max-entropies of the bulk state generated as so can be estimated using a novel technique in entropy-calculus called the entropy accumulation theorem (EAT)\cite{dupuis2020entropy,dupuis2019entropy}, which is a generalisation of the AEP. The EAT claims that even though in situations one does not have the i.i.d product structure as in the AEP replica trick, the sequential structure still guarantees that the output min/max-entropies can be estimated as the accumulation of von Neumann entropies on each site. We believe the EAT can help generalize Harlow's theorem beyond the conventional i.i.d. regime to probe the no man's land. \\

\emph{Entanglement wedge stratification.} Following the definitions of $\mathrm{EW}_\ma$ and $\mathrm{EW}_\mi$, it is suggestive that we can also define $\mathrm{EW}_\alpha$ by replacing $H_\ma$ with $H_\alpha$ in \eqref{maxew} when $\alpha<1$, and $H_\mi$ with $H_\alpha$ in \eqref{minew} when $\alpha>1$ respectively. It begs the question if $\mathrm{EW}_\alpha$ is continuously deforming with $\alpha$? If so, it then implies that the ``no man's land'' region is stratified by the boundaries of $\mathrm{EW}_\alpha$'s. Such constructions could be useful in probing the indefinite regime.

\acknowledgments
I would like to thank Chris Akers and Geoff Penington for their feedback on the draft of this work and also for explaining their work to me. I am grateful to Renato Renner for some valuable suggestions and comments. I thank Joe Renes for pointing out the SDP formulation of the smooth min-entropy.  I thank Thomas Faulkner for pointing out a mistake in the proof of Theorem~\ref{thm:dominance} in a previous version. I also thank Victor Gitton and David Sutter for discussions. This work is supported by the National Center of Competence in Research ``SwissMAP'' as well as the Swiss National Science Foundation (SNSF) grant No.200021-188541.

\appendix
\section{The chain rules}\label{sec:proof}

We give the explicit constructions of the desired bulk states that satisfy various chain rules used in the paper. The appendix is divided into two parts dealing with the two-QES setting and multi-QES setting respectively.\\

The following lemma is an important property of the purified distance $P(\cdot , \cdot)$~\cite{tomamichel2010duality,dupuis2014one}. 
\begin{customlemma}{1}\label{lem:lemma2} 
Let $\rho_{AB},\sigma_B$ be two density operators. Then there exists some linear operator $T_B$ with $ \sigma_{AB}= T_B\rho_{AB}T_B $
an extension of $\sigma_B$ such that $P(\rho_B,\sigma_B)=P(\rho_{AB},\sigma_{AB})$. Furthermore, if $\rho_{AB}$ is pure, then there exists a purification $\sigma_{AB}$ of $\sigma_B$ with $P(\rho_B,\sigma_B)=P(\rho_{AB},\sigma_{AB})$.
\end{customlemma}
The above lemmas hold essentially because the purified distance is a fidelity-based metric and we have Uhlmann's theorem to dilate the distance to the purification space. This is the advantage of using the purified distance over the trace distance: we can always find extensions and purifications without increasing the distance.\\

We also need the following lemma to establish the chain rule:
\begin{customlemma}{2}[Lemma 21 in~\cite{tomamichel2011leftover}]\label{lem:lemma1}
Let $\eps>0$ and $\rho_{ABC}$ be some pure state. Then there exists a projector $\Pi_{AC}$ on $\mathcal{H}_{AC}$ and a state $\rrho_{ABC} = \Pi_{AC}\rho_{ABC}\Pi_{AC}$ such that $P(\rrho_{ABC},\rho_{ABC})\leq\eps$ and 
\begin{equation}\label{lemma1}
    -D_\infty (\rrho_{AB}|| I_A\otimes\rho_B)\geq H_\mi (A|B)_\rho - \log \frac{2}{\eps^2}\,.\\
\end{equation}
\end{customlemma}

\subsection{The two-QES setting}\label{sec:proof1}
\subsubsection{Proof of the chain rule \eqref{chainrule11}}

We shall construct two different bulk states $\rho'_{b^mb'^m\b^m}$ and $\rho''_{b^mb'^m\b^m}$, and each obeys a chain rule relating the bulk min-entropies. The argument is based on Lemma A.8 in~\cite{dupuis2014one} and also~\cite{vitanov2013chain}.\\

For each instance of the $n$ replicas, we are given the state $\rho_{bb'\b}^{\otimes m}$ supported on $m$ replicas. Let the $\rrho'_{b^m}\in\mathcal{B}^{\eps'}(\rho_b^{\otimes m})$ such that $H^{\eps'}_\mi(b^m)_{\rho^{\otimes m}}=H_\mi(b^m)_{\rrho'}=-\log\lambda$, that is, $\lambda$ is the minimal value such that $\lambda I_{b^m}\geq \rrho'_{b^m}.$

Lemma~\ref{lem:lemma1} says we can always find the extension of $\rrho'_{b^m}$. 
\begin{equation}
    \rrho'_{b^mb'^m\b^m} = (\Pi_{b'^m\b^m}\otimes T_{b^m}) \rho^{\otimes m}_{bb'\b} (\Pi_{b'^m\b^m}\otimes T_{b^m})
\end{equation}
that is $\eps''$-close to $T_{b^m}\rho^{\otimes m}_{bb'\b}T_{b^m}$, and we choose $T_{b^m}$ according to Lemma~\ref{lem:lemma2} such that 
\begin{equation}
P(T_{b^m}\rho^{\otimes m}_{bb'\b}T_{b^m},\rho^{\otimes m}_{bb'\b})=P(\rrho'_{b^m},\rho^{\otimes m}_{b})=\eps',
\end{equation}
where the exact form of $T_{b^m}$ is irrelevant for us. The triangle inequality then implies
\begin{equation}
P(\rrho'_{b^mb'^m\b^m},\rho^{\otimes m}_{bb'\b})\leq P(T_{b^m}\rho^{\otimes m}_{bb'\b}T_{b^m},\rrho'_{b^mb'^m\b^m})+P(T_{b^m}\rho^{\otimes m}_{bb'\b}T_{b^m},\rho^{\otimes m}_{bb'\b})\leq\eps''+\eps'.
\end{equation}

The extension satisfies \eqref{lemma1},
\begin{equation}\label{middlestep}
    2^{-H_\mi(b'^{m-k}|b^m)_{\rho^{\otimes m}}}\cdot I_{b'^{m-k}}\otimes \rrho'_{b^m}\geq\rrho'_{b'^{m-k}b^m}
\end{equation}
\begin{equation}
    2^{-H_\mi(b'^{m-k}|b^m)_{\rho^{\otimes m}}-H^{\eps'}_\mi(b^m)_{\rho^{\otimes m}}}\cdot I_{b'^{m-k}b^m}\geq\rrho'_{b'^{m-k}b^m}
\end{equation}
where we again drop the small $\log(\eps)$ term. Note that we can identify any $b'^{m-k}$ systems to be the system $A$ as in Lemma~\ref{lem:lemma1}, and the rest $b'^k\b^m$ to be the system $C$. One might worry that the projector $\Pi_{b'^m\b^m}$ differs when one chooses different $A-C$ bipartitions. Fortunately, this is not the case. The proof of Lemma~\ref{lem:lemma1} goes by constructing the dual projector $\Pi_{b^m}$ explicitly and it only depends on the marginals on $b^m$. The projector $\Pi_{b^m}$ dual to $P_{b'^m\b^m}$ is defined via
\begin{equation}
    \Pi_{b'^m\b^m}\otimes \left(\rho^{\otimes m}_{b}\right)^{-\frac12} \ket{\rho}^{\otimes m}_{bb'\b} = \left(\rho_{b'\b}^{\otimes m}\right)^{-\frac12}\otimes \Pi_{b^m} \ket{\rho}^{\otimes m}_{bb'\b}
\end{equation}
where $\ket{\rho}_{bb'\b}$ is the state vector for the pure state $\rho_{bb'\b}$. Hence the projector $\Pi_{b'^m\b^m}$ remains as long as $\rho_{b'\b}$ and $\rho_{b}$ are unchanged. (cf. the proof of Lemma 21 in~\cite{tomamichel2011leftover} for details.)

Therefore, the last inequality implies that there exists a bulk state $\rrho'_{b^mb'^m\b^m}$, which is within the $\eps''$-ball of $\rho_{bb'\b}^{\otimes m}$ with some $\eps''$,  such that for any $k$ out of $m$ replicas and some $0<\eps'<\eps''$,
\begin{equation}\label{appchain11}
    H_\mi(b^mb'^{m-k})_{\rrho'}\geq H_\mi(b'^{m-k}|b^m)_{\rho^{\otimes m}} + H^{\eps'}_\mi(b^m)_{\rho^{\otimes m}}=H_\mi(b'^{m-k}|b^m)_{\rho^{\otimes m}} + H_\mi(b^m)_{\rrho'}.
\end{equation}
Note that the state $\rrho'_{b^mb'^m\b^m}$ we constructed using the projection $\Pi_{b'^m\b^m}$ and $T_{b^m}$ is generally subnormalized, but this is still a legit state to use within the $\eps''$-Ball.

The above statement is almost the chain rule statement we want and is actually enough for our purposes. However, note that $\rrho'_{b^mb'^m\b^m}$ is not promised to be a pure state and for the second chain rule and later applications, it is convenient to have the same statement for a \emph{pure} bulk state.

Since \eqref{appchain11} only entails the marginal $\rrho'_{b^mb'^m}$ rather than the global state $\rrho'_{b^mb'^m\b^m}$, we can find a purification $\rho'_{b^mb'^m\b^m}$ of $\rrho'_{b^mb'^m}$, i.e. such that $\rho'_{b^mb'^m}=\rrho'_{b^mb'^m}$, on par with $\rho^{\otimes m}_{bb'\b}$ being a purification of $\rho^{\otimes m}_{bb'}$. Provided the Hilbert space $\b$ is large enough, this can always be done. Furthermore, Lemma~\ref{lem:lemma2} promises that 
\begin{equation}
    \eps:=P(\rho'_{b^mb'^m\b^m},\rho^{\otimes m}_{bb'\b})=P(\rho'_{b^mb'^m},\rho^{\otimes m}_{bb'})\leq P(\rrho'_{b^mb'^m\b^m},\rho^{\otimes m}_{bb'\b})=\eps'',
\end{equation}
where the inequality follows from the data processing inequality of the purified distance. It's also immediate that $\eps=P(\rho'_{b^mb'^m},\rho^{\otimes m}_{bb'})\geq P(\rho'_{b^m},\rho^{\otimes m}_{b})=\eps'$, so we can obtain the desired chain rule statement:

There exists a \emph{pure} bulk state $\rho'_{b^mb'^m\b^m}$, which is within the $\eps$-ball of $\rho_{bb'\b}^{\otimes m}$ with some $\eps$,  such that for any $k$ out of $m$ replicas and some $0<\eps'<\eps$,
\begin{equation}\label{appchain1}
    H_\mi(b^mb'^{m-k})_{\rho'}\geq H_\mi(b'^{m-k}|b^m)_{\rho^{\otimes m}} + H^{\eps'}_\mi(b^m)_{\rho^{\otimes m}}=H_\mi(b'^{m-k}|b^m)_{\rho^{\otimes m}} + H_\mi(b^m)_{\rho'}.
\end{equation}
Note that this chain rule only concerns the marginal state $\rho'_{b^mb'^m}$.\\

\textbf{Remark}: Before we move on to the second chain rule, we should remark that by tracing out some $\tilde{m}$ of $b^m$ systems in \eqref{middlestep}, we shall obtain
\begin{equation}
    H_\mi(b^{m-\tilde{m}}b'^{m-k})_{\rrho'}\geq H_\mi(b'^{m-k}|b^m)_{\rho^{\otimes m}} + H_\mi(b^{m-\tilde{m}})_{\rrho'},
\end{equation}
 One key difference is that now the marginal $\rrho'_{b^{m-\tilde{m}}}$ does not necessarily maximize the smooth min-entropy $H^{\tilde{\eps}}_\mi(b^{m-\tilde{m}})$. We can bring it closer to the form of \eqref{appchain1}. When $m_0\leq k$, then we have $H_\mi(b'^{m-k}|b^m)_{\rho^{\otimes m}}=H_\mi(b'^{m-k}|b^{m-\tilde{m}})_{\rho^{\otimes m}}$. Applying the same purification step eventually yields
\begin{equation}\label{remark1}
    H_\mi(b^{m-\tilde{m}}b'^{m-k})_{\rho'}\geq H_\mi(b'^{m-k}|b^{m-\tilde{m}})_{\rho^{\otimes m}} + H_\mi(b^{m-\tilde{m}})_{\rho'},
\end{equation}
which generalizes \eqref{appchain11}. We shall later see the applications of this chain rule. \\

\subsubsection{Proof of the chain rule \eqref{chainrule22}}
Having established \eqref{appchain1}, let's consider another pure state $\rho''_{b^mb'^m\b^m}$ with the same properties as $\rho'_{b^mb'^m\b^m}$, except that we swap $b$'s and $\b$'s. That is $\rho''_{b^mb'^m\b^m}$ satisfies for any $k$ out of $m$ replicas and some $0<\eps''<\eps$,
\begin{equation}
    H_\mi(\b^mb'^k)_{\rho''}\geq H_\mi(b'^k|\b^m)_{\rho^{\otimes m}} + H^{\eps''}_\mi(\b^m)_{\rho^{\otimes m}}=H_\mi(b'^k|\b^m)_{\rho^{\otimes m}} + H_\mi(\b^m)_{\rho''}
\end{equation}
where we also replace $k$ by $m-k$. Its existence follows from the same argument as above. We can take the dual of it by exploiting the purity of $\rho''_{b^mb'^m\b^m}, \rho_{bb'\b}^{\otimes m}$ and the duality relation \eqref{eq:duality},
\begin{equation}
\begin{aligned}
    &H_\mi(b'^{k}\b^m)_{\rho''}=H_\mi(b^mb'^{m-k})_{\rho''},\\
    &H_\mi(\b^m)_{\rho''}=H_\mi(b^mb'^m)_{\rho''},\\
    &H_\mi(b'^{k}|\b^m)_{\rho^{\otimes m}}=-H_\ma(b'^{k}|b'^{m-k}b^m)_{\rho^{\otimes m}}.
\end{aligned}
\end{equation}
Put them together,
\begin{equation}
    H_\mi(b^mb'^m)_{\rho''}\leq H_\ma(b'^k|b'^{m-k}b^m)_{\rho^{\otimes m}} + H_\mi(b^mb'^{m-k})_{\rho''}.
\end{equation}
Consider the marginal $\rho''_{\b^m}$ that satisfies $H_\mi(\b^m)_{\rho''}=H^{\eps''}_\mi(\b^m)_{\rho^{\otimes m}}$. The unconditional min-entropy is only a function of the spectrum, so the maximal $\rho''_{\b^m}$ differs from $\rho_{\b}^{\otimes m}$ only in their spectrum. Therefore, we know $P(\rho''_{\b^m},\rho_{\b}^{\otimes m})=\eps''$ and $[\rho''_{b^m},\rho_{\b}^{\otimes m}]=0$. Since the dual complementary state share the same spectrum as them, and both $\rho_{bb'\b}^{\otimes m}, \rho''_{b^mb'^m\b^m}$ share the same Schmidt basis with respect to the bipartition $b^m-b'^m\b^m$, we have
\begin{equation}
    P(\rho''_{b'^mb^m},\rho_{b'b}^{\otimes m})=P(\rho''_{\b^m},\rho_{\b}^{\otimes m})=\eps''.
\end{equation}
Furthermore, the identical spectrum also implies $\rho''_{b'^mb^m}$ maximizes $H^{\eps''}_\mi(b'^mb^m)_{\rho^{\otimes m}}$, just as how $\rho''_{\b^m}$ maximizes $H^{\eps''}_\mi(\b^m)_{\rho^{\otimes m}}$. Thus, we obtain the second chain rule:
\begin{equation}\label{appchain2}
    H_\mi(b^mb'^m)_{\rho''}=H^{\eps''}_\mi(b^mb'^m)_{\rho^{\otimes m}}\leq H_\ma(b'^k|b'^{m-k}b^m)_{\rho^{\otimes m}} + H_\mi(b^mb'^{m-k})_{\rho''}.
\end{equation}
Note that this chain rule only concerns the marginal state $\rho''_{b^mb'^m}$.

\subsubsection{$m$-replica symmetry}

We claim that the above constructed states $\rho'_{b^mb'^m\b^m}$, $\rho''_{b^mb'^m\b^m}$ are symmetric with respect to permuting the $m$ replicas. Recall that we perform several operations provided in Lemmas~\ref{lem:lemma2} \&\ref{lem:lemma1} on the maximizer $\rrho'_{b^m}$ to obtain the state $\rrho'_{b^mb'^m\b^m}$. These operations preserve the permutation symmetry among the $m$ families of replicas as they do not distinguish them by construction, provided $\rrho'_{b^m}$ is permutation-invariant. For example, the $T_{b^m}$ and $\Pi_{b'^m\b^m}$ operators are algebraically constructed from $\rrho'_{b^m}$ and $\rho^{\otimes m}_{bb'\b}$. Therefore, all we need to worry about is if the maximizer $\rrho'_{b^m}$ is permutation-invariant. Note that the maximizer for $H^{\eps'}_\mi(b^m)_{\rho^{\otimes m}}$ might not be unique, but the claim is we can always find a maximizer that preserves the symmetry of $\rho^{\otimes m}$. To see this, we need the following formulation of the smooth min-entropy via a semidefinite program, which is parameterized by $\rho_{bb'\b}^{\otimes m}$. $H^{\eps'}_\mi(b^m)_{\rho^{\otimes m}}$ is given by (cf. equation (6.37) in~\cite{tomamichel2015quantum})
\begin{equation}
\begin{aligned}
\max:& -\log\lambda\\
  \text{subject to}:&\, \lambda\in\mathbb{R},\;\rrho_{b^mb'^m\b^m}\in\s(b^mb'^m\b^m)\\
     &\, \lambda I_{b^m}\geq \tr_{b'^m\b^m} \rrho_{b^mb'^m\b^m}\\
       &\,\tr\left[\rrho_{b^mb'^m\b^m}\rho_{bb'\b}^{\otimes m}\right] \geq 1-\eps'^2
\end{aligned}
\end{equation}
Suppose that the maximum is realized at some $(\lambda^\star,\rrho^\star_{b^mb'^m\b^m})$, and $\rrho^\star_{b^mb'^m\b^m}$ is not necessarily $m$-replica-symmetric. Consider 
\begin{equation}
    \rrho^*_{b^mb'^m\b^m}:=\frac{1}{|S_m|}\sum_{i\in S_m} \pi_i \rrho^\star_{b^mb'^m\b^m} \pi^\dagger_i
\end{equation}
where we apply a random permutation to the maximizer and obtain a permutation-invariant state. We can readily check that the constraints are invariant under the action of the random permutation so $\rrho^*_{b^mb'^m\b^m}$ is also a maximizer with the same optimal value $\lambda^\star$. In particular, its marginal $\rrho^*_{b^m}$ is replica-symmetric, and we denote it as $\rrho'_{b^m}$. This shows that we can assume without loss of generality that the maximizer is always $m$-replica-symmetric. 

Similarly, in the last purification step we take to obtain $\rho'_{b^mb'^m\b^m}$ from $\rrho'_{b^mb'^m}$, it's Uhlmann's theorem which allows us to find the purification $\rho'_{b^mb'^m\b^m}$ with the same purified distance to $\rho^{\otimes m}_{bb'\b}$ as between the marginals. That is $F(\rho^{\otimes m}_{bb'},\rrho'_{b^mb'^m})=|\braket{\rho^{\otimes m}_{bb'b}}{\rho'_{b^mb'^m\b^m}}|^2$. It's therefore clear that if we use a coherent superposition over random permutations 
\begin{equation}
    \ket{\rho'_{b^mb'^m\b^m}}\rightarrow \frac{1}{\sqrt{|S_m|}}\sum_{i\in S_m} \pi_i \ket{\rho'_{b^mb'^m\b^m}}\,,
\end{equation}
the permuted state is as good as the pre-processed one in fidelity. Therefore, the purification $\rho'_{b^mb'^m\b^m}$ can be set to be replica-symmetric without loss of generality. The same argument also applies to $\rho''_{b^mb'^m\b^m}$.

\subsection{The multi-QES setting}\label{sec:proof2}
Since the condition \eqref{generalqes} is basically a collection of the various inequalities comparing min/max-entropies with the area differences, the same construction used in the previous section to give the desired chain rules \eqref{appchain1}\& \eqref{appchain2} can be recycled to give us the desired state $\rho^*_{b^{m_0}_0\cdots b^{m_k}_k}$. However, one complication is that we need to enforce two chain rules on this state, so we better make sure the constructions that lead to the chain rules are compatible with each other. Fortunately for us, this can indeed be achieved. We also remark that the desired state is not guaranteed to be pure as in the previous constructions, and we do not need the purity here anyway. 

Let us start with a pure bulk state $\rho'_{b^m_0\cdots b^m_l}$ that satisfies the multi-regions counterpart of \eqref{appchain1}. According to Lemma~\ref{lem:lemma1}, for some arbitrary $\eps'$, the state $\rho'_{b^m_0\cdots b^m_l}$, which is $\eps'$-close to $\rho^{\otimes m}_{b_0\cdots b_l}$, can be constructed using the projector $\Pi_{b^m_{k+1}\cdots b^m_l}$,
\begin{equation}
    \rho'_{b^m_0\cdots b^m_l} = I_{b^m_{0}\cdots b^m_k}\otimes\Pi_{b^m_{k+1}\cdots b^m_l}\, \rho^{\otimes m}_{b_0\cdots b_l}\, I_{b^m_{0}\cdots b^m_k}\otimes\Pi_{b^m_{k+1}\cdots b^m_l}\,,
\end{equation}
so the marginal state $\rho'_{b^m_0\cdots b^m_k}$ is $\rho^{\otimes m}_{b_0\cdots b_k}$, and it satisfies the operator inequality for arbitrary $\{m_i\}_{i=k+1}^l$ with order $m_{k+1}\geq\cdots\geq m_l=0$.
 \begin{equation}\label{opinequality1}
     2^{-H_\mi(b^{m_{k+1}}_{k+1}\cdots b^{m_{l}}_{l}|b^m_0\cdots b^m_k)_{\rho^{\otimes m}}}\cdot I_{b^{m_{k+1}}_{k+1}\cdots b^{m_{l}}_{l}}\otimes \rho'_{b^m_0\cdots b^m_k}\geq \rho'_{b^m_0\cdots b^m_k b^{m_{k+1}}_{k+1}\cdots b^{m_l}_l}\,.
 \end{equation}

Consider now another pure bulk state $\rho''_{b^m_0\cdots b^m_l}$ that satisfies the multi-regions counterpart of \eqref{appchain2} (that is \eqref{superchain1}) for $0<\eps''<1$ and any $\{m_i\}_{i=0}^k$ with $m_k\leq\cdots\leq m_0=m$: 
\begin{equation}
 H_\mi(b^m_0\cdots b^m_k)_{\rho''} \leq H_\ma(b^{m-m_0}_0\cdots b^{m-m_k}_k|b^{m_0}_0\cdots b^{m_k}_k)_{\rho^{\otimes m}}+  H_\mi(b^{m_0}_0\cdots b^{m_k}_k)_{\rho''}\,,
\end{equation}
and 
\begin{equation}\label{smoothmaximizer}
    H^{\eps''}_\mi(b^m_0\cdots b^m_k)_{\rho^{\otimes m}}=H_\mi(b^m_0\cdots b^m_k)_{\rho''}\,.
\end{equation}

Since this chain rule is a property of the marginal state $\rho''_{b^m_0\cdots b^m_k}$ only, we can make use of Lemma~\ref{lem:lemma1} again to construct the desired bulk state $\rho^*_{b^m_0\cdots b^m_l}$ from $\rho'_{b^m_0\cdots b^m_l}$ via $\rho''_{b^m_0\cdots b^m_k}$. The lemma says there exists an operator $T_{b^m_0\cdots b^m_k}$ such that an \emph{alternative} extension of $\rho''_{b^m_0\cdots b^m_k}$ preserves the purified distance. The exact form of $T_{b^m_0\cdots b^m_k}$ is irrelevant for us.  This extension is defined as
\begin{equation}
    \rho^*_{b^m_0\cdots b^m_l}=T_{b^m_0\cdots b^m_k}\,\rho'_{b^m_0\cdots b^m_l}\,T_{b^m_0\cdots b^m_k}\,,
\end{equation}
which satisfies
\begin{equation}
    \rho''_{b^m_0\cdots b^m_k}=\rho^*_{b^m_0\cdots b^m_k}=T_{b^m_0\cdots b^m_k}\,\rho'_{b^m_0\cdots b^m_k}\,T_{b^m_0\cdots b^m_k}\,,
\end{equation}
and
\begin{equation}
    P( \rho'_{b^m_0\cdots b^m_l}, \rho^*_{b^m_0\cdots b^m_l}) = P(\rho'_{b^m_0\cdots b^m_k}, \rho^*_{b^m_0\cdots b^m_k})=P(\rho^{\otimes m}_{b_0\cdots b_k}, \rho''_{b^m_0\cdots b^m_k})\leq\eps''\,
\end{equation}
where the last inequality follows from \eqref{smoothmaximizer}. Recall that $P(\rho^{\otimes m}_{b_0\cdots b_l},\rho'_{b^m_0\cdots b^m_l})$ by construction via Lemma~\ref{lem:lemma2}. Then the triangle inequality implies
\begin{equation}
    P(\rho^{\otimes m}_{b_0\cdots b_l},\rho^*_{b^m_0\cdots b^m_l})\leq P(\rho^{\otimes m}_{b_0\cdots b_l},\rho'_{b^m_0\cdots b^m_l})+P(\rho'_{b^m_0\cdots b^m_l},\rho^*_{b^m_0\cdots b^m_l})\leq \eps'+\eps''=:\eps\,.
\end{equation}
Therefore, the desired bulk state $\rho^*_{b^m_0\cdots b^m_l}$ is a feasible choice, i.e. it is inside the ball $\mathcal{B}_\eps (\rho^{\otimes m}_{b_0\cdots b_l})$.\\

 The key point now is that since the domain of the operator $T_{b^m_0\cdots b^m_k}$ is orthogonal to the domain of the projector $\Pi_{b^m_{k+1}\cdots b^m_l}$, the bulk state $\rho^*_{b^m_0\cdots b^m_l}$ so defined still obeys the chain rule above concerning $\rho''_{b^m_0\cdots b^m_k}$ only. Also, we thus can apply $T_{b^m_0\cdots b^m_k}$ to both sides of \eqref{opinequality1},
\begin{equation}
    2^{-H_\mi(b^{m_{k+1}}_{k+1}\cdots b^{m_{l}}_{l}|b^m_0\cdots b^m_k)_{\rho^{\otimes m}}}\cdot I_{b^{m_{k+1}}_{k+1}\cdots b^{m_{l}}_{l}}\otimes \rho^*_{b^m_0\cdots b^m_k}\geq \rho^*_{b^m_0\cdots b^m_k b^{m_{k+1}}_{k+1}\cdots b^{m_l}_l}\,.
\end{equation}
As remarked in \eqref{remark1} above, we can trace out some replicas of the $b^m_0\cdots b^m_k$ systems to obtain $b^{m_0}_0\cdots b^{m_k}_k$:

\begin{equation}
    2^{-H_\mi(b^{m_{k+1}}_{k+1}\cdots b^{m_{l}}_{l}|b^m_0\cdots b^m_k)_{\rho^{\otimes m}}}\cdot I_{b^{m_{k+1}}_{k+1}\cdots b^{m_{l}}_{l}}\otimes \rho^*_{b^{m_0}_0\cdots b^{m_k}_k}\geq \rho^*_{b^{m_0}_0\cdots b^{m_k}_k b^{m_{k+1}}_{k+1}\cdots b^{m_l}_l}\,.
\end{equation}

Since $\rho^*_{b^{m_0}_0\cdots b^{m_k}_k}=\rho''_{b^{m_0}_0\cdots b^{m_k}_k}$ maximizes the smooth min-entropy \eqref{smoothmaximizer}, we have
\begin{equation}
    2^{-H_\mi(b^{m_{k+1}}_{k+1}\cdots b^{m_{l}}_{l}|b^m_0\cdots b^m_k)_{\rho^{\otimes m}}-H^{\eps''}_\mi(b^m_0\cdots b^m_k)_{\rho^{\otimes m}}}\cdot I_{b^{m_0}_0\cdots b^{m_l}_l}\geq \rho^*_{b^{m_0}_0\cdots b^{m_l}_l}\,.
\end{equation}

It then implies the chain rule \eqref{superchain2}
\begin{equation}
\begin{aligned}
    H_\mi(b^{m_0}_0\cdots b^{m_{l}}_{l})_{\rho^*} \geq H_\mi(b^{m_{k+1}}_{k+1}\cdots b^{m_{l}}_{l}|b^{m_0}_0\cdots b^{m_k}_k)_{\rho^{\otimes m}}+  H_\mi(b^{m_0}_0\cdots b^{m_k}_k)_{\rho^*}\,.
\end{aligned}
\end{equation}

This bulk state $\rho^*_{b^m_0\cdots b^m_l}$ thus possesses both sets of the chain rules \eqref{superchain1}\&\eqref{superchain2} we want to match with \eqref{generalqes}\,.

\section{The max-entropy AEP replica trick}\label{sec:maxchainrule}

Similarly to \eqref{minAEP}, the max-entropy AEP shall give us an upper bound with some feasible state $\rho'$.
  \begin{equation}
           S(B)_\rho := \lim_{m\rightarrow\infty}\lim_{n\rightarrow\frac12}  \frac{1}{m(1-n)} \min_{\tilde{\rho}_B^m\approx_\varepsilon\rho_B^{\otimes m}} \log\tr\left[\tilde{\rho}_{B}^m\right]^n \leq \lim_{m\rightarrow\infty}\lim_{n\rightarrow\frac12} \frac{1}{m(1-n)} \log\frac{Z_{n,m}}{Z_1^{nm}}\,.
 \end{equation}
Note that we need to analytically continue $Z_{n,m}$ to $Z_{\frac12,m}$.  We're gonna consider the other two chain rules relating the smooth max-entropies. Again, it amounts to make a good choice of the feasible state for any integer $m$ to obtain useful bounds. 

Let us look at the chain rules that relate the max-entropies, which are structurally identical to the chain rules \eqref{chainrule1}\& \eqref{chainrule2} we've seen for the min-entropies. We can choose a pure bulk state $\rho'_{b^mb'^m\b^m}$, which is within the $\eps$-ball of $\rho_{bb'\b}^{\otimes m}$, such that for any $k$ out of $m$ replicas, and some $\eps'<\eps$,
 \begin{equation}\label{appchain3}
    H_\ma(b'^mb^m)_{\rho'}= H^{\eps'}_\ma(b'^mb^m)_{\rho^{\otimes m}}\leq H_\ma(b'^k|b^mb'^{m-k})_{\rho^{\otimes m}}+H_\ma(b^mb'^{m-k})_{\rho'}\,.
\end{equation}
As before, we have 
\begin{equation}
     H_\ma(b'^{m-k}|b^mb'^{m-k})_{\rho^{\otimes m}} = kH_\ma(b'|b)_\rho\,.
\end{equation}
The condition $H_\mi(b'|b)_\rho>\frac{A_2-A_1}{4G_N}$ then implies 
\begin{equation}
    H_\ma(b'^mb^m)_{\rho'}<k\frac{A_2-A_1}{4G_N}+H_\ma(b^mb'^{m-k})_{\rho'}\,,
\end{equation}
which is equivalent to
\begin{equation}
    \tr\sqrt{\rho'_{b'^mb^m}}<2^{\frac{k}{2}\frac{A_2-A_1}{4G_N}}\tr \sqrt{\rho'_{b^mb'^{m-k}}}\,,
\end{equation}
that is
\begin{equation}
    2^{\frac{m}{2}\frac{A_1}{4G_N}}\tr\sqrt{\rho'_{b'^mb^m}}<2^{\frac{k}{2}\frac{A_2}{4G_N}+\frac{m-k}{2}\frac{A_1}{4G_N}}\tr \sqrt{\rho'_{b^mb'^{m-k}}}\,.
\end{equation}
This is the same statement as \eqref{A1dominant} at $n=\frac12$ except with the opposite inequality. However, note that the factor $(1-n)$ also flips the sign as $n$ goes from positive integers to $\frac12$. This means the smaller LHS above shall be the dominant term in the action. It matches with our claim \eqref{claim2} for the chosen bulk state. 

The other chain rule follows from applying \eqref{appchain3} to $\b$'s and the duality relation \eqref{eq:duality}. Following the same arguments used in establishing \eqref{appchain2}, we have for some state $\rho''_{b^mb'^m\b^m}$, any $k$ out of $m$ replicas, and some $\eps''<\eps$,
\begin{equation}\label{chainrule44}
    H_\ma(b^mb'^{m-k})_{\rho''}\geq H_\mi(b'^{m-k}|b^m)_{\rho^{\otimes m}}+H_\ma(b^m)_{\rho''} = H_\mi(b'^{m-k}|b^m)_{\rho^{\otimes m}}+H^{\eps''}_\ma(b^m)_{\rho^{\otimes m}}\,.
\end{equation}
When $H_\mi(b'|b)_\rho>\frac{A_2-A_1}{4G_N}$, the same steps as before lead us to
\begin{equation}
    H_\ma(b^mb'^{m-k})_{\rho'}>(m-k)\frac{A_2-A_1}{4G_N}+H_\ma(b^m)_{\rho''}\,,
\end{equation}
and thus
\begin{equation}
    2^{\frac{m-k}{2}\frac{A_1}{4G_N}+\frac{k}{2}\frac{A_2}{4G_N}}\tr\sqrt{\rho''_{b^mb'^{m-k}}}>2^{\frac{m}{2}\frac{A_2}{4G_N}}\tr \sqrt{\rho''_{b^m}}\,.
\end{equation}
Again, the RHS shall be the dominant term for this chosen bulk state as in \eqref{claim1}.

Since the bulk states $\rho'_{b^mb'^m\b^m}$ and $\rho''_{b^mb'^m\b^m}$ are chosen to minimize the bulk smooth max-entropies for the $b^mb'^m$ and $b^m$ regions respectively, we can apply the AEP on the bulk states.  When $H_\ma(b'|b)_\rho<\frac{A_2-A_1}{4G_N}$,
\begin{equation}
    \begin{aligned}
    S(B)_\rho &\leq \lim_{m\rightarrow\infty}\lim_{n\rightarrow\frac12}\frac{1}{m(1-n)}\log\frac{Z_{n,m}[A_{1,2}]}{Z_1[A_{1,2}]^{nm}} \\
    &= \frac{A_1}{4G_N} + \lim_{m\rightarrow\infty}\lim_{n\rightarrow\frac12}\frac{1}{m(1-n)}\tr(\rho^{\otimes n}_{b^mb'^{m}}\cdot U_{\tau|b_1}\otimes\cdots\otimes U_{\tau|b_m}\otimes U_{\tau|b'_1}\otimes\cdots\otimes U_{\tau|b'_m})\\
    & = \frac{A_1}{4G_N} + \lim_{m\rightarrow\infty}\frac{1}{m}H_\mi^{\eps'}(b^mb'^m)_{\rho^{\otimes m}}\\
    & = \frac{A_1}{4G_N} + H(bb')_\rho\,.
\end{aligned}
\end{equation}

When $H_\mi(b'|b)_\rho>\frac{A_2-A_1}{4G_N}$,
\begin{equation}
  \begin{aligned}
    S(B)_\rho &\leq \lim_{m\rightarrow\infty}\lim_{n\rightarrow\frac12}\frac{1}{m(1-n)}\log\frac{Z_{n,m}[A_{1,2}]}{Z_1[A_{1,2}]^{nm}} \\
    &= \frac{A_2}{4G_N} + \lim_{m\rightarrow\infty}\lim_{n\rightarrow\frac12}\frac{1}{m(1-n)}\tr(\rho''^{\otimes n}_{b^m}\cdot U_{\tau|b_1}\otimes\cdots\otimes U_{\tau|b_m})\\
    & = \frac{A_2}{4G_N} + \lim_{m\rightarrow\infty}\frac{1}{m}H_\mi^{\eps''}(b^m)_{\rho^{\otimes m}}\\
    & = \frac{A_2}{4G_N} + H(b)_\rho\,.
\end{aligned}  
\end{equation}
These upper bounds are identical to \eqref{upperbound}. 

\section{The resolvent calculation}\label{sec:JTEOW}

We would like compute the trace of the resolvent matrix \eqref{eq:resolvent}, which is the Stieltjes transform of the spectral density function, using the expansion in $\lambda$. It involves objects like $\tr \rho_R^n$. Using \eqref{eq:generalspectrum}, they read,
\begin{equation}
    \tr \rho_R^n =\sum_{i_1, i_2,\cdots i_n}c^2_{i_1}\cdots c^2_{i_n} \braket{\psi_{i_1}}{\psi_{i_2}}\braket{\psi_{i_2}}{\psi_{i_3}}\cdots\braket{\psi_{i_n}}{\psi_{i_1}}.
\end{equation}
As we have mentioned, the inner products are no longer orthogonal when $n>1$, we need to compute them using the gravitational path integral. Let us introduce some diagrammatic rules following~\cite{penington2019replica}. The boundary condition for $\braket{\psi_i}{\psi_j}$ is depicted by  
\begin{equation}
    \includegraphics[width=0.3\linewidth]{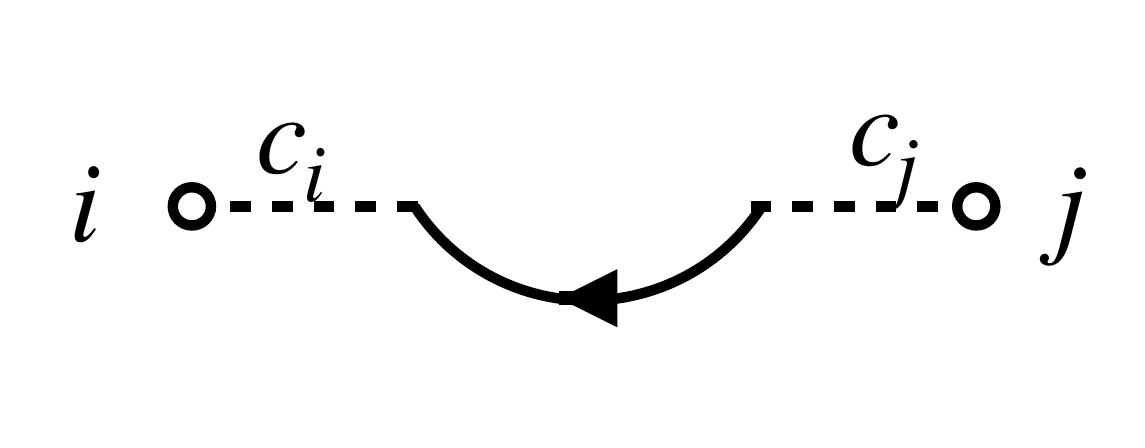}
\end{equation}
where the black line denotes the asymptotic boundary with the arrow pointing the time direction from the ket to the bra; and the boundary is joined with the EOW brane of favour index $i,j$ respectively, depicted as the dashed leg. For each EOW brane leg $i$, we associate with the weights $c_i$ prescribed by the Schmidt coefficient. For computing quantities like $\tr \rho_R^n$, we should contract the EOW brane indices and this gives us a closed boundary condition and the path integral instructs us to sum the on-shell action over all the possible topologies/geometries that can be filled in. We focus on the non-crossing geometries that dominate over the crossing ones as explained in~\cite{penington2019replica}. The rule is that for each connected bulk domain joining $l$ boundaries, one should include a contribution $Z_l/Z_1^l$, which is simply equal to $2^{(1-l)S}$ in the microcanonical ensemble of entropy $S$. Also, for each dashed index loop, we sum over all the weights with a delta function. 

We directly do the diagrammatic sum for the resolvent matrix that illustrates the rules.
\begin{equation}\label{eq:resexpansion}
    \includegraphics[width=0.9\linewidth]{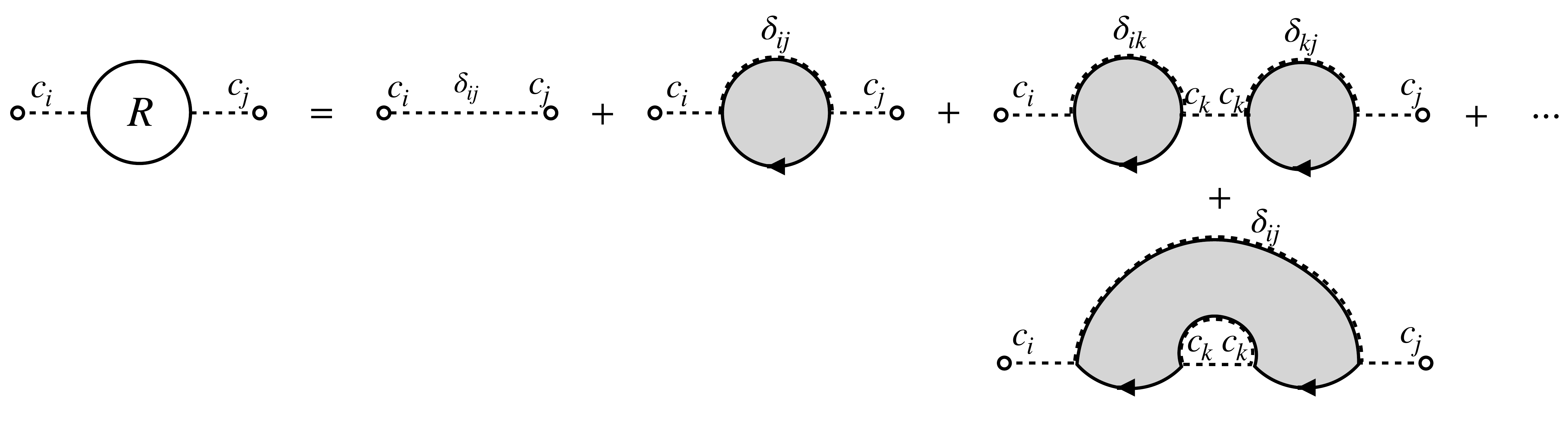}
\end{equation}
where the resolvent matrix is denoted on the left hand side, and the sum of the first few terms are depicted. They translate to 
\begin{equation}
    R_{ij}c_ic_j=\frac{c_ic_j\delta_{ij}}{\lambda}+\frac{c_ic_j\delta_{ij}}{\lambda^2}+\frac{c_i^2c_j^2\delta_{ij}+2^{-S}c_ic_j\delta_{ij}}{\lambda^3}+\cdots.
\end{equation}
where for each index line we need to add in another factor of $1/\lambda$, and the repeated indices in the middle (those not associated with the uncontracted ends) are summed over. We have also used that $\sum_k c_k^2=1$. Now we exploit the re-summation trick by rearranging the sum as follows,
\begin{equation}
    \includegraphics[width=1.0\linewidth]{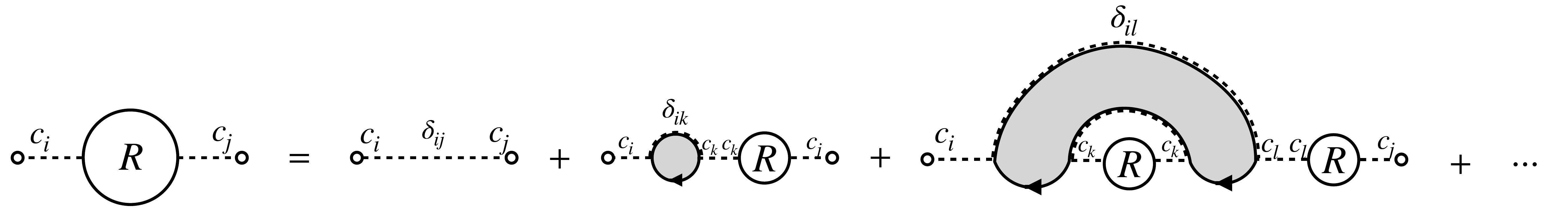}
\end{equation}
where the second term includes all the summands in \eqref{eq:resexpansion} with the first boundary contracted with itself, and the third terms includes all the summands in \eqref{eq:resexpansion} with the first boundary contracted with one other boundary, and so on. Note that in the third term, the contracted loop with $R_{ij}$ gives
\begin{equation}
    \sum_{i,j} R_{ij}c_ic_j\delta_{ij} =: R^{(1)}
\end{equation}
as we've defined in the main text. Also, the uncontracted dash line with $R_{ij}$ gives
\begin{equation}
    \sum_{l} c_lc_jR_{lj}c_ic_l\delta_{il} =: R_{ij}c_i^3c_j.
\end{equation}

The sum then reads,
\begin{equation}
    R_{ij}c_ic_j=\frac{c_ic_j\delta_{ij}}{\lambda}+\frac{R_{ij}c_i^3c_j}{\lambda}\sum_{l=1}^\infty 2^{(1-l)S}\left(R^{(1)}\right)^{l-1}.
\end{equation}

Finally we can compute the geometric series and obtain \eqref{eq:resum},
 \begin{equation}
     \lambda R_{ij} = \delta_{ij}+\frac{c_i^2R_{ij}}{1-R^{(1)}2^{-S}}\,,
 \end{equation}
from which we derive the coupled equations \eqref{eq:cubicequations} and $R(\lambda;\rho)$ can be analytically solved.

\bibliographystyle{jhep}
\bibliography{qes}

\end{document}